\documentclass[twoside,leqno]{article}

\usepackage[letterpaper]{geometry}

\usepackage{xcolor}
\definecolor{c0}{HTML}{641a80}
\definecolor{c1}{HTML}{b73779}
\definecolor{Green}{HTML}{009B55}
\definecolor{Red}{HTML}{AF3235}

\usepackage{amsmath,amssymb}

\usepackage{hyperref}
\hypersetup{
    colorlinks,
    linkcolor=c0,
    citecolor=c1,
    urlcolor=c0,
    hypertexnames=false
}

\usepackage[capitalize,nameinlink,noabbrev]{cleveref}
\crefformat{equation}{#2(#1)#3}
\crefmultiformat{equation}{{}#2(#1)#3}%
{ and~#2(#1)#3}{, #2(#1)#3}{ and~#2(#1)#3}

\usepackage{ltexpprt}

\usepackage{dsfont}

\usepackage{tikz}
\usetikzlibrary{calc,graphs,patterns} 
\usetikzlibrary{decorations.pathreplacing,calligraphy}

\usepackage{algorithm}
\usepackage{algorithmicx}
\usepackage[noend]{algpseudocode}
\algrenewcommand\algorithmiccomment[1]{\hfill{\color{gray}$\triangleright$~#1}}

\newenvironment{definition}{\begin{Definition}}{\end{Definition}}
\newtheorem{problem}{Problem}[section]

\newenvironment{restatethm}[1][]{\noindent{\textsc{#1}.}\begingroup\it}{\endgroup}

\usepackage{nicematrix}

\usepackage{bm}
\renewcommand{\vec}{\mathbf}

\newcommand{\Gaussian}{\operatorname{Gaussian}}

\newcommand{\HODLR}{\operatorname{HODLR}}

\newcommand{\OPT}{\mathsf{opt}}
\newcommand{\ERR}{\mathsf{err}}

\newcommand{\sleft}{s_{\textup{\tiny{L}}}}
\newcommand{\sright}{s_{\textup{\tiny{R}}}}
\newcommand{\tleft}{t_{\textup{\tiny{L}}}}
\newcommand{\tright}{t_{\textup{\tiny{R}}}}

\newcommand{\EE}{\operatorname{\mathbb{E}}}
\newcommand{\PP}{\operatorname{\mathbb{P}}}

\newcommand{\R}{\mathbb{R}}

\newcommand{\T}{\mathsf{T}}
\newcommand{\F}{\mathsf{F}}

\newcommand{\rank}{\operatorname{rank}}
\newcommand\numberthis{\addtocounter{equation}{1}\tag{\theequation}}

\newcommand{\argmin}{\operatornamewithlimits{argmin}}

\newcommand{\poly}{\mathop\mathrm{poly}}

\newcommand{\llbracket}{[\mkern-2.5mu[}
\newcommand{\rrbracket}{]\mkern-2.5mu]}

\newcommand{\jp}{j\pm 1}

\usepackage{booktabs}

\crefname{Definition}{Definition}{Definitions}
\crefname{assumption}{Assumption}{Assumptions}
\crefname{problem}{Problem}{Problems}
\crefname{theorem}{Theorem}{Theorems}
\Crefname{theorem}{Theorem}{Theorems}
\crefname{algorithm}{Algorithm}{Algorithms}

\usepackage[backend=biber,backref,style=alphabetic,maxalphanames=4,maxcitenames=99,maxbibnames=99,giveninits=true,url=false,doi=false,isbn=false,eprint=true,date=year]{biblatex}
\DeclareSourcemap{
  \maps[datatype=bibtex, overwrite]{
    \map{
      \step[fieldset=editor, null]
    }
  }
}
\DefineBibliographyStrings{english}{%
  backrefpage = {cited on page},%
  backrefpages = {cited on pages}%
}
\DeclareFieldFormat[article,misc,online,inproceedings,thesis,report,book,inbook]{title}{{#1}}

\makeatletter
\newenvironment{nproof}[1][]{\@beginproof{#1}}{\@endproof}
\makeatother

  \usepackage{nth}
  \usepackage{intcalc}

  \newcommand{\cAAAI}[1]{AAAI\ Conference\ on\ Artificial (AAAI)}

\begin{document}

\newcommand\relatedversion{}
\renewcommand\relatedversion{\thanks{The full version of the paper can be accessed at \protect\url{https://arxiv.org/abs/2407.04686}}} 

\title{\Large Near-optimal hierarchical matrix approximation from matrix-vector products\relatedversion}
\author{
\begin{tabular}{ccc}
Tyler Chen\thanks{New York University (\texttt{tyler.chen@nyu.edu})}
& Feyza Duman Keles\thanks{New York University (\texttt{fd2135@nyu.edu})}
& Diana Halikias\thanks{Cornell University (\texttt{dh736@cornell.edu})}
\\
Cameron Musco\thanks{UMass Amherst (\texttt{cmusco@cs.umass.edu})}
& Christopher Musco\thanks{New York University (\texttt{cmusco@nyu.edu})}
& David Persson \thanks{EPFL (\texttt{david.persson@epfl.ch})}
\end{tabular}
}

\date{}

\maketitle







\fancyfoot[C]{\thepage}

\begin{abstract} \small
We describe a randomized algorithm for producing a near-optimal hierarchical off-diagonal low-rank (HODLR) approximation to an $n\times n$ matrix $\vec{A}$, accessible only though matrix-vector products with $\vec{A}$ and $\vec{A}^\T$. 
    We prove that, for the rank-$k$ HODLR approximation problem, our method achieves a $(1+\beta)^{\log(n)}$-optimal approximation in expected Frobenius norm using $O(k\log(n)/\beta^3)$ matrix-vector products.
    In particular, the algorithm obtains a $(1+\varepsilon)$-optimal approximation with $O(k\log^4(n)/\varepsilon^3)$ matrix-vector products, and for any constant $c$, an $n^c$-optimal approximation with $O(k \log(n))$ matrix-vector products.
    Apart from matrix-vector products, the additional computational cost of our method is just $O(n \poly(\log(n), k, \beta))$. 
    We complement the upper bound with a lower bound, which shows that any matrix-vector query algorithm requires at least $\Omega(k\log(n) + k/\varepsilon)$ queries to obtain a $(1+\varepsilon)$-optimal approximation.

    Our algorithm can be viewed as a robust version of widely used ``peeling'' methods for recovering HODLR matrices and is, to the best of our knowledge, the first matrix-vector query algorithm to enjoy theoretical worst-case guarantees for approximation by any hierarchical matrix class.
    To control the propagation of error between levels of hierarchical approximation, we introduce a new perturbation bound for low-rank approximation, which shows that the widely used Generalized Nystr\"om method enjoys inherent stability when implemented with noisy matrix-vector products. 
    We also introduce a novel \emph{randomly perforated} matrix sketching method to further control the error in the peeling algorithm. 
 \end{abstract}

\section{Introduction}
\label{sec:intro}

Many linear algebra tasks can be solved more efficiently in the presence of {structure}.
As such, a fundamental framework for designing matrix algorithms is to first approximate relevant matrices with structured ones, and then to solve the resulting structured problem. For example, matrices are frequently approximated by \emph{low-rank} matrices, which can be stored in less space, admit faster multiplication by vectors, and more. The power of low-rank structure has motivated significant work on finding \emph{near-optimal} {low-rank} approximations of arbitrary matrices \cite{DrineasMahoney:2005,DrineasKannanMahoney:2006a, HalkoMartinssonTropp:2011,ClarksonWoodruff:2013,NelsonNguyen:2013,Woodruff:2014,MuscoMusco:2015,MuscoMusco:2017}.

For problems in computational science and data science that involve multi-scale phenomena, however, low-rank approximation is often insufficient. 
Frequently, \emph{hierarchical low-rank structure} is more appropriate; long-range interactions at a given scale can be well-approximated by low-rank matrices, while shorter-range interactions can be recursively treated at a finer scale. Like low-rank matrices, hierarchical low-rank matrices are cheap to store and work with.
For instance, matrix-vector multiplies, linear solves, and other operations can typically be done in linear or nearly linear time in the dimension of the matrix \cite{AmbikasaranDarve:2013,ChandrasekaranDewildeGu:2007,VogelXiaCauley:2016,OuXia:2022}. 
For this reason, 
hierarchical low-rank approximations lie at the heart of many numerical methods, including the Fast Multipole Method (FMM) \cite{GreengardRokhlin:1987,YangDuraiswamiGumerov:2003a,YingBirosZorin:2004,YokotaIbeidKeyes:2017}, and have applications ranging from partial differential equation (PDE) solvers and control problems to the non-uniform Fourier transform \cite{BormGrasedyckHackbusch:2003,BallaniKressner:2016,WilberEpperlyBarnett:2024}. In fact, the study of hierarchical low-rank matrices is one of the most active topics of research in modern numerical linear algebra \cite{Hackbusch:2015,BallaniKressner:2016}.

There are many different types of hierarchical low-rank matrices, including $\mathcal{H}$ matrices, $\mathcal{H}^2$ matrices, and hierarchical semiseparable (HSS) matrices \cite{BallaniKressner:2016}. 
In this work, we consider one of the most general classes: hierarchical off-diagonal low-rank (HODLR) matrices.\footnote{HSS matrices are special cases of HODLR matrices, as are other well-studied structured classes, like diagonal plus low-rank matrices \cite{MuscoMuscoWoodruff:2021,TuncelVavasisXu:2023,BertsimasCopenhaverMazumder:2017}.}

\begin{definition}\label{def:hodlr}
Fix a rank parameter $k$. 
An $n\times n$ matrix $\vec{A}$ is $\HODLR(k)$ if $n\leq k$ or if $\vec{A}$ can be partitioned into $(n/2)\times (n/2)$ blocks 
\begin{equation*}
    \vec{A} 
    = 
    \begin{bmatrix}
        \vec{A}_{1,1} & \vec{A}_{1,2} \\
        \vec{A}_{2,1} & \vec{A}_{2,2}
    \end{bmatrix}
\end{equation*}
such that $\vec{A}_{1,2}$ and $\vec{A}_{2,1}$ are each of rank at most $k$ and $\vec{A}_{1,1}$ and $\vec{A}_{2,2}$ are each $\HODLR(k)$.\footnote{\Cref{def:hodlr} applies to matrices with dimension $n = n_{\text{base}}\cdot 2^L$ for some $n_{\text{base}} \leq k$ and integer $L\geq 0$.
Throughout this paper, we will always assume that $n$ and $k$ are related in this way.
More general definitions of HODLR matrices are possible, but \Cref{def:hodlr} is the most standard \cite{BallaniKressner:2016}.
}
\end{definition}
HODLR matrices are efficient to work with: they can be stored using just $O(nk\log (n/k))$ numbers, multiplied by vectors with $O(nk\log(n/k))$ operations, and admit linear system solves in $O(nk^3\log(n/k) + nk^2\log^2(n/k))$ time \cite{BallaniKressner:2016}. 

We are interested in finding a near-optimal HODLR approximation to a matrix $\vec{A}$ that can \emph{only} be accessed through black-box matrix-vector product queries $\vec{x}\mapsto \vec{A}\vec{x}$ and matrix-transpose-vector product queries $\vec{x} \mapsto \vec{A}^\T\vec{x}$ (both of which will henceforth be referred to as ``matvec queries to $\vec{A}$''). 
This problem has received significant attention due to both practical and theoretical importance \cite{LinLuYing:2011,Martinsson:2011,Martinsson:2016,LevittMartinsson:2024a,BoulleHalikiasTownsend:2023,BoulleTownsend:2023,LevittMartinsson:2024}. For example, black-box HODLR approximation methods can be used to obtain fast direct solvers in settings where we have access to efficient matrix-vector products, e.g., via an FMM or iterative method \cite{LinLuYing:2011,Martinsson:2016}. Black-box methods are also central in the  field of \emph{operator or PDE learning}, where $\vec A$ is an unknown differential operator and one can access matrix-vector products through physical experiments or simulations involving different forcing functions \cite{ BoulleTownsend:2022,BoulleHalikiasOttoTownsend:2024,SchaferOwhadi:2024}. 

In these applications and many others, matvec queries with $\vec A$ are expensive, and typically dominate other computational costs. As such, we hope to minimize the number of queries required to find a near-optimal HODLR approximation for $\vec{A}$. 
Matrix-vector query complexity has become central in theoretical work on numerical linear algebra due to the practical importance of the access model and the fact that it generalizes other models, such as the matrix sketching and Krylov subspace models \cite{SunWoodruffYangZhang:2021}. There has been recent progress giving tight query complexity bounds for central problems like linear system solving, eigenvector approximation, and trace estimation \cite{SimchowitzElAlaouiRecht:2018,BravermanHazanSimchowitzWoodworth:2020,MeyerMuscoMuscoWoodruff:2021,DharangutteMusco:2021b,ChewiDios-PontLi:2024,JiangPhamWoodruffZhang:2024}. Closer to our setting, there has also been significant work on the query complexity of structured matrix approximation, with  classes like low-rank matrices \cite{BakshiClarksonWoodruff:2022,BakshiNarayanan:2023}, diagonal matrices \cite{BekasKokiopoulouSaad:2007,TangSaad:2011,BastonNakatsukasa:2022,DharangutteMusco:2023}, sparse matrices \cite{CurtisPowellReid:1974,ColemanMore:1983,ColemanCai:1986,WimalajeewaEldarVarshney:2013,DasarathyShahBhaskarNowak:2015,AmselChenDumanKelesHalikiasMuscoMusco:2024}, and beyond \cite{WatersSankaranarayananBaraniuk:2011,SchaferKatzfussOwhadi:2021, HalikiasTownsend:2023}.

\subsection{Formal problem setup}

If $\vec{A}$ is exactly $\HODLR(k)$, it is well known that it can be {recovered} using $O(k \log (n/k))$ matvec queries using the so-called \emph{peeling algorithm} of Lin, Lu, and Ying \cite{LinLuYing:2011,Martinsson:2016,LevittMartinsson:2024a}.\footnote{For exact recovery, the number of queries used by the peeling algorithm is within a constant factor of optimal; see \cref{thm:lowerbd} and \cref{sec:lower_bounds} for a formal lower bound.}
We describe this algorithm in \cref{sec:peeling_intro}.
However, in most practical situations, $\vec{A}$ is  not  \emph{exactly} HODLR. 
This has resulted in broad concern about whether peeling algorithms applied to non-HODLR matrices might produce inaccurate approximations \cite{LinLuYing:2011,BoukaramTurkiyyahKeyes:2019,HalikiasTownsend:2023,BoulleEarlsTownsend:2022,BoulleHalikiasTownsend:2023}.
Notably, a recent SIAM Linear Algebra Best Paper \cite{BoulleTownsend:2022} poses an algorithmic challenge which roughly asks whether a near-HODLR matrix be approximated at nearly the same cost as algorithms for recovering an exactly HODLR matrix.
Our work addresses this challenge.
In particular, we study the following HODLR {approximation} problem, which makes no assumptions about $\vec{A}$.
\begin{problem}\label{prob:approx}
Given an $n\times n$ matrix $\vec{A}$, accessible only by matvec queries, a rank parameter $k\geq 1$, and an accuracy parameter $\Gamma \geq 1$, find a $\HODLR(k)$ matrix $\widetilde{\vec{A}}$ such that
\[
\| \vec{A} - \widetilde{\vec{A}} \|_\F \leq \Gamma \cdot \min_{\vec{H}\in\HODLR(k)} \|\vec{A} - \vec{H} \|_\F.
\]
\end{problem}
Interesting parameter regimes for \Cref{prob:approx} can range from $\Gamma = (1+\varepsilon)$ for $\varepsilon\in(0,1)$ to larger approximation factors when $\vec A$ is very close to $\HODLR(k)$, e.g., $\Gamma = O(\log n)$ or $\Gamma = O(n^c)$ for a small constant $c$.

\Cref{prob:approx} can be trivially solved for $\Gamma = 1$ using $n$ matrix-vector products. Via multiplication by an identity matrix, we can recover all entries of $\vec A$ and then compute an exactly optimal HODLR approximation by computing optimal rank-$k$ approximations to $\vec{A}$'s top right and bottom left $(n/2)\times(n/2)$ blocks, and recursing on the top left and bottom right blocks. However, outside this baseline and the case when $\vec{A}$ is exactly HODLR, we are unaware of any non-trivial results on \cref{prob:approx}. This is despite the vast literature on HODLR matrix approximation and on efficient matrix-vector query methods for vanilla low-rank approximation \cite{ClarksonWoodruff:2009,RokhlinSzlamTygert:2009,CohenElderMusco:2015,TroppYurtzeverUdellCevher:2017,BakshiClarksonWoodruff:2022,TroppWebber:2023,BakshiNarayanan:2023}.

\subsection{Main results}
\label{sec:contributions}
Our main contribution is an efficient algorithm for solving \cref{prob:approx}. 
\begin{theorem}\label{thm:main}
 Fix a rank parameter $k$ and accuracy parameter $\beta$. 
 Let $L =\lceil \log_2(n/k) \rceil$.
 There exists a non-adaptive\footnote{An algorithm that computes matrix-vector products $\vec{B}_1\vec{x}_1, \ldots, \vec{B}_t\vec{x}_t$ where $\vec{B}_i = \vec{A}$ or $\vec{B}_i=\vec{A}^\T$ for all $i \in [t]$ is \emph{adaptive} if the choice of $\vec B_i$ and $\vec{x}_i$ depends on the results of prior products $\vec{B}_1\vec{x}_1,\ldots, \vec{B}_{i-1}\vec{x}_{i-1}$. The algorithm is \emph{non-adaptive} (also called a \emph{sketching algorithm}) if the queries are all chosen in advance. Non-adaptive methods, like those studied in this paper, are often preferred, as they allow the queries to be evaluated in parallel.}
 algorithm (\Cref{alg:main}) that uses  $O(k/\beta^3 \cdot L)$ matvec queries to $\vec{A}$ and solves \cref{prob:approx} to accuracy $\Gamma = (1+\beta)^{L}$ with probability at least $99/100$\footnote{More generally, we show that \Cref{alg:main} outputs $\widetilde{\vec{A}}$  with $ \mathbb{E}  [\| \vec{A} - \widetilde{\vec{A}} \|_\F] \leq \Gamma \cdot \min_{\vec{H}\in\HODLR(k)} \|\vec{A} - \vec{H} \|_\F$. The high probability bound stated in \Cref{thm:main} can derived from this expectation bound via Markov's inequality -- see  \Cref{sec:analysis}.}. 
Apart from matvec queries, the algorithm requires $O(n \cdot  \poly(\log(n),k,\beta))$ additional runtime.
\end{theorem}

We highlight two interesting instantiations of \Cref{thm:main}. For any $\varepsilon>0$, we obtain accuracy $(1+\varepsilon)$ with $O(k\log^4(n/k)/\varepsilon^3)$ total matvec queries by setting $\beta = O(\varepsilon/\log(n/k))$. Alternatively, for any constant $c>0$, we obtain accuracy $n^c$ with $O(k\log(n/k))$ total matvec queries by setting $\beta = 2^{c} - 1 = \Theta(1)$. 
This matches the complexity of existing methods for solving the \emph{recovery problem}, i.e., the case when $\vec{A}$ is exactly HODLR. 

\Cref{thm:main} is obtained through a variant of the popular {peeling algorithm} for exact HODLR matrix recovery. 
This algorithm obtains an approximation via a ``top down'' approach. Specifically, it computes low-rank approximations to $\vec{A}$'s top right and bottom left blocks via randomized sketching, implicitly subtracts the results from $\vec{A}$, and then recurses on the upper left and lower right blocks. A major technical challenge for applying the peeling algorithm to \Cref{prob:approx} is how to control error that can accumulate across levels of recursion when $\vec{A}$'s top right and lower left blocks are not \emph{exactly} rank-$k$ (i.e., when the matrix is not exactly HODLR). This potential accumulation of error  has been highlighted in several prior works,  including in the earliest work on the peeling method \cite{LinLuYing:2011,BoukaramTurkiyyahKeyes:2019,BoulleEarlsTownsend:2022}, and has been the key challenge in extending peeling to the solve the HODLR approximation problem.

We resolve this long-standing challenge using two new techniques. First, we prove that if peeling is implemented with the so-called \emph{Generalized Nystr\"om Method} for low-rank approximation \cite{ClarksonWoodruff:2009,Nakatsukasa:2020}, sufficient oversampling leads to an algorithm that requires $kL/\poly(\beta)$ matrix-vector products to achieve error $\Gamma = (1+\beta)^{L+1}$. Our proof requires a novel perturbation analysis of sketching for low-rank approximation, and brings to light a surprising fact: the same result cannot be obtained if the more standard \emph{Randomized SVD} method \cite{HalkoMartinssonTropp:2011} is used for low-rank approximation within peeling. Second, we obtain our final result (with a better polynomial dependence on $\beta$) by combining peeling with a new ``randomly perforated'' sketching distribution that allows for even stronger control of error buildup across recursive levels. Details of the existing peeling algorithm are given in \Cref{sec:peeling_intro}, and our improvements are described in \Cref{sec:techniques}.

We complement our upper bound from \cref{thm:main} with a nearly matching lower-bound. 

\begin{theorem}
\label{thm:lowerbd}
For any $\varepsilon > 0$, any (possibly adaptive and randomized) algorithm that solves \cref{prob:approx} with error $\Gamma = (1+\varepsilon)$ and probability $\ge 1/25$ requires $\Omega(k \log(n/k) + k / \varepsilon)$ matvec queries to $\vec{A}$.\footnote{Formally, for a fixed constant $c$, we show that there is a distribution over inputs $\vec A \in \R^{n\times n}$ for $n \geq ck/\epsilon$ such that any (possibly randomized) algorithm using $O(k \log(n/k) + k/\epsilon)$ matvecs succeeds with probability $< 1/25$ over the randomness of the algorithm and the choice of input.} Moreover, any algorithm that solves \cref{prob:approx} for any finite $\Gamma$ and any non-zero probability, requires $\Omega(k \log(n/k))$ queries. 
\end{theorem}

\Cref{thm:lowerbd} establishes that our $O(k\log (n/k))$ query result to achieve error $n^c$ from \Cref{thm:main} cannot be improved in terms of the number of matvec queries, although it might be possible to obtain better error (e.g., a $\log(n)$ or even constant factor approximation) with the same number of matvecs. Additionally, \Cref{thm:lowerbd} establishes that our $O(k\log^4(n/k)/\varepsilon^3)$ query result for $(1+\varepsilon)$ error cannot be improved by more than log factors and a $1/\varepsilon^2$  factor. Interestingly, the lower bound shows a separation between the complexity of vanilla low-rank approximation and hierarchical low-rank approximation in terms of dependence on $\varepsilon$. The best known low-rank approximation algorithms use just $\tilde O(k/\varepsilon^{1/3})$ matrix-vector products to achieve relative error $(1+\varepsilon)$ \cite{BakshiClarksonWoodruff:2022,MeyerMuscoMusco:2024}.\footnote{The best known lower bound for vanilla $k$-rank approximation is $\Omega(k + 1/\varepsilon^{1/3})$ matrix-vector products \cite{BakshiClarksonWoodruff:2022}. Combining the $k$ and $\varepsilon$ terms, as we do in our \Cref{thm:lowerbd} for hierarchical approximation, is an open question.} In contrast, \Cref{thm:lowerbd} shows that a sublinear dependence on $1/\varepsilon$ \emph{cannot} be achieved for HODLR matrix approximation.

The proof of \Cref{thm:lowerbd} builds on a growing body of work on lower bounds for adaptive matrix-vector product algorithms \cite{SimchowitzElAlaouiRecht:2018,BravermanHazanSimchowitzWoodworth:2020,ChewiDeDiosPontLiLuNarayanan:2023}, which require techniques beyond those used to prove lower bounds, e.g., in the non-adaptive sketching setting. Our proof draws specifically on a recent result on the number of matrix-vector products required to obtain an optimal block diagonal approximation to a matrix $\vec{A}$ \cite{AmselChenDumanKelesHalikiasMuscoMusco:2024}.

\smallskip

\noindent\textbf{Organization.}
The remainder of the paper is organized as follows.
In \cref{sec:peeling_intro}, we describe the classic peeling algorithm for HODLR matrix recovery, and discuss the challenges in applying it to the HODLR matrix approximation problem (\cref{prob:approx}).
In \cref{sec:techniques}, we discuss the techniques we use to analyze our variant of peeling, including our new perturbation bound for low-rank approximation.
Then, in \cref{sec:alg}, we describe the exact implementation of our algorithm and introduce notation required for our main analysis. The final proof of \Cref{thm:main} is given in \cref{sec:analysis}. In \cref{sec:lower_bounds}, we prove our lower bound, \Cref{thm:lowerbd}.
Finally, in \cref{sec:examples} we show the results of some numerical experiments.

\section{The Peeling Algorithm}
\label{sec:peeling_intro}

In this section, we describe the classic peeling method for HODLR matrix recovery, on which our algorithm is based. As discussed, existing analyses of this method apply only in the setting when the input matrix $\vec A$ is exactly HODLR. We illustrate here how errors can arise when peeling  is applied to solve \cref{prob:approx} in the general case when $\vec{A}$ is only approximated by a HODLR matrix.  
The precise notation used in the figures will be fully described in \cref{sec:alg_notation}.

\subsection{Low-rank approximation}
\label{sec:LRA_algs}

The peeling algorithm makes use of matrix-vector query algorithms for low-rank approximation: given a matrix $\vec{B}\in\mathbb{R}^{m_1\times m_2}$, accessible only via matrix-vector queries, we would like to find a rank-$k$ approximation to $\vec{B}$, or to exactly recover $\vec B$ when $\mathrm{rank}(\vec B) \le k$.
It is well-known that the best low-rank approximation to $\vec{B}$ in the Frobenius norm is $\llbracket \vec{B} \rrbracket_k$, the rank-$k$ truncated SVD of $\vec{B}$.

Perhaps the most well-known algorithms for this task are the Randomized SVD (RSVD) \cite{HalkoMartinssonTropp:2011} and the Generalized Nystr\"om Method \cite{ClarksonWoodruff:2009,TroppYurtzeverUdellCevher:2017,Nakatsukasa:2020}, summarized below:

\noindent
\begin{minipage}{.42\textwidth}
\begin{algorithm}[H]
\caption{Randomized SVD}
\fontsize{11}{15}\selectfont
\begin{algorithmic}[1]
\Procedure{RSVD}{$\vec{B},k,\sright$}
\State $\vec{\Omega}\sim\Gaussian(m_2,\sright)$
\State $\vec{Q} = \operatorname{orth}(\vec{B}\vec{\Omega})$
\State $\vec{X} = \vec{Q}^\T\vec{B}$ \Comment{$\argmin_{\vec{Z}}\|\vec{B} - \vec{Q}\vec{Z}\|_\F$}
\State \Return $\vec{Q}\llbracket\vec{X}\rrbracket_k$
\EndProcedure
\end{algorithmic}
\end{algorithm}
\end{minipage}
\hfill
\begin{minipage}{.55\textwidth}
\begin{algorithm}[H]
\caption{Generalized Nystr\"om Method}
\fontsize{11}{15}\selectfont
\begin{algorithmic}[1]
\Procedure{GNM}{$\vec{B},k,\sright,\sleft$}
\State $\vec{\Omega}\sim\Gaussian(m_2,\sright)$,\hfill $\vec{\Psi}\sim\Gaussian(m_1,\sleft)$
\State $\vec{Q} = \operatorname{orth}(\vec{B}\vec{\Omega})$
\State $\vec{X} =(\vec{\Psi}^\T \vec{Q})^\dagger \vec{\Psi}^\T\vec{B}$ \Comment{$\argmin_{\vec{Z}}\|\vec{\Psi}^\T\vec{B} - \vec{\Psi}^\T\vec{Q}\vec{Z}\|_\F$}
\State \Return $\vec{Q}\llbracket\vec{X}\rrbracket_k$ 
\EndProcedure
\end{algorithmic}
\end{algorithm}
\end{minipage}
\vspace{1em}

Both methods first compute an orthonormal basis $\vec Q$ for the column span of  the matrix $\vec B \vec \Omega$, formed by multiplying $\vec B$ by a random sketching matrix $\vec \Omega$ with $\sright$ columns.
Throughout, we take this matrix to be Gaussian for simplicity, although most other popular sketching distributions can also be employed; see e.g. \cite{Woodruff:2014,HalkoMartinssonTropp:2011}. RSVD then projects $\vec B$ onto this column span and forms the best rank-$k$ approximation of the result. 
Generalized Nystr\"om follows the same idea, but instead computes an approximate projection of $\vec B$ onto $\vec Q$ by using a second sketch $\vec \Psi$ on the left. 

If $\mathrm{rank}(\vec{B}) \le k$, RSVD and Generalized Nystr\"om both recover $\vec{B}$ exactly (with probability one) if, respectively, $\sright \geq k$ or $\sright,\sleft \geq k$.
More generally, these methods can obtain a near-optimal low-rank approximation for arbitrary $\vec{B}$.
In particular, if $\sright = O(k/\beta)$, then the output of RSVD satisfies $\| \vec B - \vec Q \llbracket\vec{X}\rrbracket_k\|_\F \le (1+\beta)\cdot \|\vec{B} - \llbracket\vec{B}\rrbracket_k\|_\F$ with high probability \cite{Sarlos:2006}. 
The output of Generalized Nystr\"om gives the same guarantee when $\sright = O(k/\beta)$ and $\sleft = O(k/\beta^3)$ \cite{TroppYurtzeverUdellCevher:2016}.


\subsection{Exact HODLR recovery}\label{sec:exact}

Recall from \cref{def:hodlr} that any $\HODLR(k)$ matrix $\vec{A} \in \mathbb{R}^{n \times n}$ is composed to four $(n/2) \times (n/2)$ sized blocks. The off-diagonal blocks, $\vec A_{1,2}$ and $\vec A_{2,1}$, are rank-$k$ and the on-diagonal blocks, $\vec A_{1,1}$ and $\vec A_{2,2}$, are themselves $\HODLR(k)$.
The key idea of the peeling algorithm is to first recover the low-rank off-diagonal blocks $\vec A_{1,2}$ and $\vec A_{2,1}$, to implicitly subtract them from the matrix, and to then continue on to recursively recover the on-diagonal $\HODLR(k)$ blocks. More formally, peeling relies on the following observations:

\bigskip

\hypertarget{obs:first}{\noindent\textbf{Observation 1.}}
 We can perform matrix-vector products with the off-diagonal blocks $\vec A_{1,2}$ and $\vec A_{2,1}$ (and their transposes) using matrix-vector products with $\vec{A}$ and $\vec{A}^\T$. In particular, we can compute products with $\vec{A}_{2,1}$ and $(\vec{A}_{2,1})^\T$ by
\begin{equation}
    \begin{bmatrix}
        \vec{A}_{1,1} & \vec{A}_{1,2} \\
        \vec{A}_{2,1} & \vec{A}_{2,2}
    \end{bmatrix}    
    \begin{bmatrix}
        \vec{\Omega} \\
        \vec{0}
    \end{bmatrix}
    = 
    \begin{bmatrix}
        \sim \\
        \vec{A}_{2,1} \vec{\Omega}
    \end{bmatrix}
    ,\quad
    \begin{bmatrix}
        \vec{A}_{1,1} & \vec{A}_{1,2} \\
        \vec{A}_{2,1} & \vec{A}_{2,2}
    \end{bmatrix}^\T
    \begin{bmatrix}
         \vec{0} \\ \vec{\Psi}
    \end{bmatrix}
    = 
    \begin{bmatrix}
        (\vec{A}_{2,1})^\T \vec{\Psi}
        \\ \sim
    \end{bmatrix},
\end{equation}
and analogously with $\vec{A}_{1,2}$ and $(\vec{A}_{1,2})^\T$ by 
\begin{equation}
    \begin{bmatrix}
        \vec{A}_{1,1} & \vec{A}_{1,2} \\
        \vec{A}_{2,1} & \vec{A}_{2,2}
    \end{bmatrix}    
    \begin{bmatrix}
        \vec{0} \\
        \vec{\Omega}
    \end{bmatrix}
    = 
    \begin{bmatrix}
        \vec{A}_{1,2} \vec{\Omega} \\
        \sim
    \end{bmatrix}
    ,\quad
    \begin{bmatrix}
        \vec{A}_{1,1} & \vec{A}_{1,2} \\
        \vec{A}_{2,1} & \vec{A}_{2,2}
    \end{bmatrix}^\T
    \begin{bmatrix}
         \vec{\Psi} \\ \vec{0}
    \end{bmatrix}
    = 
    \begin{bmatrix}
        \sim \\
        (\vec{A}_{1,2})^\T\vec{\Psi}
    \end{bmatrix}.
\end{equation}
Here, ``\:$\sim$\:'' indicates a block of the output that we ignore in our computations.
As discussed in \cref{sec:LRA_algs}, algorithms like RSVD and Generalized Nystr\"om can exactly recover a rank-$k$ matrix (with probability one) using $k$ products with each the matrix and its transpose. Thus, when $\vec A$ is exactly $\HODLR(k)$, we can fully recover both off-diagonal blocks using just $4k$ total queries to $\vec{A}$ (implementing $k$ queries with each of $\vec A_{1,2},(\vec{A}_{1,2})^\T,\vec A_{2,1},$ and $(\vec{A}_{2,1})^\T$).\footnote{Algorithms in the literature \cite{LinLuYing:2011,LevittMartinsson:2024a,HalikiasTownsend:2023} commonly use a sketch size $k+p$, where $p$ is a small  ``oversampling'' parameter (e.g., $p=2$). 
This is important for adding robustness when recovering matrices which are numerically, but not exactly, rank-$k$. 
This is not necessary if we assume exact computations and exactly rank-$k$ blocks.}

\vspace{2em}
\begin{figure}[ht]
    \centering
    \scalebox{.8}{\begin{tikzpicture}
    \node[] at (-.85,2.05) {\large$\vec{A}^{(2)} = $};
    
    \draw[] (0,0) rectangle (4,4);
    
    \draw[] (0,0) rectangle (4,4);

    \draw[] (0,0) rectangle (2,2);
    \draw[] (2,2) rectangle (4,4);

    \foreach \x in {0,...,3}{
    \draw[dotted,line width=.25pt] (\x,0) -- (\x,4);
    \draw[dotted,line width=.25pt] (0,\x) -- (4,\x);
    }
    
    \draw[fill=black!10] (0,2) rectangle (1,3) node[pos=.5] {$\vec{A}^{(2)}_{2,1}$};
    \draw[fill=black!10] (1,3) rectangle (2,4) node[pos=.5] {$\vec{A}^{(2)}_{1,2}$};
    \draw[fill=black!10] (2,0) rectangle (3,1) node[pos=.5] {$\vec{A}^{(2)}_{4,3}$};
    \draw[fill=black!10] (3,1) rectangle (4,2) node[pos=.5] {$\vec{A}^{(2)}_{3,4}$};

    \draw[fill=black!20] (0,3) rectangle (1,4);
    \draw[fill=black!20] (1,2) rectangle (2,3);
    \draw[fill=black!20] (2,1) rectangle (3,2);
    \draw[fill=black!20] (3,0) rectangle (4,1);
\end{tikzpicture}
\hspace{2em}
\begin{tikzpicture}
    \node[] at (-.85,2.05) {\large$\vec{\Omega}^{+} = $};

    \draw[] (0,0) rectangle (1,4);
    
    \foreach \x in {1,3}{
    \draw[fill=black!10] (0,{4-\x}) rectangle (1,{4-\x+1}) node[pos=.5] {$\vec{\Omega}_{\x}$};
    }
\end{tikzpicture}
\hspace{.5em}
\begin{tikzpicture}
    \node[] at (-.85,2.05) {\large$\vec{\Psi}^{-} = $};

    \draw[] (0,0) rectangle (1,4);
    
    \foreach \x in {2,4}{
    \draw[fill=black!10] (0,{4-\x}) rectangle (1,{4-\x+1}) node[pos=.5] {$\vec{\Psi}_{\x}$};
    }
\end{tikzpicture}
\hspace{.5em}
\begin{tikzpicture}
    \node[] at (-.85,2.05) {\large$\vec{\Omega}^{-} = $};

    \draw[] (0,0) rectangle (1,4);
    
    \foreach \x in {2,4}{
    \draw[fill=black!10] (0,{4-\x}) rectangle (1,{4-\x+1}) node[pos=.5] {$\vec{\Omega}_{\x}$};
    }
\end{tikzpicture}
\hspace{.5em}
\begin{tikzpicture}
    \node[] at (-.85,2.05) {\large$\vec{\Psi}^{+} = $};

    \draw[] (0,0) rectangle (1,4);
    
    \foreach \x in {1,3}{
    \draw[fill=black!10] (0,{4-\x}) rectangle (1,{4-\x+1}) node[pos=.5] {$\vec{\Psi}_{\x}$};
    }
\end{tikzpicture}}
    \caption{
    State of the hierarchical matrix at the start of the $\ell=2$  level of peeling. $\vec A^{(2)}$ denotes the matrix $\vec A$ after subtracting the low-rank off-diagonal blocks recovered at the first level -- these blocks are zero in $\vec A^{(2)}$ and shown as white in the figure.
    For $j = 1, 3$, we can simultaneously obtain the products $\vec{A}^{(2)}_{j+1,j}\vec{\Omega}_j$ and the transpose-products $(\vec{A}^{(2)}_{j+1,j})^\T\vec{\Psi}_{j+1}$ from the sketches $\vec{A}\vec{\Omega}^+$ and $\vec{A}^\T\vec{\Psi}^-$.
    Letting $\vec \Omega^+$ and $\vec \Psi^-$ have $k$ columns and blocks chosen to be appropriate sketching matrices, these products can be used to exactly recover  the rank-$k$ blocks $\vec{A}^{(2)}_{j+1,j}$ for $j = 1,3$. 
    Obtaining products with and recovering $\vec{A}^{(2)}_{j-1,j}$ for $j = 2, 4$ can be done analogously using $\vec{\Omega}^{-}$ and $\vec{\Psi}^+$. Overall, we can recover all four rank-$k$ off diagonal blocks at level $\ell =2$ using just $4k$ queries to $\vec A$ ($k$ for each of $\vec \Omega^+, \vec \Omega^-, \vec \Psi^+,$ and $\vec \Psi^-$).
    }
    \label{fig:peeling}
\end{figure}
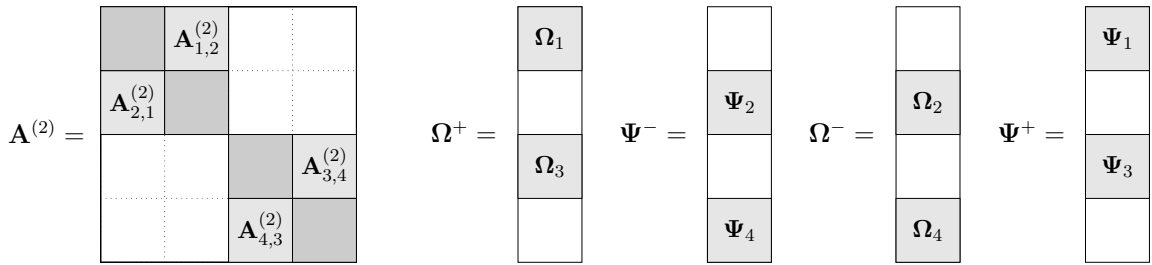

\bigskip

\hypertarget{obs:recurse}{\noindent\textbf{Observation 2.}}
After  {exactly} obtaining the blocks $\vec{A}_{2,1}$ and $\vec{A}_{1,2}$, we can \emph{simultaneously} perform matrix-vector products with the on-diagonal blocks $\vec{A}_{1, 1}$ and $\vec{A}_{2, 2}$ (and their transposes) using matrix-vector queries to $\vec A$. In particular, we can compute for any $\vec{\Omega}_1$ and $\vec{\Omega}_1$,
\[
    \begin{bmatrix}
        \vec{A}_{1,1} & \vec{A}_{1,2} \\
        \vec{A}_{2,1} & \vec{A}_{2,2}
    \end{bmatrix} 
    \begin{bmatrix}
        \vec{\Omega}_1 \\
        \vec{\Omega}_2
    \end{bmatrix}
    -  
    \begin{bmatrix}
        \vec{A}_{1,2} \vec{\Omega}_2 \\
        \vec{A}_{2,1} \vec{\Omega}_1
    \end{bmatrix}
    = 
    \begin{bmatrix}
        \vec{A}_{1,1}\vec{\Omega}_1 \\
        \vec{A}_{2,2}\vec{\Omega}_2
    \end{bmatrix},
    \]
    and analogously 
    \[
    \begin{bmatrix}
        \vec{A}_{1,1} & \vec{A}_{1,2} \\
        \vec{A}_{2,1} & \vec{A}_{2,2}
    \end{bmatrix}^\T
    \begin{bmatrix}
        \vec{\Psi}_1 \\
        \vec{\Psi}_2
    \end{bmatrix}
    -  
    \begin{bmatrix}
        (\vec{A}_{1,2})^\T\vec{\Psi}_2  \\
        (\vec{A}_{2,1})^\T\vec{\Psi}_1 
    \end{bmatrix}
    = 
    \begin{bmatrix}
        (\vec{A}_{1,1})^\T\vec{\Psi}_1 \\
        (\vec{A}_{2,2})^\T\vec{\Psi}_2
    \end{bmatrix}.
    \]
    Since $\vec A_{1,1}$ and $\vec A_{2,2}$ are themselves $\HODLR(k)$, 
  by \hyperlink{obs:first}{Observation 1}, we can recover their low-rank off diagonal blocks (each of size $(n/4) \times (n/4)$) using just $4k$ queries to $\vec{A}$ (see  \cref{fig:peeling}).

\bigskip

\noindent\textbf{Observation 3.} We can repeat this process, recursing toward the diagonal, to recover smaller and smaller off-diagonal low-rank blocks. 
Critically, the number of matrix-vector products used to recover the off-diagonal blocks at a given level depends only on the rank parameter $k$ and is \emph{independent} of the level itself.
After $L = \lceil \log_2(n/k)\rceil$ recursive steps, the blocks will be of size at most $k$.
At this point, the diagonal blocks can be recovered all at once using $k$ matrix-vector products. 
Overall, $4k$ queries to $\vec A$ are used at each level, and $k$ are used at the final level, giving total cost of $O(k\cdot (L+1))$ matvec queries. 

\subsection{Peeling in the presence of error}\label{sec:error}

The peeling algorithm described in \cref{sec:exact} assumes that the matrix $\vec{A}$ is exactly $\HODLR(k)$.
In particular, \hyperlink{obs:recurse}{Observation 2} does not hold if the off-diagonal blocks at the previous level are not recovered exactly.
This is illustrated in \cref{fig:peeling_err}.

As noted in \cite{LinLuYing:2011}, ``it is a natural concern that whether the error from low-rank decompositions on top levels accumulates in the peeling steps.''
In particular, if all of the error at a given level propagated to the next level through perturbations on the desired sketches, the error could double at each level; i.e., result in an \emph{exponential} blow-up of the error with respect to the number of levels.
In fact, as we demonstrate in \cref{sec:hard_instance}, certain variants of the peeling algorithm can exhibit such a failure mode on hard instances.
Thus, understanding the propagation of error from one level to the next is critical in the design and analysis of peeling algorithms.
In \cref{sec:techniques} we discuss the techniques we use to control and analyze this error propagation. 

\vspace{2em}
\begin{figure}[htb]
    \centering
    \scalebox{.7}{\begin{tikzpicture}
    \node[] at (-.85,4.1) {\large$\vec{A}^{(3)} = $};
    
    \draw[] (0,0) rectangle (8,8);
    
    \draw[] (0,0) rectangle (4,4);
    \draw[] (4,4) rectangle (8,8);

    \draw[] (0,4) rectangle (2,6);
    \draw[] (2,6) rectangle (4,8);
    \draw[] (4,0) rectangle (6,2);
    \draw[] (6,2) rectangle (8,4);

    \foreach \x in {0,...,7}{
    \draw[dotted,line width=.25pt] (\x,0) -- (\x,8);
    \draw[dotted,line width=.25pt] (0,\x) -- (8,\x);
    }
    
    \draw[fill=black!10] (0,6) rectangle (1,7); 
    \draw[fill=black!10] (1,7) rectangle (2,8); 
    \draw[fill=black!10] (2,4) rectangle (3,5); 
    \draw[fill=black!10] (3,5) rectangle (4,6); 
    \draw[preaction={fill, black!10},pattern=crosshatch, pattern color = Green!30] (4,2) rectangle (5,3) node[pos=.5,text=Green] {$\vec{A}^{(3)}_{6,5}$};
    \draw[fill=black!10] (5,3) rectangle (6,4); 
    \draw[fill=black!10] (6,0) rectangle (7,1); 
    \draw[fill=black!10] (7,1) rectangle (8,2); 

    \draw[fill=black!20] (0,7) rectangle (1,8);
    \draw[fill=black!20] (1,6) rectangle (2,7);
    \draw[fill=black!20] (2,5) rectangle (3,6);
    \draw[fill=black!20] (3,4) rectangle (4,5);
    \draw[fill=black!20] (4,3) rectangle (5,4);
    \draw[fill=black!20] (5,2) rectangle (6,3);
    \draw[fill=black!20] (6,1) rectangle (7,2);
    \draw[fill=black!20] (7,0) rectangle (8,1);

    \fill[pattern=north east lines, pattern color = Red!30] (0,2) rectangle (1,3) node[pos=.5,text=Red] {$\vec{A}^{(3)}_{6,1}$};
    \fill[pattern=north east lines, pattern color = Red!30] (2,2) rectangle (3,3) node[pos=.5,text=Red] {$\vec{A}^{(3)}_{6,3}$};
    \fill[pattern=north east lines, pattern color = Red!30] (6,2) rectangle (7,3) node[pos=.5,text=Red] {$\vec{A}^{(3)}_{6,7}$};

    \fill[pattern=north west lines, pattern color = Red!30] (4,6) rectangle (5,7) node[pos=.5,text=Red] {$\vec{A}^{(3)}_{2,5}$};
    \fill[pattern=north west lines, pattern color = Red!30] (4,4) rectangle (5,5) node[pos=.5,text=Red] {$\vec{A}^{(3)}_{4,5}$};
    \fill[pattern=north west lines, pattern color = Red!30] (4,0) rectangle (5,1) node[pos=.5,text=Red] {$\vec{A}^{(3)}_{8,5}$};
\end{tikzpicture}
\hspace{3em}
\begin{tikzpicture}
    \node[] at (-.85,4.1) {\large$\vec{\Omega}^{+} = $};

    \draw[] (0,0) rectangle (1,8);
    
    \foreach \x in {1,3,5,7}{
    \draw[fill=black!10] (0,{8-\x}) rectangle (1,{8-\x+1}) node[pos=.5] {$\vec{\Omega}_{\x}$};
    }
\end{tikzpicture}
\hspace{1em}
\begin{tikzpicture}
    \node[] at (-.85,4.1) {\large$\vec{\Psi}^{-} = $};

    \draw[] (0,0) rectangle (1,8);
    
    \foreach \x in {2,4,6,8}{
    \draw[fill=black!10] (0,{8-\x}) rectangle (1,{8-\x+1}) node[pos=.5] {$\vec{\Psi}_{\x}$};
    }
\end{tikzpicture}}
    \caption{
    State of the hierarchical matrix at the start of the $\ell=3$  level of peeling. $\vec A^{(3)}$ denotes the matrix $\vec A$ after subtracting off-diagonal blocks that were (approximately) recovered at the first two levels. 
    We would like to use the sketches $\vec{A}^{(3)}_{6,5}\vec{\Omega}_5$ and $(\vec{A}^{(3)}_{6,5})^\T\vec{\Psi}_{6}$  
    to obtain a low-rank approximation to the off-diagonal block $\vec{A}^{(3)}_{6,5}$.
    However, if the off-diagonal blocks at previous levels were not recovered exactly (since  these blocks may not be exactly rank-$k$), then after subtraction, these blocks will not be exactly zero in $\vec A^{(3)}$. 
    We thus obtain perturbed sketches of the form ${\color{Green}\vec{A}^{(3)}_{6,5}\vec{\Omega}_5} + {\color{Red}\vec{A}^{(3)}_{6,1}\vec{\Omega}_1}+{\color{Red}\vec{A}^{(3)}_{6,3}\vec{\Omega}_3}+{\color{Red}\vec{A}^{(3)}_{6,7}\vec{\Omega}_7}$ and ${\color{Green}(\vec{A}^{(3)}_{6,5})^\T\vec{\Psi}_6} + {\color{Red}(\vec{A}^{(3)}_{2,5})^\T\vec{\Psi}_2}+{\color{Red}(\vec{A}^{(3)}_{4,5})^\T\vec{\Psi}_4}+{\color{Red}(\vec{A}^{(3)}_{8,5})^\T\vec{\Psi}_8}$.
    }
    \label{fig:peeling_err}
    \vspace{-.5em}
\end{figure}
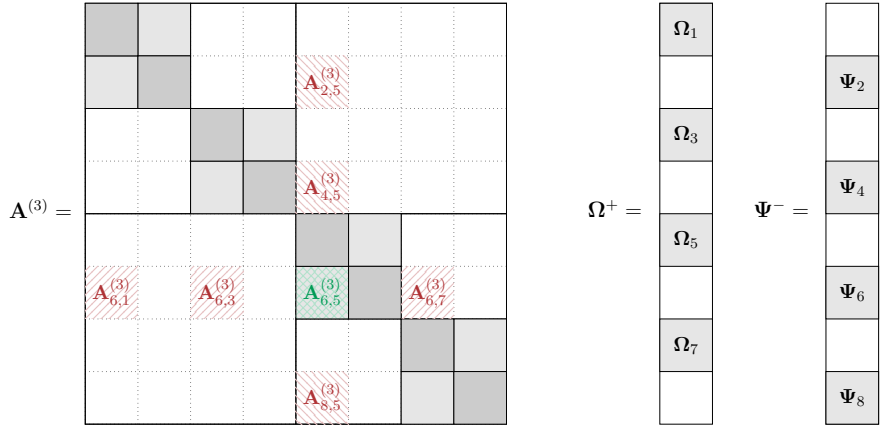

\clearpage

\section{Techniques}
\label{sec:techniques}

As discussed in \Cref{sec:error}, when the peeling algorithm is applied to a matrix $\vec A$ that is not exactly HODLR, errors at previous levels can pollute the matrix-vector products at a given level. 
The aim of this paper is to show that these errors can be controlled in such a way that the peeling algorithm can still solve \cref{prob:approx} to high accuracy.
To do this we must understand: (1) how matrix-vector query algorithms for low-rank approximation are impacted by noise in their matvecs and (2) how this noise can be controlled in the setting of peeling.\footnote{
In \cite{BoulleHalikiasTownsend:2023}, a similar approach is taken to analyze an infinite dimensional variant of the peeling algorithm for the task of approximating the Green's function of an elliptic partial differential (PDE) operator by a HODLR operator. 
However, \cite{BoulleHalikiasTownsend:2023} relies strongly on the fact that the off-diagonal blocks of the Green's function have rapid spectral decay and are thus low-rank, up to exponentially small error~\cite{BebendorfHackbusch:2003}. 
This allows the overall error to be controlled, even if the error can grow exponentially across levels.
In contrast, we make no assumptions about $\vec{A}$.
}
In the next two subsections we outline the high-level ideas we use to address each question.

\subsection{A perturbation bound for low-rank approximation}\label{sec:aperturbationBound}

To understand how matvec query errors impact low-rank approximation algorithms, including RSVD and Generalized Nystr\"om (see \cref{sec:LRA_algs}), we develop the following  perturbation bound, which we believe will be of general interest. 
The bound is proven in \cref{section:perturbation_rsvd}.

\newcommand{\RSVDperturb}{Let $\vec{B}\in\mathbb{R}^{m_1\times m_2}$, $\vec{\Omega}\in\mathbb{R}^{m_2\times s}$, $\vec{E}_1\in\mathbb{R}^{m_1\times s}$, and $\vec{E}_2\in\mathbb{R}^{s\times m_2}$, 
    and let $\vec{Q} = \operatorname{orth}(\vec{B}\vec{\Omega} + \vec{E}_1)$.
    Write $\vec B$ in its singular value decomposition as
    \begin{equation*}
        \vec{B} = \vec U \vec \Sigma \Vec{V^\T} =  \begin{bmatrix} \vec{U}_{\textup{top}} & \vec{U}_{\textup{bot}} \end{bmatrix} \begin{bmatrix} \vec{\Sigma}_{\textup{top}} & \\ & \vec{\Sigma}_{\textup{bot}} \end{bmatrix} \begin{bmatrix} \vec{V}_{\textup{top}}^\T \\\vec{V}_{\textup{bot}}^\T\end{bmatrix},
    \end{equation*}
    where $\vec U$ and $ \vec V$ are orthonormal and $\vec \Sigma$ is diagonal, and the ``top'' blocks represent the top $k$ columns; i.e., $\vec U_{\textup{top}} \in \R^{m_1 \times k}, \vec V_{\textup{top}} \in \R^{m_2 \times k},$ and $\vec\Sigma_{\textup{top}} \in \R^{k \times k} $.  
    Define also
    \begin{equation*}
        \vec{\Omega}_{\textup{top}} = \vec{V}_{\textup{top}}^\T \vec{\Omega}, \quad \vec{\Omega}_{\textup{bot}} = \vec{V}_{\textup{bot}}^\T \vec{\Omega}. 
    \end{equation*}    
    Then, assuming $\rank(\vec{\Omega}_{\textup{top}}) = k$,
    \begin{align*}
        \|\vec{B} - \vec{Q}\llbracket\vec{Q}^\T\vec{B} + \vec{E}_2 \rrbracket_k\|_{\F} 
        \leq \|\vec{E}_1 \vec{\Omega}_{\textup{top}}^{\dagger}\|_\F + 2 \|\vec{E}_2\|_\F + \big(\|\vec{\Sigma}_{\textup{bot}}\|_\F^2 + \|\vec{\Sigma}_{\textup{bot}} \vec{\Omega}_{\textup{bot}} \vec{\Omega}_{\textup{top}}^{\dagger}\|_{\F}^2\big)^{1/2}.
    \end{align*}}
\begin{theorem}
    \label{thm:RSVD_perturb}
    \RSVDperturb
\end{theorem}
\Cref{thm:RSVD_perturb} is most naturally a perturbation bound for the RSVD method, which, as discussed in \cref{sec:LRA_algs}, first computes an orthonormal basis $\vec Q$ for $\vec B \vec \Omega$ where $\vec \Omega$ is a random sketching matrix, and then outputs $\vec{Q}\llbracket\vec{Q}^\T\vec{B} \rrbracket_k$ -- the best rank-$k$ approximation to $\vec B$ lying in the span of $\vec Q$. 
The error term $\vec E_1$ captures error that occurs during the initial sketching step, i.e., due to noise in computing the matrix-vector products $\vec B \vec \Omega$. 
Similarly, $\vec E_2$ captures noise in the second projection step. 
\Cref{thm:RSVD_perturb} can be used to recover the standard bounds for  RSVD implemented with exact matrix-vector products (i.e., where $\vec{E}_1$ and $\vec{E}_2$ are zero). 
In particular, in this case the RSVD Frobenius error is bounded by: 
\[
    \|\vec{B} - \vec{Q}\llbracket\vec{Q}^\T\vec{B} \rrbracket_k\|_{\F} \le \big(\|\vec{\Sigma}_{\textup{bot}}\|_\F^2 + \|\vec{\Sigma}_{\textup{bot}} \vec{\Omega}_{\textup{bot}} \vec{\Omega}_{\textup{top}}^{\dagger}\|_{\F}^2\big)^{1/2} = \big(\|\vec B - \llbracket \vec B \rrbracket_k\|_\F^2 + \|\vec{\Sigma}_{\textup{bot}} \vec{\Omega}_{\textup{bot}} \vec{\Omega}_{\textup{top}}^{\dagger}\|_{\F}^2\big)^{1/2}.
\]
When the sketching matrix $\vec \Omega$ is Gaussian, we can observe that $\vec{\Omega}_{\textup{bot}}$ and $\vec{\Omega}_{\textup{top}}$ are independent Gaussian matrices since $\vec V_{\textup{top}}$ and $\vec V_{\textup{bot}}$ are orthogonal to each other.
This allows us to apply a standard bound (see \cref{thm:gaussian_expectations}) to argue that when $s = O(k/\beta)$, $\EE[\|\vec{\Sigma}_{\textup{bot}} \vec{\Omega}_{\textup{bot}} \vec{\Omega}_{\textup{top}}^{\dagger}\|_{\F}^2] \le \beta \|\vec{\Sigma}_{\textup{bot}}\|_\F^2 = \beta \|\vec B - \llbracket \vec B \rrbracket_k\|_\F^2$, which ultimately gives that $\EE[\|\vec{B} - \vec{Q}\llbracket\vec{Q}^\T\vec{B} \rrbracket_k\|_{\F}] \le (1+\beta) \|\vec B - \llbracket \vec B \rrbracket_k\|_\F.$

\Cref{thm:RSVD_perturb} can also be used as a perturbation bound for the Generalized Nystr\"om Method. 
In this case, $\vec E_2$ is used to account for the error due to using an approximate projection onto $\vec{Q}$ via a sketched regression problem, as well as the errors in the matvecs with $\vec{A}^\T$ used to set up the regression problem.
We discuss this further in \cref{sec:perturb_GNYS}. As with RSVD, \cref{thm:RSVD_perturb} matches the current best known bounds for Generalized Nystr\"om implemented with exact matvecs.

\subsection{Low-rank approximation in the peeling algorithm}\label{sec:overviewPerturb}

We now consider how \Cref{thm:RSVD_perturb} can be applied to control the propagation of error in the peeling algorithm implemented with different base algorithms for low-rank approximation of the off-diagonal blocks..

\smallskip

\noindent\textbf{Randomized SVD.} 
First consider implementing peeling with RSVD as the base low-rank approximation algorithm. In this case, $\vec E_1$ and $\vec E_2$ will be nonzero since we cannot compute exact matrix-vector products with an off-diagonal low-rank block $\vec B$ due to noise from previous levels (i.e., we cannot exactly compute $\vec B \vec \Omega$ and $\vec Q^\T \vec B$). Fortunately, the error  $\vec E_1$ that arises in peeling has a particular structure that can be used to bound the  
$\|\vec{E}_1 \vec{\Omega}_{\textup{top}}^{\dagger}\|_\F$ term in \cref{thm:RSVD_perturb}. In particular, recalling \cref{fig:peeling_err}, we will have 
\begin{equation}\label{eqn:error_form}
    \vec E_1 = \sum_{j} \vec{M}_j \vec{\Omega}_j 
    = 
    \vec{M}\vec{\widetilde \Omega}
    ,\qquad
    \vec{M}=
    \begin{bmatrix}
        \vec{M}_1 & 
        \vec{M}_2 & 
        \cdots 
    \end{bmatrix}
    ,\qquad 
    \vec{\widetilde \Omega} = 
    \begin{bmatrix}
        \vec{\Omega}_1 \\
        \vec{\Omega}_2 \\
        \vdots
    \end{bmatrix},
\end{equation}
where the $\vec{M}_j$ are fixed matrices that depend on the recovery error at the previous levels of peeling, and the $\vec \Omega_j$ are independent Gaussian matrices.\footnote{In \cref{fig:peeling_err}, the $\vec M_j \vec \Omega_j$ terms when $\vec B = \vec{A}^{(3)}_{6,5}$ are $\vec{A}^{(3)}_{6,1}\vec{\Omega}_1$, $\vec{A}^{(3)}_{6,3}\vec{\Omega}_3$, and $\vec{A}^{(3)}_{6,7}\vec{\Omega}_7$.} 
$\vec{\widetilde \Omega}$ is itself a Gaussian matrix that is independent from $\vec{\Omega}_{\textup{top}}$, allowing us to bound the $\EE[\|\vec{E}_1 \vec{\Omega}_{\textup{top}}^{\dagger}\|_\F^2]$ term in \Cref{thm:RSVD_perturb} by $\EE[\|\vec{E}_1 \vec{\Omega}_{\textup{top}}^{\dagger}\|_\F^2] = \EE[\|\vec{M}  \vec{\widetilde \Omega} \vec{\Omega}_{\textup{top}}^{\dagger}\|_\F^2]] \le \beta \|\vec M \|_\F^2$ when $s = O(k/\beta)$, again using \cref{thm:gaussian_expectations}. I.e., with a large enough sketch size, we can drive down the error term $\vec E_1$ to be arbitrarily small, and ensure that it does not accumulate too much across the levels of peeling.

Unfortunately, we cannot handle the $\|\vec E_2\|_\F$ term in \cref{thm:RSVD_perturb} in the same way. We have $\vec E_2 = \sum_{j} \vec M_j \vec Q_j$ for orthogonal matrices $\vec Q_j$ computed for each off-diagonal block at the current level. The $\vec Q_j$ and $\vec M_j$ matrices may be arbitrarily correlated, and so it is not clear that $\|\vec E_2\|_\F$ can be made small compared to the noise from the previous levels, $\|\vec M\|_\F$. In fact, as discussed in \cref{sec:hard_instance}, theoretical and empirical evidence suggests the existence of hard instances where standard peeling with RSVD can suffer exponential error blow-up across levels. 

\smallskip

\noindent\textbf{Generalized Nystr\"om.}
\label{sec:perturb_GNYS}
 Our first key observation is that when peeling is implemented with the Generalized Nystr\"om method as the base low-rank approximation algorithm, rather than RSVD, the $\| \vec E_2 \|_\F$ error term in \cref{thm:RSVD_perturb} can in fact be bounded. For peeling implemented with Generalized Nystr\"om, $\vec E_2$ encapsulates two sources of error: the first is due to the fact that Generalized Nystr\"om does not exactly return $\vec{Q}\llbracket\vec{Q}^\T\vec{B} \rrbracket_k$, but rather computes $\vec X = \argmin_{\vec{Z}}\|\vec{\Psi}^\T\vec{B} - \vec{\Psi}^\T\vec{Q}\vec{Z}\|_\F$ for a random sketching matrix $\vec \Psi$ and returns $\vec{Q}\llbracket\vec{X} \rrbracket_k$. We can think of $\vec X$ as an approximation to $\vec Q^\T \vec B = \argmin_{\vec Z} \|\vec B - \vec Q \vec Z\|_\F$. It is well known that this source of error can be controlled by setting the size of the sketch $\vec \Psi$ large enough -- this leads to the standard analysis of Generalized Nystr\"om from the literature \cite{TroppYurtzeverUdellCevher:2016}.

The second source of error in $\vec E_2$ is due to noise in computing the matrix-vector products $\vec \Psi^\T \vec B$. Fortunately, analogously to how we bounded $\vec E_1$ for RSVD in \cref{sec:aperturbationBound}, we can leverage the fact that in peeling, this noise has the structure of equation \eqref{eqn:error_form} -- it is the product of a fixed matrix $\vec M$ (the recovery error at previous levels of peeling) and a random Gaussian matrix $\vec{\widetilde \Psi}$ that is independent of $\vec \Psi$. Using this structure, we are able to bound its impact on the final approximation quality. In particular, we use \cref{thm:RSVD_perturb} to give the following bound for the Generalized Nystr\"om Method with Gaussian errors as in \cref{eqn:error_form}.
This bound is proven in \cref{section:perturbation_rsvd}.

\newcommand{\GNperturbexp}{Let $\vec{B}\in\mathbb{R}^{m_1\times m_2}$, $\vec{M}\in\mathbb{R}^{m_1\times p}$, and $\vec{N}\in\mathbb{R}^{q\times m_2}$ be fixed matrices, and let $\vec{\Omega}\sim\Gaussian(m_2,\sright)$, $\widetilde{\vec{\Omega}}\sim\Gaussian(p,\sright)$, $\vec{\Psi}\sim\Gaussian(m_1,\sleft)$, and $\widetilde{\vec{\Psi}}\sim\Gaussian(q,\sleft)$ be independent.
Define 
\[
\vec{Q} := \operatorname{orth}(\vec{B}\vec{\Omega} + \vec{M}\widetilde{\vec{\Omega}})
,\qquad
\vec{X} := ( \vec{\Psi}^\T\vec{Q})^\dagger ( \vec{\Psi}^\T \vec{B} + \widetilde{\vec{\Psi}}^\T \vec{N} ).
\]
Then, provided $\sright > 2k+1$ and $\sleft > 2\sright+1$
\begin{align*}
    \EE\Bigl[\|\vec{B} - \vec{Q}\llbracket \vec{X} \rrbracket_k\|_{\F}^2 \Big]
    \leq E_1 + E_2 + 2\sqrt{E_1 E_2},
\end{align*}
where
\begin{align*}
    E_1 &:= \left(1+\frac{k}{\sright-k-1} \right) \| \vec{B} - \llbracket \vec{B} \rrbracket_k \|_\F^2
    \\
    E_2 &:= \frac{18k}{\sright-k-1} \|\vec{M}\|_\F^2 + \frac{8\sright}{\sleft-\sright-1}\|\vec{N}\|_\F^2 + \frac{32\sright}{\sleft-\sright-1} \| \vec{B} - \llbracket \vec{B} \rrbracket_k \|_\F^2.
\end{align*}}
\begin{theorem}\label{thm:GN_perturb_exp}
\GNperturbexp
\end{theorem}
We can see from \cref{thm:GN_perturb_exp} that, for fixed noise matrices $\vec M$ and $\vec N$, if we set $\sright$ large enough compared to $k$ and $\sleft$ large enough compared to $\sright$, we can drive the expected error of the rank-$k$ approximation $\vec{Q}\llbracket \vec{X} \rrbracket_k$ arbitrarily close to the error of the best possible rank-$k$ approximation $\llbracket \vec{B} \rrbracket_k$. 
Formally, in \cref{sec:analysis} (\cref{thm:main_full}) we use \cref{thm:GN_perturb_exp} to show that for any $\beta \in (0,1)$, if we set $\sright = O(k/\beta^2)$ and $\sleft = O(\sright/\beta^2) = O(k/\beta^4)$, then at each level of peeling, our approximation error blows up by at most a $(1+\beta)$ factor. Over $L+1$ levels, we obtain final accuracy $\Gamma = (1+\beta)^{L+1}$. This matches our main result, \cref{thm:main}, but with a slightly worse dependence on $\beta$: we require $O(k/\beta^4 \cdot L)$ rather than $O(k/\beta^3 \cdot L)$ total matvecs. In \cref{sec:sketching_dist} we show how to give the improved $\beta$ dependence using a novel sketching approach to further control the noise terms $\vec M$ and $\vec N$.

We note that the algorithm described above, which simply implements standard peeling with Generalized Nystr\"om is particularly simple and performs well experimentally -- see \cref{sec:examples}. It would be interesting to understand if our bounds on $\sright$ and $\sleft$ are tight -- we suspect that they are not. In particular, even in the case of exact matrix-vector products, we suspect that $\sleft = O(k/\beta^2)$ suffices to achieve a $(1+\beta)$ approximation, even though current best known bound is $O(k/\beta^3)$. Giving an improved bound in this setting would very likely yield an improved bound in our setting.

\subsection{Randomly perforated Gaussian sketching}
\label{sec:sketching_dist}

We next discuss how we can further improve the query complexity of peeling with Generalized Nystr\"om-based low-rank approximation by altering the sketching distribution to explicitly reduce matvec query errors that arise due to inexact recovery at previous levels. 

Referring back to \cref{fig:peeling,fig:peeling_err}, we observe that at each level the peeling algorithm employs two random sketching matrices on the right: $\vec \Omega^+$ and $\vec \Omega^-$ and two on the left: $\vec \Psi^+$ and $\vec \Psi^-$. 
For sake of exposition we focus here just on the right sketches, $\vec \Omega^+$ and $\vec \Omega^-$, which are both  block matrices with $2^{\ell}$ blocks, alternating between random Gaussian blocks (i.e., $\vec{\Omega}_1$, $\vec{\Omega}_2$, $\vec{\Omega}_3$, ...) and zero blocks. 
In particular, $\vec \Omega^+$ has even blocks set to zero, while $\vec \Omega^-$ has odd blocks set to zero. 

This ``perforated'' structure arises naturally from \hyperlink{obs:first}{Observation 1} and \hyperlink{obs:recurse}{Observation 2}. 
Intuitively, since $\vec \Omega^+$ has even blocks set to zero, it has no interaction with the diagonal blocks at level $\ell$ with even indices. 
Thus, it can be used to simultaneously produce right sketches of every bottom-left off-diagonal block at level $\ell$, without incurring any error due to the unrecovered on-diagonal blocks. 

Of course, as discussed, when the peeling algorithm is applied to a non-HODLR matrix, the sketch does incur error due to inexact recovery at the previous levels. 
In particular, referring to \cref{fig:peeling_err}, when approximating a given off-diagonal block at level $\ell$ (e.g., $\vec A_{6,5}^{(3)}$ in the figure), the sketch incurs error from $2^{\ell-1}-1$ blocks that were only approximately recovered in previous levels of peeling (i.e., $\vec A_{6,1}^{(3)}, \vec A_{6,3}^{(3)},$ and $\vec A_{6,7}^{(3)}$ in the figure).

Our key idea to reduce this error is simple: we will increase the sparsity of $\vec \Omega^+$ and $\vec \Omega^-$, so that a higher fraction of blocks are set to zero.  
Thus, when recovering each block at level $\ell$, we will incur error due to a smaller number of inexactly recovered blocks from the previous levels. We still need non-zero blocks for each of the $2^\ell$ off-diagonal blocks that are recovered at level $\ell$, so our sparser matrices will have a large number of block columns than before.  
See \cref{fig:peeling_err_perf} for an illustration.

\begin{figure}[ht]
    \centering
    \scalebox{.7}{\begin{tikzpicture}
    \node[] at (-.85,4.1) {\large$\vec{A}^{(3)} = $};
    
    \draw[] (0,0) rectangle (8,8);
    
    \draw[] (0,0) rectangle (4,4);
    \draw[] (4,4) rectangle (8,8);

    \draw[] (0,4) rectangle (2,6);
    \draw[] (2,6) rectangle (4,8);
    \draw[] (4,0) rectangle (6,2);
    \draw[] (6,2) rectangle (8,4);

    \foreach \x in {0,...,7}{
    \draw[dotted,line width=.25pt] (\x,0) -- (\x,8);
    \draw[dotted,line width=.25pt] (0,\x) -- (8,\x);
    }
    
    \draw[fill=black!10] (0,6) rectangle (1,7) node[pos=.5] {};
    \draw[fill=black!10] (1,7) rectangle (2,8) node[pos=.5] {};
    \draw[fill=black!10] (2,4) rectangle (3,5) node[pos=.5] {};
    \draw[fill=black!10] (3,5) rectangle (4,6) node[pos=.5] {};
    \draw[preaction={fill, black!10},pattern=crosshatch, pattern color = Green!30] (4,2) rectangle (5,3) node[pos=.5,text=Green] {$\vec{A}^{(3)}_{6,5}$};
    \draw[fill=black!10] (5,3) rectangle (6,4) node[pos=.5] {};
    \draw[fill=black!10] (6,0) rectangle (7,1) node[pos=.5] {};
    \draw[fill=black!10] (7,1) rectangle (8,2) node[pos=.5] {};

    \draw[fill=black!20] (0,7) rectangle (1,8);
    \draw[fill=black!20] (1,6) rectangle (2,7);
    \draw[fill=black!20] (2,5) rectangle (3,6);
    \draw[fill=black!20] (3,4) rectangle (4,5);
    \draw[fill=black!20] (4,3) rectangle (5,4);
    \draw[fill=black!20] (5,2) rectangle (6,3);
    \draw[fill=black!20] (6,1) rectangle (7,2);
    \draw[fill=black!20] (7,0) rectangle (8,1);

    \fill[pattern=north east lines, pattern color = Red!30] (0,2) rectangle (1,3) node[pos=.5,text=Red] {$\vec{A}^{(3)}_{6,1}$};

\end{tikzpicture}
\hspace{1em}
\begin{tikzpicture}
    \node[] at (-.85,4.1) {\large$\vec{\Omega}^{+} = $};

    \draw[] (0,0) rectangle (3,8);
    \draw[dashed] (0,0) rectangle (1,8);

    \foreach \x in {1,2}{
        \draw[dotted,line width=.25pt] (\x,0) -- (\x,8);
    }
    \foreach \y in {1,2,...,7}{
        \draw[dotted,line width=.25pt] (0,\y) -- (3,\y);
    }
    
    \foreach \x in {1,5}{
    \draw[fill=black!10] (0,{8-\x}) rectangle (1,{8-\x+1}) node[pos=.5] {$\vec{\Omega}_{\x}$};
    }
    \foreach \x in {3}{
    \draw[fill=black!10] (1,{8-\x}) rectangle (2,{8-\x+1}) node[pos=.5] {$\vec{\Omega}_{\x}$};
    }
    \foreach \x in {7}{
    \draw[fill=black!10] (2,{8-\x}) rectangle (3,{8-\x+1}) node[pos=.5] {$\vec{\Omega}_{\x}$};
    }
\end{tikzpicture}
\hspace{1em}
\begin{tikzpicture}
    \node[] at (-.85,4.1) {\large$\vec{\Psi}^{-} = $};

    \draw[] (0,0) rectangle (3,8);
    \draw[dashed] (1,0) rectangle (2,8);

    \foreach \x in {1,2}{
        \draw[dotted,line width=.25pt] (\x,0) -- (\x,8);
    }
    \foreach \y in {1,2,...,7}{
        \draw[dotted,line width=.25pt] (0,\y) -- (3,\y);
    }

    \foreach \x in {4}{
    \draw[fill=black!10] (0,{8-\x}) rectangle (1,{8-\x+1}) node[pos=.5] {$\vec{\Psi}_{\x}$};
    }    
    \foreach \x in {6}{
    \draw[fill=black!10] (1,{8-\x}) rectangle (2,{8-\x+1}) node[pos=.5] {$\vec{\Psi}_{\x}$};
    }
    \foreach \x in {2,8}{
    \draw[fill=black!10] (2,{8-\x}) rectangle (3,{8-\x+1}) node[pos=.5] {$\vec{\Psi}_{\x}$};
    }

\end{tikzpicture}}
    \caption{
    As in \cref{fig:peeling,fig:peeling_err}, we aim to obtain the sketches $\vec{A}^{(3)}_{j\pm 1,j}\vec{\Omega}^+_j$ and $(\vec{A}^{(3)}_{j\pm 1,j})^\T\vec{\Psi}^-_{j\pm1}$ from the sketches $\vec{A}^{(3)}\vec{\Omega}^+$ and $(\vec{A}^{(3)})^\T\vec{\Psi}^-$.
    The key observation is that by using \emph{randomly perforated Gaussian sketches}, which decrease the number of nonzero blocks per block-column (while increasing the number of column.
    This results in less error.
    We illustrate the error for block $\vec A_{6,5}^{(3)}$. 
    Note that we obtain sketches ${\color{Green}\vec{A}^{(3)}_{6,5}\vec{\Omega}_5} + {\color{Red}\vec{A}^{(3)}_{6,1}\vec{\Omega}_1}$ and ${\color{Green}(\vec{A}^{(3)}_{6,5})^\T\vec{\Psi}_6}$.
    Thus, the error is decreased compared to \cref{fig:peeling_err}.
    }
    \label{fig:peeling_err_perf}
\end{figure}
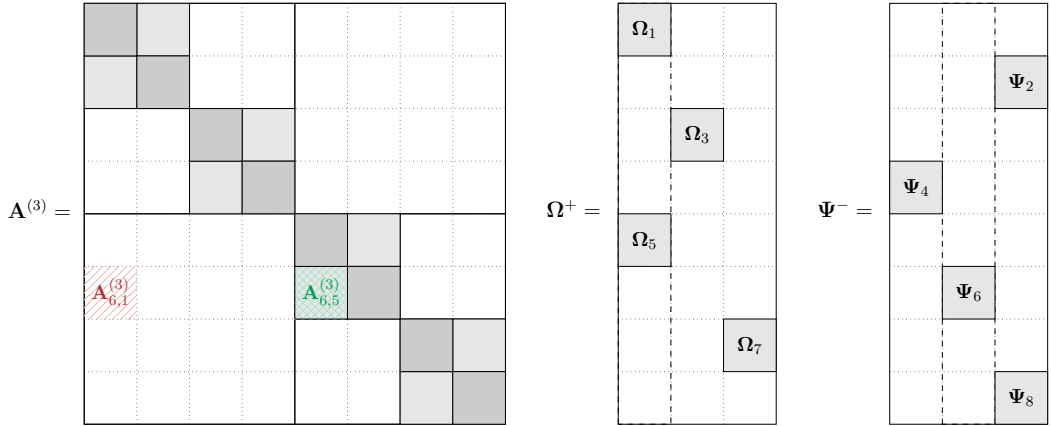

Of course, just decreasing the number of inexactly recovered blocks that introduce error into each of our sketches may not decrease the magnitude of this error if the  error is all concentrated on a few blocks. 
Thus, we will chose the nonzero blocks of our sketches randomly, ensuring that the expected error when recovering each block at level $\ell$ is small.
Formally, our modified sketches $\vec \Omega^+, \vec \Omega^-, \vec \Psi^+,$ and $\vec \Psi^-$ will be $2^\ell \times t$ block matrices, with a single Gaussian block placed randomly in each odd (resp. even) block row. 
They can be thought of as block Kronecker products of Gaussian matrices with \emph{Count-sketch-like} sparse sketching matrices -- see \cref{sec:sketching_dists} for a more complete description. 
We call our sketches \emph{randomly perforated Gaussian sketches}.

By increasing the number of block columns $t$ in our randomly perforated sketches, we decrease the expected matvec error introduced by inexactly recovered blocks at each level.
In particular, we show that the expected noise added to each matvec query is decreased by a factor of $1/t$ in the squared Frobenius norm. 
We describe our variant of Generalized Nystr\"om-based peeling with perforated sketching in \cref{alg:main} and analyze it in \cref{sec:analysis} (\cref{thm:main_full}), using the perturbation bound for Generalized Nystr\"om described in \cref{sec:overviewPerturb}. Ultimately, this approach gives  our main result \cref{thm:main}. 

We note that perforated sketching can also be combined with vanilla RSVD-based peeling, giving a similar bound to \cref{thm:main}. 
We describe the resulting algorithm in \cref{sec:RSVD_peel}.
We numerically compare our main Generalized Nystr\"om Method-based algorithm with the RSVD-based algorithm in \cref{sec:examples}.

\section{The Generalized Nystr\"om Method Peeling Algorithm}
\label{sec:alg}

We now describe our variant of the peeling algorithm, which is summarized as \cref{alg:main}.
In \Cref{sec:alg_notation,sec:sketching_dists} we describe preliminary notation and give pseudocode for the algorithm.
Then, in \Cref{sec:left_subspace,sec:lra_from_subspace}  we describe how the two key stages of the Generalized Nystr\"om Method are approximately implemented in the peeling algorithm and introduce notation we will use in the analysis. 
Finally, in \cref{sec:diag_recover} we discuss how to recover the diagonal blocks at the final level. Analysis of the method appears in \Cref{sec:analysis}.

\subsection{Notation for sketching distributions}
\label{sec:sketching_dists}

We begin by defining notation needed to describe the ``randomly perforated Gaussian'' sketches used in our variant of the peeling algorithm.

\begin{definition}\label{def:product}
Let $\vec{X}\in \mathbb{R}^{p\times v}$ and $\vec{Y}\in\mathbb{R}^{pu\times t}$. We define the product $\vec{X}\bullet\vec{Y} \in \mathbb{R}^{pu\times vt}$ by
\[
\underbrace{
\begin{bmatrix}
    x_{1,1} & \cdots & x_{1,v} \\
    \vdots & & \vdots \\
    x_{p,1} & \cdots & x_{p,v} \\
\end{bmatrix}}_{\vec{X}}
\bullet 
\underbrace{
\begin{bmatrix}
    \vec{Y}_{1} \\
    \vdots \\
    \vec{Y}_{p} \\
\end{bmatrix}}_{\vec{Y}}
= 
\underbrace{
\begin{bmatrix}
    x_{1,1}\vec{Y}_{1} & \cdots & x_{1,v}\vec{Y}_{1} \\
    \vdots & & \vdots \\
    x_{p,1}\vec{Y}_{p} & \cdots & x_{p,v}\vec{Y}_{p} \\
\end{bmatrix}}_{\vec{X}\bullet\vec{Y}},
\]
where $\vec{Y}_1, \ldots, \vec{Y}_p \in \mathbb{R}^{u\times t}$.
\end{definition}
Note that ``~$\bullet$~'' is a block row-wise Kronecker product; i.e. the block-rows of $\vec{X}\bullet\vec{Y}$ are obtained as the Kronecker product of the rows of $\vec{X}$ and the block-rows of $\vec{Y}$.

\begin{definition}
We say $\vec{\Omega} \sim \Gaussian(p,q)$ if each entry of $\vec{\Omega}\in\mathbb{R}^{p\times q}$ is independently a standard normal random variable.
\end{definition}

\begin{definition}\label{def:CS}
We say $\bm{\xi}\sim\operatorname{CountSketch}(d,t)$ if $\bm{\xi}\in\{0,1\}^{d\times t}$ is such that each row of $\bm{\xi}$ independently has exactly one nonzero entry in a uniformly random column. 
We respectively denote the $(j,i)$ entry of $\bm{\xi}$ as $\xi_{j,i}$.
\end{definition}

\begin{definition}\label{def:perforatedCS}
We say $\bm{\xi}^+,\bm{\xi}^-\sim\operatorname{PerfCountSketch}(d,t)$ if 
\begin{equation*}
    \bm{\xi}^{+} 
    =
    \begin{bmatrix}
        1&0&1&0&\cdots&
    \end{bmatrix}^\T
    \bullet
    \bm{\xi}
    ,\qquad
    \bm{\xi}^{-} 
    =
    \begin{bmatrix}
        0&1&0&1&\cdots&
    \end{bmatrix}^\T
    \bullet
    \bm{\xi},
\end{equation*}
where $\bm{\xi}\sim\operatorname{CountSketch}(d,t)$.
We respectively denote the $(j,i)$ entry of $\bm{\xi}^{+}$ and $\bm{\xi}^{-}$ as $\xi^+_{j,i}$ and $\xi^-_{j,i}$.
\end{definition}

We are now ready to formally define the randomly perforated Gaussian sketching distribution. 
\begin{definition}\label{def:RandPerfGaussian}
We say $\bm{\xi}^+,\bm{\xi}^-,\vec{\Omega}^+,\vec{\Omega}^-\sim\operatorname{RandPerfGaussian}(n,d,s,t)$ if 
\[
\vec{\Omega}^+ = \bm{\xi}^+ \bullet  \begin{bmatrix}
    \vec{\Omega}_{1} \\
    \vdots \\
    \vec{\Omega}_{d} \\
\end{bmatrix},
\qquad
\vec{\Omega}^- = \bm{\xi}^- \bullet  \begin{bmatrix}
    \vec{\Omega}_{1} \\
    \vdots \\
    \vec{\Omega}_{d} \\
\end{bmatrix},
\]
where $\bm{\xi}^+,\bm{\xi}^-\sim\operatorname{PerfCountSketch}(d,t)$ and each $\vec{\Omega}_i\sim \Gaussian(n/d,s)$ independently.
\end{definition}

The sketching matrices $\vec{\Omega}^+$ and $\vec{\Psi}^-$ shown in \cref{fig:peeling_err_perf} illustrate the sparsity structure of the randomly perforated Gaussian sketching distribution (\cref{def:RandPerfGaussian}) with $t=3$.
In the case that $t=1$, then the randomly perforated Gaussian sketching matrices are of the form used by the standard peeling algorithm -- see \cref{fig:peeling}.

\subsection{Notation for iteration and pseudocode}
\label{sec:alg_notation}

At level $\ell$, the peeling algorithm aims to approximate the $2^\ell$ low-rank off-diagonal blocks of $\vec{A}$, each of which is of size $n/2^\ell$.

The resulting approximation to these blocks is placed in a matrix $\vec{H}^{(\ell)}$.
This is repeated for $L := \lceil \log_2(n/k) \rceil \leq O(\log(n/k))$ levels, at which points the blocks are of size at most $k$. 
At the final level $\ell=L$, the algorithm also obtains an approximation $\widehat{\vec{H}}$ to the $2^L$ diagonal blocks.
The final approximation is then expressed in terms of the approximations at each level as 
\begin{equation}\label{eqn:final_approx}
    \widetilde{\vec{A}} = \vec{H}^{(1)} + \cdots + \vec{H}^{(L)} + \widehat{\vec{H}}.
\end{equation}

Suppose we are at level $\ell$, and let $\vec{A}^{(\ell)}$ be the matrix $\vec{A}$ at step $\ell$ after subtracting the approximations $\vec{H}^{(\ell-1)}, \ldots, \vec{H}^{(1)}$ from the previous levels; that is 
\[
\vec{A}^{(\ell)} 
:= \vec{A} - (\vec{H}^{(\ell-1)} + \cdots + \vec{H}^{(1)}).
\]
Partition $\vec{A}^{(\ell)}$ into a $2^\ell\times 2^\ell$ block matrix with blocks of size $(n/2^\ell)\times (n/2^\ell)$:
\begin{equation}
\label{eqn:Aell_decomp}
    \vec{A}^{(\ell)}
=
\begin{bmatrix}
    \vec{A}^{(\ell)}_{1,1} & \vec{A}^{(\ell)}_{1,2} & \cdots & \vec{A}^{(\ell)}_{1,d} \\
    \vec{A}^{(\ell)}_{2,1} & \vec{A}_{2,2} & \cdots & \vec{A}^{(\ell)}_{2,2^\ell} \\
    \vdots & \vdots && \vdots \\
    \vec{A}^{(\ell)}_{d,1} & \vec{A}^{(\ell)}_{2^\ell,2} & \cdots & \vec{A}^{(\ell)}_{2^\ell,2^\ell} \\
\end{bmatrix}.
\end{equation}
We will use the Generalized Nystr\"om Method to simultaneously approximate the relevant blocks
\begin{equation}
    \label{eqn:desired_factors}
\vec{A}^{(\ell)}_{2,1},~
\vec{A}^{(\ell)}_{1,2},~
\vec{A}^{(\ell)}_{4,3},~
\vec{A}^{(\ell)}_{3,4},~
,\ldots,~
\vec{A}^{(\ell)}_{2^\ell,2^\ell-1},
\vec{A}^{(\ell)}_{2^\ell-1,2^\ell},
\end{equation}
with low-rank matrices $\vec{H}^{(\ell)}_1\approx \vec{A}^{(\ell)}_{2,1}$, $\vec{H}^{(\ell)}_2\approx \vec{A}^{(\ell)}_{1,2}$, ..., $\vec{H}^{(\ell)}_d\approx \vec{A}^{(\ell)}_{2^\ell-1,2^\ell}$.
The low-rank matrices obtained will then be placed in the block tridiagonal matrix
\begin{equation}
    \label{eqn:Hell_decomp}
\vec{H}^{(\ell)} =
\operatorname{block-tridiag}\hspace{-.25em}
    \left(\hspace{-1.25em} 
    \begin{array}{c}
        \begin{NiceArray}[columns-width = 2.5em]{cccccc} \vec{H}^{(\ell)}_2 & \vec{0} & \vec{H}^{(\ell)}_4 & \vec{0} & \cdots & \vec{H}^{(\ell)}_{2^\ell} \end{NiceArray} \\
        \begin{NiceArray}[columns-width = 2.5em]{ccccccc} \vec{0} & \vec{0} &  \vec{0}& \vec{0}& \vec{0}& \cdots & \vec{0} \end{NiceArray} \\
        \begin{NiceArray}[columns-width = 2.5em]{cccccc} \vec{H}^{(\ell)}_1 & \vec{0} & \vec{H}^{(\ell)}_3 & \vec{0} & \cdots & \vec{H}^{(\ell)}_{2^\ell-1} \end{NiceArray}
    \end{array} 
    \hspace{-1.25em}\right),
\end{equation}
and the algorithm proceeds to the next level.

Note that the sketching matrices used by the peeling algorithm to recover the matrices in \cref{eqn:desired_factors} depends on the parity of the column index $j$: when the first index is even the second index is one smaller and when the first index is odd the second index is one larger.
To handle the two cases simultaneously in our analysis, we will introduce the following notation.
Fix a value $j\in\{1, \ldots, d\}$.
We will use ``$\pm$'' and ``$\mp$'' to indicate addition or subtraction depending on the parity of $j$:
\begin{equation}
\label{eqn:pm_def}
\pm 
:= \begin{cases}
    + & \text{$j$ odd} \\
    - & \text{$j$ even}
\end{cases}
,\qquad
\mp 
:= \begin{cases}
    - & \text{$j$ odd} \\
    + & \text{$j$ even}
\end{cases}.
\end{equation}
For instance example $j\pm 1$ means $j+1$ if $j$ is odd and $j-1$ if $j$ is even, and similarly $\xi_{j}^{\pm}$ means $\xi_{j}^+$ if $j$ is odd and $\xi_{j}^-$ if $j$ is even.

The diagonal blocks $\vec{A}^{(\ell)}_{j,j}$ may have large norm, and the sketching matrices used by the peeling algorithm are constructed explicitly to avoid touching these diagonal blocks. 
Assuming the algorithm has not accumulated much error, then besides the diagonal blocks and sub/super diagonal blocks we aim to recover, the other blocks in \cref{eqn:desired_factors}  will be on the scale of the best possible error.
In particular, as described in \cref{sec:peeling_intro}, if $\vec{A}$ is exactly HODLR, then these blocks are all zero.

At the final level $\ell=L$, we must also approximate the on-diagonal blocks.
Let $\widehat{\vec{A}}^{(L)}$ be the matrix after removing our approximation to the off-diagonal blocks of $\vec{A}^{(L)}$; i.e. $\widehat{\vec{A}}^{(L)} := \vec{A} - (\vec{H}^{(L)} + \cdots + \vec{H}^{(1)})$. 
Partition $\widehat{\vec{A}}^{(L)}$ into a $2^L\times 2^L$ block matrix with blocks of size $(n/2^L)\times (n/2^L)$:
\begin{equation}
\label{eqn:AL_decomp}
\widehat{\vec{A}}^{(L)}
=
\begin{bmatrix}
    \widehat{\vec{A}}^{(L)}_{1,1} & \widehat{\vec{A}}^{(L)}_{1,2} & \cdots & \widehat{\vec{A}}^{(L)}_{1,2^L} \\
    \widehat{\vec{A}}^{(L)}_{2,1} & \widehat{\vec{A}}_{2,2} & \cdots & \widehat{\vec{A}}^{(L)}_{2,2^L} \\
    \vdots & \vdots && \vdots \\
    \widehat{\vec{A}}^{(L)}_{2^L,1} & \widehat{\vec{A}}^{(L)}_{2^L,2} & \cdots & \widehat{\vec{A}}^{(L)}_{2^L,2^L} \\
\end{bmatrix}.
\end{equation}
We note that $\widehat{\vec{A}}^{(L)}_{i,j} = \vec{A}^{(L)}_{i,j}$ for all $i\neq \jp$; indeed, the approximation $\vec{H}^{(L)}$ to the off-diagonal low-rank blocks at level $L$ only updates the blocks $\vec{A}^{(L)}_{\jp,j}$.
We will approximate the blocks $\widehat{\vec{A}}^{(L)}_{j,j}$ with matrices $\widehat{\vec{H}}_j$ which we place in the matrix
\begin{equation}\label{eqn:Hhat_def}
    \widehat{\vec{H}}
=
\operatorname{block-diag}
    \Bigl(\hspace{-.5em}
        \begin{array}{cccccc} \widehat{\vec{H}}_1 & \widehat{\vec{H}}_2 & \widehat{\vec{H}}_3 & \cdots & \widehat{\vec{H}}_{2^L}
    \end{array} \hspace{-.5em}\Big).
\end{equation}

We now have the notation needed to fully describe \cref{alg:main}.
In \cref{sec:left_subspace,sec:lra_from_subspace,sec:diag_recover} we describe in more detail the main conceptual pieces of the algorithm and introduce additional notation used in our analysis.

\begin{algorithm}[ht]
\caption{Generalized Nystr\"om Method Peeling algorithm for HODLR approximation}\label{alg:main}
\fontsize{11}{15}\selectfont
\begin{algorithmic}[1]
\Procedure{GeneralizedNystr\"omPeeling}{$\vec{A},k,\sright,\tright,\sleft,\tleft$}
\State Set $L = \lceil \log_2(n/k) \rceil$
\Comment{Final level blocks of size at most $k$}
\For{$\ell=1,2,\ldots, L$}
\State Allocate and partition $\vec{H}^{(\ell)}$ as in \cref{eqn:Hell_decomp} \Comment{blocks of size $n/2^\ell\times n/2^\ell$}
\State $\bm{\xi}^+,\bm{\xi}^-,\vec{\Omega}^+,\vec{\Omega}^- \sim \operatorname{RandPerfGaussian}(n,2^\ell,\sright,\tright)$ \Comment{as in \cref{def:RandPerfGaussian}}
\State $\bm{\zeta}^+,\bm{\zeta}^-,\vec{\Psi}^+,\vec{\Psi}^- \sim \operatorname{RandPerfGaussian}(n,2^\ell,\sleft,\tleft)$ \Comment{as in \cref{def:RandPerfGaussian}}

\State Compute $\vec{A}^{(\ell)}\vec{\Omega}^\pm$
\Comment{$2\sright\tright$ matvecs with $\vec{A}$}
\State Compute $(\vec{\Psi}^\pm)^\T \vec{A}^{(\ell)}$
\Comment{$2\sleft\tleft$ matvecs with $\vec{A}^\T$}

\For{$j=1,2,\ldots, 2^\ell$}
\State Set $\rho$ such that $\xi_{j,\rho}^\pm = 1$
\State Set $\sigma$ such that $\zeta^{\mp}_{\jp,\sigma}=1$
\State $\vec{Q}^{(\ell)}_j = \operatorname{orth}(\vec{Y}^{(\ell)}_j)$\Comment{$\vec{Y}^{(\ell)}_j$ is $(\jp,\rho)$ block of $\vec{A}^{(\ell)}\vec{\Omega}^\pm$}
\State $\vec{X}^{(\ell)}_j = \big((\vec{\Psi}^\mp_{\jp})^\T \vec{Q}^{(\ell)}_j \big)^\dagger \vec{Z}^{(\ell)}_j$ \Comment{$\vec{Z}^{(\ell)}_j$ is $(\sigma,j)$ block of $(\vec{\Psi}^\pm)^\T \vec{A}$}
\State $\vec{H}^{(\ell)}_j = \vec{Q}^{(\ell)}_j\llbracket\vec{X}^{(\ell)}_j\rrbracket_k$  \Comment{$\vec{H}^{(\ell)}_j$ is $(\jp,j)$-th block of $\vec{H}^{(\ell)}$}
\EndFor
\EndFor
\State Allocate and partition $\widehat{\vec{H}}$ as in \cref{eqn:Hhat_def} \Comment{blocks of size $n/2^L\times n/2^L$}
\State $\widehat{\bm{\zeta}}^+,\widehat{\bm{\zeta}}^-,\widehat{\vec{\Psi}}^+,\widehat{\vec{\Psi}}^- \sim \operatorname{RandPerfGaussian}(n,2^L,\sleft,\tleft)$ \Comment{as in \cref{def:RandPerfGaussian}}
\State $\widehat{\bm{\zeta}} = \widehat{\bm{\zeta}}^+ + \widehat{\bm{\zeta}}^-$, $\widehat{\vec{\Psi}} = \widehat{\vec{\Psi}}^+ + \widehat{\vec{\Psi}}^-$
\State Compute $\widehat{\vec{\Psi}}^\T\widehat{\vec{A}}^{(L)}$ \Comment{$\sleft\tleft$ matvecs with $\vec{A}^\T$}
\For{$j=1,2,\ldots, 2^L$} 
\State Set ${\sigma}$ such that ${\widehat \zeta}_{j,\sigma}=1$
\State $\widehat{\vec{X}}_j = \big((\vec{\Psi}_{j})^\T \big)^\dagger\widehat{\vec{Z}}_j$  \Comment{$\widehat{\vec{Z}}_j$ is $(\sigma,j)$ block of $(\widehat{\vec{\Psi}})^\T \widehat{\vec{A}}^{(L)}$}
\State $\widehat{\vec{H}}_j = \widehat{\vec{X}}_j$  \Comment{$\widehat{\vec{H}}_j$ is $(j,j)$-th block of $\widehat{\vec{H}}$}
\EndFor
\State \Return $\widetilde{\vec{A}} = \vec{H}^{(1)} + \cdots + \vec{H}^{(L)} + \widehat{\vec{H}}$
\EndProcedure
\end{algorithmic}
\end{algorithm}

\subsection{Range approximation}
\label{sec:left_subspace}

We first describe how to obtain an approximate range $\vec{Q}^{(\ell)}_j$ for each $\vec{A}_{\jp,j}$ at level $\ell$.
Towards this end, let
\[
\bm{\xi}^+,\bm{\xi}^-,\vec{\Omega}^+,\vec{\Omega}^-
\sim \operatorname{RandPerfGaussian}(n,d,\sright,\sright)
\]
as defined in \cref{def:RandPerfGaussian}.
We will access $\vec{A}$ via the sketches $\vec{A}^{(\ell)}\vec{\Omega}^+$ and $\vec{A}^{(\ell)}\vec{\Omega}^-$, each of which can be computed using $\sright\tright$ matvecs with $\vec{A}$.

Fix $j$ and let $\rho = \rho(j)$ be the (unique) index such that $\xi_{j,\rho}^\pm = 1$.
We will try to recover a good approximation to the range of $\vec{A}^{(\ell)}_{\jp,j}$ using the information from the $(\jp,\rho)$ block of the sketch $\vec{A}\vec{\Omega}^\pm$.
This is illustrated in \cref{fig:peeling_err_perf}.
In particular, we will use the sketch
\begin{equation*}
    \vec{Y}^{(\ell)}_j
    := \vec{A}^{(\ell)}_{\jp,:} \vec{\Omega}^{\pm}_{:,\rho}
    = \sum_{i=1}^{d}  \vec{A}^{(\ell)}_{\jp,i} (\xi^{\pm}_{i,\rho}\vec{\Omega}_{i}).
\end{equation*}
It will be useful to write 
\begin{equation}
\label{eqn:Ej_def}
    \vec{Y}^{(\ell)}_j
    =\vec{A}^{(\ell)}_{\jp,j} \vec{\Omega}_j + \vec{E}^{(\ell)}_j
    ,\qquad 
    \vec{E}^{(\ell)}_j := \sum_{i\neq j} \xi^{\pm}_{i,\rho} \vec{A}^{(\ell)}_{\jp,i} \vec{\Omega}_{i}.
\end{equation}
We will then set 
\begin{equation}\label{eqn:V_def}
    \vec{Q}^{(\ell)}_{j} = \operatorname{orth}(\vec{Y}^{(\ell)}_j)
    = \operatorname{orth}(\vec{A}^{(\ell)}_{\jp,j} \vec{\Omega}_{j} + \vec{E}^{(\ell)}_j)
    ,
\end{equation}
which will still be an approximate top subspace of $\vec{A}^{(\ell)}_{\jp,j}$, provided the noise term $\vec{E}^{(\ell)}_j$ is not too large. 
In particular, note that if $\vec{A}$ is exactly HODLR, then $\vec{E}^{(\ell)}_j$ is zero.
This is because in this setting, $\vec{A}^{(\ell)}_{\jp,i} = \vec{0}$ except when $i = j$ or $i = j \pm 1$, and $\xi^\pm_{\jp,\rho} = 0$ so that there is no contribution from the on-diagonal block $\vec{A}^{(\ell)}_{\jp,\jp}$.

\subsection{Low-rank approximation from given subspace}
\label{sec:lra_from_subspace}

We now describe how to obtain the low-rank approximations $\vec{H}^{(\ell)}_j$ to each $\vec{A}^{(\ell)}_{\jp,j}$ at level $\ell$, given the approximate left subspaces $\vec{Q}^{(\ell)}_j$ computed by the procedure described in \cref{sec:left_subspace}. Let
\[
\bm{\zeta}^+,\bm{\zeta}^-,\vec{\Psi}^+,\vec{\Psi}^-
\sim \operatorname{RandPerfGaussian}(n,d,\sleft,\tleft)
\]
as defined in \cref{def:RandPerfGaussian}.
We will access $\vec{A}$ via the sketches $(\vec{A}^{(\ell)})^\T\vec{\Psi}^+$ and $(\vec{A}^{(\ell)})^\T\vec{\Psi}^-$, each of which can be computed using $\sleft\tleft$ matvecs with $\vec{A}^\T$.

Fix $j$ and let $\sigma = \sigma(j)$ be the (unique) index such that $\zeta^{\mp}_{\jp,\sigma}=1$.
We will try to recover a low-rank approximation of $\vec{A}^{(\ell)}_{\jp,j}$ whose column space is $\vec{Q}^{(\ell)}_j$ from the $(\sigma,j)$ block of the sketch $(\vec{\Psi}^{\mp})^\T \vec{A}^{(\ell)}$.
This is illustrated in \cref{fig:peeling_err_perf}.
In particular,  we will use the sketch
\begin{equation*}
    \vec{Z}^{(\ell)}_j := (\vec{\Psi}^{\mp}_{:,\sigma})^\T \vec{A}^{(\ell)}_{:,j}  
    = \sum_{i=1}^{d} (\zeta^\mp_{i,\sigma}\vec{\Psi}_{i})^\T \vec{A}^{(\ell)}_{i,j}.
\end{equation*}
Similar to above, we therefore have
\begin{equation}
\label{eqn:Fj_def}
    \vec{Z}^{(\ell)}_j = (\vec{\Psi}^{\mp}_{\jp})^\T\vec{A}^{(\ell)}_{\jp,j}  + \vec{F}^{(\ell)}_j
    ,\qquad 
    \vec{F}^{(\ell)}_j := \sum_{i\neq \jp} \zeta^{\mp}_{i,\sigma} (\vec{\Psi}_{i})^\T \vec{A}^{(\ell)}_{i,j}.
\end{equation}
Now, we obtain obtain the right factors by 
\begin{equation}\label{eqn:X_def}
    \vec{X}^{(\ell)}_j := \big((\vec{\Psi}^{\mp}_{\jp})^\T \vec{Q}^{(\ell)}_j \big)^\dagger \vec{Z}^{(\ell)}_j 
    = \argmin_{\vec{X}} \|  (\vec{\Psi}_{\jp})^\T \vec{Q}^{(\ell)}_j \vec{X} - \vec{Z}^{(\ell)}_j \|_\F.
\end{equation}
Finally, we obtain an approximation $\vec{H}^{(\ell)}_j = \vec{Q}^{(\ell)}_j \llbracket \vec{X}^{(\ell)}_j\rrbracket_k$ to $\vec{A}^{(\ell)}_{\jp,j}$.

If $\vec{F}^{(\ell)}_j$ is small, then $\vec{Q}^{(\ell)}_j\vec{X}^{(\ell)}_j$ will be nearly the best approximation to $\vec{A}^{(\ell)}_{\jp,j}$ with range equal to the column span of $\vec{Q}^{(\ell)}_j$. 
Similar to before, if $\vec{A}$ is exactly HODLR, the noise term $\vec{F}^{(\ell)}_j$ is zero.

\subsection{Recovering the diagonal blocks}
\label{sec:diag_recover}

The strategy we use to recover the diagonal blocks is similar to the off-diagonal blocks. 
However, since the blocks have at most $k$ columns, there is no need to obtain the left subspace.\footnote{This is also true for the off-diagonal blocks in levels where $n/2^\ell \leq \sright$.
In such cases $\vec{Q}_j^{(\ell)}$ will be (square) orthogonal, so it could be set to the identity a priori. 
We do not separate this case to simplify the notation in our analysis.} We let
\[
\widehat{\bm{\zeta}}^+,\widehat{\bm{\zeta}}^-,\widehat{\vec{\Psi}}^+,\widehat{\vec{\Psi}}^-
\sim \operatorname{RandPerfGaussian}(n,d,\sleft,\tleft)
\]
as defined in \cref{def:RandPerfGaussian}.
Next, define
\[
\widehat{\bm{\zeta}} = \widehat{\bm{\zeta}}^+ + \widehat{\bm{\zeta}}^-
,\qquad
\widehat{\vec{\Psi}} = \widehat{\vec{\Psi}}^+ + \widehat{\vec{\Psi}}^-
\]

Fix $j$ and let $\sigma = \sigma(j)$  be the (unique) index such that $\widehat{\zeta}_{j,\sigma} = 1$. 
We will try to recover an approximation to $\widehat{\vec{A}}^{(L)}_{j,j}$ using the $(\sigma,j)$  block of the sketch $(\vec{A}^{(L+1)})^\T \widehat{\vec{\Psi}}$.
In particular, we will use the sketch 
\begin{equation*}
    \widehat{\vec{Z}}_j := (\widehat{\vec{\Psi}}_{:,\sigma})^\T \widehat{\vec{A}}^{(L)}_{:,j}  
    = \sum_{i=1}^{d} \xi_{i,\sigma} (\vec{\Psi}_{i})^\T \widehat{\vec{A}}^{(L)}_{i,j} 
\end{equation*}
Similar to above, we therefore have 
\begin{equation}
    \widehat{\vec{Z}}_j = (\vec{\Psi}_{j})^\T\widehat{\vec{A}}^{(L)}_{j,j}  + \widehat{\vec{F}}_j
    ,\qquad 
    \widehat{\vec{F}}_j := \sum_{i\neq j}  \xi_{i,\sigma} (\vec{\Psi}_{i})^\T \widehat{\vec{A}}^{(L)}_{i,j}.
\end{equation}
We obtain the diagonal blocks by 
\begin{equation}\label{eqn:X_def_diag}
    \widehat{\vec{X}}_j := \big((\vec{\Psi}_{j})^\T  \big)^\dagger \widehat{\vec{Z}}_j 
    = \argmin_{\vec{H}} \|  (\vec{\Psi}_{j})^\T  \vec{H} - \widehat{\vec{Z}}_j \|_\F.
\end{equation}
Finally, we obtain an approximation $\widehat{\vec{H}}_j = \widehat{\vec{X}}_j$ to $\widehat{\vec{A}}^{(L)}_{j,j}$.

\section{Analysis}
\label{sec:analysis}

We are now prepared to prove our main accuracy guarantee for \cref{alg:main}, stated as \cref{thm:main_full}.
The simplified \cref{thm:main} stated in \Cref{sec:contributions} will follow as a direct corollary.

In \cref{section:perturbation_rsvd} we first prove \cref{thm:RSVD_perturb}, a deterministic perturbation bound for  RSVD and Generalized Nystr\"{o}m implemented with inexact matvec queries (see discussion in \cref{sec:aperturbationBound}). We then use this bound to prove \cref{thm:GN_perturb_exp}, is a probabilistic error bound that holds for Generalized Nystr\"{o}m implemented with random Gaussian sketches.
In \cref{sec:main_analysis} we use this bound to prove our main approximation guarantees.
In particular, our analysis applies \cref{thm:GN_perturb_exp} to each off-diagonal block at each level. 
Assuming we have obtained near-optimal low-rank approximations to the off-diagonal blocks at the previous levels, we show that we can obtain near-optimal low-rank approximations to the off-diagonal blocks at the current level.

\subsection{Perturbation bound for the RSVD}\label{section:perturbation_rsvd}

We begin with proving \cref{thm:RSVD_perturb}, which we restate below.

\medskip

\begin{restatethm}[\cref*{thm:RSVD_perturb}]
\RSVDperturb
\end{restatethm}

\begin{proof}
    Consider the matrix $\vec{C}:= (\vec{I} - \vec{Q} \vec{Q}^\T) \vec{B} + \vec{Q}(\vec{Q}^\T \vec{B} + \vec{E}_2)$. 
    The best rank-$k$ Frobenius-norm approximation to $\vec{C}$ with range contained in $\operatorname{range}(\vec{Q})$ is  $\vec{Q}\llbracket\vec{Q}^\T\vec{C}\rrbracket_k$ \cite[Theorem 3.5]{Gu:2015}.
    Using this, the triangle inequality, and the relations $\|\vec{B} - \vec{C}\|_{\F} = \|\vec Q \vec{E}_2\|_\F = \|\vec{E}_2\|_\F$ and $\vec{Q}^\T \vec{C} = \vec{Q}^\T \vec{B} + \vec{E}_2$ we obtain
    \begin{align}
        \|\vec{B} - \vec{Q}\llbracket\vec{Q}^\T\vec{B} + \vec{E}_2 \rrbracket_k\|_{\F}
        &=\|\vec{B} - \vec{Q}\llbracket\vec{Q}^\T\vec{C} \rrbracket_k\|_{\F}
        \nonumber\\&\leq \|\vec{B} - \vec{C}\|_\F + \|\vec{C} - \vec{Q}\llbracket\vec{Q}^\T\vec{C}\rrbracket_k\|_{\F} 
        \nonumber\\&\leq \|\vec{B} - \vec{C}\|_\F + \|\vec{C} - \vec{Q}\llbracket\vec{Q}^\T\vec{B}\rrbracket_k\|_{\F}  
        \nonumber\\&\leq 2\|\vec{B} - \vec{C}\|_\F + \|\vec{B} - \vec{Q}\llbracket\vec{Q}^\T\vec{B}\rrbracket_k\|_{\F}
        \nonumber\\&=2\|\vec{E}_2\|_\F + \|\vec{B} - \vec{Q}\llbracket\vec{Q}^\T\vec{B}\rrbracket_k\|_{\F}.\label{eq:inequality1}
    \end{align}
    Now let $\vec{M} := \vec{\Omega}_{\textup{top}}^{\dagger} \vec{V}_{\textup{top}}^\T$. Note that $\rank((\vec{B}\vec{\Omega} + \vec{E}_1)\vec{M}) \leq k$ and let $\vec{P} \in \R^{m_1 \times m_1}$ denote the orthogonal projector onto $\text{range}((\vec{B}\vec{\Omega} + \vec{E}_1)\vec{M})$.  
    Again, using \cite[Theorem 3.5]{Gu:2015} and the triangle inequality, we obtain
\begin{align}
    \|\vec{B} - \vec{Q} \llbracket\vec{Q}^\T \vec{B}\rrbracket_k\|_\F 
    &\leq \|(\vec{I} - \vec{P}) \vec{B}\|_\F
    \nonumber\\&\leq \|(\vec{I} - \vec{P}) \vec{B} \vec{\Omega} \vec{M}\|_\F + \|(\vec{I} - \vec{P}) \vec{B}(\vec{I} - \vec{\Omega} \vec{M})\|_\F 
    \nonumber\\&\leq \|(\vec{I} - \vec{P}) \vec{B} \vec{\Omega} \vec{M}\|_\F + \|\vec{B}(\vec{I} - \vec{\Omega} \vec{M})\|_\F.\label{eq:inequality2}
\end{align}
Now, following \cite[Chapter 5]{Connolly:2023} we can bound each term \eqref{eq:inequality2}. First note that
\begin{equation*}
    (\vec{I}-\vec{P}) \vec{B} \vec{\Omega} \vec{M} = (\vec{I}-\vec{P}) (\vec{B} \vec{\Omega} + \vec{E}_1) \vec{M} - (\vec{I}-\vec{P}) \vec{E}_1 \vec{M} = -(\vec{I}-\vec{P})\vec{E}_1\vec{M}.
\end{equation*}
Hence, 
\begin{equation}\label{eq:inequality3}
    \|(\vec{I}-\vec{P}) \vec{B} \vec{\Omega} \vec{M}\|_{\F} \leq \|\vec{E}_1\vec{M}\|_{\F} = \|\vec{E}_1 \vec{\Omega}_{\textup{top}}^{\dagger}\|_{\F}.
\end{equation}
For the second term, observe that since we have assumed $\rank(\vec{\Omega}_{\textup{top}})= k$, $\vec{V}_{\textup{top}} \vec{V}_{\textup{top}}^\T\vec{\Omega}\vec M  = \vec{V}_{\textup{top}} \vec{V}_{\textup{top}}^\T\vec{\Omega}(\vec{V}_{\textup{top}}^\T\vec{\Omega})^\dagger \vec{V}_{\textup{top}}^\T =\vec{V}_{\textup{top}} \vec{V}_{\textup{top}}^\T$. We thus have:
\begin{equation*}
    \vec{B}(\vec{I} - \vec{\Omega} \vec{M}) = \vec{B}\vec{V}_{\textup{bot}} \vec{V}_{\textup{bot}}^\T(\vec{I} - \vec{\Omega} \vec{M}) = \vec{B}\vec{V}_{\textup{bot}} \vec{V}_{\textup{bot}}^\T - \vec{B}\vec{V}_{\textup{bot}} \vec{V}_{\textup{bot}}^\T \vec{\Omega} \vec{M}.
\end{equation*}
By the Pythagorean theorem then we have
\begin{align*}
    \|\vec{B}(\vec{I} - \vec{\Omega} \vec{M})\|_\F & = \big(\|\vec{B}\vec{V}_{\textup{bot}} \vec{V}_{\textup{bot}}^\T\|_\F^2 + \|\vec{B}\vec{V}_{\textup{bot}} \vec{V}_{\textup{bot}}^\T \vec{\Omega} \vec{M}\|_\F^2\big)^{1/2} 
    \\
    & = \big(\|\vec{\Sigma}_{\textup{bot}}\|_\F^2 + \|\vec{\Sigma}_{\textup{bot}} \vec{\Omega}_{\textup{bot}} \vec{\Omega}_{\textup{top}}^{\dagger}\|_{\F}^2\big)^{1/2} \numberthis \label{eq:inequality4}.
\end{align*}
Combining \eqref{eq:inequality1} and \eqref{eq:inequality4} yields the desired inequality.
\end{proof}

\subsubsection{Probabilistic bounds}

To apply \Cref{thm:RSVD_perturb} to a setting with Gaussian sketching matrices, we make use of the following well-known fact about the Frobenius norm of Gaussian matrices and their pseudoinverses,  proofs of which can be found throughout the literature; see for instance \cite[Proposition A.5]{HalkoMartinssonTropp:2011}.
\begin{theorem}\label{thm:gaussian_expectations}
    Let $\vec{X}$ be $u\times v$ and $\vec{G}\sim\Gaussian(v,q)$ and $\vec{H} \sim \Gaussian(p,q)$ independently. 
    Then, if $q > p+1$,
    \begin{equation*}
        \EE\Bigl[ \| \vec{X} \vec{G}\vec{H}^\dagger \|_\F^2 \Bigr]
        =\|\vec{X}\|_\F^2 \cdot \EE\Bigl[\|\vec{H}^{\dagger}\|_\F^2\Bigr] =\frac{p}{q-p-1}\| \vec{X} \|_\F^2.
    \end{equation*}
\end{theorem}

We proceed with proving the following lemma, which relates to the error of the approximate projection step of the Generalized Nystr\"om Method, when implemented with approximate matvecs as they arise in the peeling algorithm -- see discussion in \Cref{sec:overviewPerturb}.

\begin{lemma}\label{thm:right_factor}
Let $\vec{B}\in\mathbb{R}^{m_1\times m_2}$, 
$\vec{Q} \in\mathbb{R}^{m_1\times \sright}$ with orthonormal columns, and
$\vec{N} \in\mathbb{R}^{q\times m_2}$ be fixed matrices and $\vec{\Psi}\sim \Gaussian(m_1,\sleft)$ and $\widetilde{\vec{\Psi}} \sim \Gaussian(q,\sleft)$ be independent. 
Then, provided $\sleft > \sright+1$,
\[
\EE\Bigl[ \| \vec{Q}^\T \vec{B} - ( \vec{\Psi}^\T\vec{Q})^\dagger ( \vec{\Psi}^\T \vec{B} + \widetilde{\vec{\Psi}}^\T \vec{N} )  \|_\F^2 \Bigr]
=
\frac{\sright}{\sleft-\sright-1} 
\Bigl(
\| \vec{B} - \vec{Q}\vec{Q}^\T \vec{B} \|_\F^2  + \| \vec{N}\|_\F^2 \Bigr).
\]
\end{lemma}

\begin{proof}
Extend $\vec{Q}$ to a square orthogonal matrix $[\vec{Q} ~ \widetilde{\vec{Q}}]$.
Write $\vec{\Psi}_1 = \vec{\Psi}^\T \vec{Q}$ and $\vec{\Psi}_2 = \vec{\Psi}^\T \widetilde{\vec{Q}}$. 
Since $\vec{Q}$ and $\widetilde{\vec{Q}}$ are orthogonal, $\vec{\Psi}_1  \sim \Gaussian(\sleft,\sright)$ and $\vec{\Psi}_2 \sim \Gaussian(\sleft,m_1-\sright)$ are independent.

Further, since $[\vec{Q} ~ \widetilde{\vec{Q}}]$ is orthogonal, $\vec{Q}\vec{Q}^\T + \widetilde{\vec{Q}} \widetilde{\vec{Q}}^\T$ is the identity, and 
\[
\vec{\Psi}^\T \vec{B}
=\vec{\Psi}^\T (\vec{Q}\vec{Q}^\T + \widetilde{\vec{Q}}\widetilde{\vec{Q}}^\T)\vec{B}
= \vec{\Psi}_1 \vec{Q}^\T \vec{B} + \vec{\Psi}_2 \widetilde{\vec{Q}}^\T \vec{B}.
\]
Therefore, recalling our definition of $\vec{\Psi}_1 = \vec{\Psi}^\T \vec Q$, 
\begin{equation*}
    \vec{\Psi}_1^\dagger ( \vec{\Psi}^\T \vec{B} + \vec{\widetilde{\Psi}}^\T \vec{N} )
    = 
    \vec{\Psi}_1^\dagger \vec{\Psi}_1 \vec{Q}^\T \vec{B} + \vec{\Psi}_1^\dagger \vec{\Psi}_2 \widetilde{\vec{Q}}^\T \vec{B}
    + \vec{\Psi}_1^\dagger \widetilde{\vec{\Psi}}^\T \vec{N}.
\end{equation*}
Since $\sleft\geq \sright$, with probability 1, $\vec{\Psi}_1^\dagger \vec{\Psi}_1 = \vec{I}$. Thus, rearranging the above equation, we find that 
\begin{align}\label{eq:cam1}
\vec{Q}^\T \vec{B} - \vec{\Psi}_1^\dagger(\vec{\Psi}^\T \vec{B} + \widetilde{\vec{\Psi}}^\T \vec{N}) = - \vec{\Psi}_1^{\dagger} \vec{\Psi}_2 \widetilde{\vec{Q}}^\T \vec{B} - \bm{\Psi}_1^{\dagger} \widetilde{\vec{\Psi}}^\T \vec{N}.
\end{align}
Next, as shown in \cite[Lemma A.2(i)]{PerssonBoulleKressner:2024}, for deterministic matrices $\vec{X},\vec{Y},\vec{Z}$ and a Gaussian matrix $\vec{\Psi}_2$, 
\begin{equation}\label{eqn:frobenius_sum}
    \EE\Bigl[\|\vec{X}\vec{\Psi}_2\vec{Y} + \vec{Z}\|_\F^2\Big]
= \|\vec{X}\|_\F^2 \| \vec{Y}\|_\F^2 + \|\vec{Z}\|_\F^2.
\end{equation}
Therefore, using that $\vec{\Psi}_1, \vec{\Psi}_2$ and $\widehat{\vec{\Psi}}$ are independent and applying \cref{eqn:frobenius_sum} and \cref{thm:gaussian_expectations} to bound the righthand side of \eqref{eq:cam1}, we obtain
\begin{align*}
\EE\Bigl[ \| \vec{Q}^\T \vec{B}-\vec{\Psi}_1^\dagger ( \vec{\Psi}^\T \vec{B} + \widetilde{\vec{\Psi}}^\T \vec{N} )
 \|_\F^2 \Bigr]
&= \EE\Bigl[ \| \vec{\Psi}_1^\dagger \vec{\Psi}_2 \widetilde{\vec{Q}}^\T \vec{B} + \vec{\Psi}_1^\dagger \widetilde{\vec{\Psi}}^\T \vec{N} \|_\F^2 \Bigr]
\\ & = \EE\Bigl[\EE\Bigl[ \| \vec{\Psi}_1^\dagger \vec{\Psi}_2 \widetilde{\vec{Q}}^\T \vec{B} + \vec{\Psi}_1^\dagger \widetilde{\vec{\Psi}}^\T \vec{N} \|_\F^2 \Big| \bm{\Psi}_1,\widetilde{\vec{\Psi}}\Bigr]\Bigr]
\\&= \EE\Bigl[  \| \vec{\Psi}_1^\dagger\|_\F^2 \|\widetilde{\vec{Q}}^\T \vec{B} \|_\F^2 + \|\vec{\Psi}_1^\dagger \widetilde{\vec{\Psi}} ^\T\vec{N} \|_\F^2 \Bigr].
\\&= \frac{\sright}{\sleft-\sright-1} \| \widetilde{\vec{Q}}^\T \vec{B} \|_\F^2 + \frac{\sright}{\sleft-\sright-1} \| \vec{N} \|_\F^2.
\end{align*}
Noting that $\| \widetilde{\vec{Q}}^\T \vec{B} \|_\F^2 = \|\vec{B} - \vec{Q}\vec{Q}^\T \vec{B}\|_{\F}^2$ yields the desired result. 
\end{proof}

Finally, we use \cref{thm:RSVD_perturb,thm:right_factor} to prove \cref{thm:GN_perturb_exp}.

\medskip

\begin{restatethm}[\cref*{thm:GN_perturb_exp}]
    \GNperturbexp
\end{restatethm}

\begin{proof}
Let $\vec{B}$ have SVD as partitioned in \cref{thm:RSVD_perturb} and define $\vec{\Omega}_{\textup{top}} = \vec{V}_{\textup{top}}^\T \vec{\Omega}$ and $\vec{\Omega}_{\textup{bot}} = \vec{V}_{\textup{bot}}^\T \vec{\Omega}$.  Note that, by the unitary invariance of Gaussian random matrices $\vec{\Omega}_{\textup{top}} \sim \Gaussian(k,\sright)$ and $\vec{\Omega}_{\textup{bot}}\sim\Gaussian(m_2-k,\sright)$ independently of one another and $\widetilde{\vec{\Omega}}$.
Note that $\vec{\Omega}_{\textup{top}}$ has rank $k$ with probability one, and \cref{thm:RSVD_perturb} asserts 
\begin{equation}
    \|\vec{B} - \vec{Q}\llbracket\vec{X} \rrbracket_k\|_{\F} 
    \leq \underbrace{\|\vec{M}\widetilde{\vec{\Omega}} \vec{\Omega}_{\textup{top}}^{\dagger}\|_\F}_{A}  + \underbrace{2\|\vec{Q}^\T\vec{B} - \vec{X} \|_\F }_{B}
    + \underbrace{\big(\|\vec{\Sigma}_{\textup{bot}}\|_\F^2 + \|\vec{\Sigma}_{\textup{bot}} \vec{\Omega}_{\textup{bot}} \vec{\Omega}_{\textup{top}}^{\dagger}\|_{\F}^2\big)^{1/2}}_{C}.\label{eqn:RSVD_perturb_plugin}
\end{equation}
We will set $E_1 = \EE[C^2]$ and $E_2$ as an upper bound for for $2\mathbb{E}[A^2 + B^2]$ which is itself an upper bound for $\EE[(A+B)^2]$.
Then, by the linearity of expectation and Cauchy--Schwarz,
\begin{align*}
    \EE[(A+B+C)^2]
    &= \EE[(A+B)^2] + \EE[C^2] + 2\EE[(A+B)C]
    \\&\leq \EE[(A+B)^2] + \EE[C^2] + 2\sqrt{\EE[(A+B)^2]\EE[C^2]}
    \\&\leq E_2 + E_1 + 2\sqrt{E_2E_1}.
\end{align*}
It now remains to compute $\EE[A^2]$, $\EE[B^2]$, and $\EE[C^2]$.

First, using \cref{thm:gaussian_expectations} we have that
\begin{align}
    \EE[A^2] &= \EE\Bigl[ \| \vec{M} \widetilde{\vec{\Omega}} \vec{\Omega}_{\textup{top}}^\dagger \|_\F^2 \Bigr]
    =  \frac{k}{\sright-k-1} \|\vec{M}\|_\F^2;\label{eq:Abound}\\
    \EE[C^2] &=\left(1+\frac{k}{\sright-k-1}\right)\|\vec{\Sigma}_{\textup{bot}}\|_\F^2 = \left(1+\frac{k}{\sright-k-1}\right)\|\vec{B} - \llbracket\vec{B}\rrbracket_k\|_\F^2.\label{eq:Cbound}
\end{align}

Since $\vec{Q} \vec Q^\T \vec{B}$ is the best Frobenius-norm approximation to $\vec{B}$ with range contained in $\operatorname{range}(\vec{Q})$, using  \cref{thm:RSVD_perturb} along with the inequality $(x+y)^2 \leq 2x^2 + 2y^2$ gives
\begin{equation}
    \|\vec{B} - \vec{Q}\vec{Q}^\T \vec{B}\|_\F^2 
    \leq \|\vec{B} - \vec{Q}\llbracket\vec{Q}^\T \vec{B} \rrbracket_k\|_\F^2 
    \leq 2 \left( \|\vec{M}\vec{\Omega}\vec{\Omega}_{\textup{top}}^{\dagger}\|_\F^2 + \|\vec{\Sigma}_{\textup{bot}}\|_\F^2 +  \|\vec{\Sigma}_{\textup{bot}}\vec{\Omega}_{\textup{bot}}\vec{\Omega}_{\textup{top}}^{\dagger}\|_\F^2 \right).  \label{eq:projectionbound}
\end{equation}
Applying \Cref{thm:right_factor}, \eqref{eq:projectionbound}, and \Cref{thm:gaussian_expectations} yields
\begin{align*}
    \EE[B^2] &= \frac{4\sright}{\sleft-\sright-1}\Bigr(\EE\Bigr[\|\vec{B}-\vec{Q}\vec{Q}^\T \vec{B}\|_\F^2\Bigl] + \|\vec{N}\|_\F^2\Bigl)\\
    & \leq \frac{4\sright}{\sleft-\sright-1}\left(2\EE\Bigl[\|\vec{M}\widetilde{\vec{\Omega}}\vec{\Omega}_{\textup{top}}^{\dagger}\|_\F^2\Bigr] + 2\|\vec{\Sigma}_{\textup{bot}}\|_\F^2 + 2 \EE\Bigl[\|\vec{\Sigma}_{\textup{bot}}\vec{\Omega}_{\textup{bot}}\vec{\Omega}_{\textup{top}}^{\dagger}\|_\F^2\Bigr] + \|\vec{N}\|_\F^2\right)\\
    & = \frac{4\sright}{\sleft-\sright-1}\left(\frac{2k}{\sright-k-1}\|\vec{M}\|_\F^2 + 2\|\vec{\Sigma}_{\textup{bot}}\|_\F^2 + \frac{2k}{\sright-k-1}\|\vec{\Sigma}_{\textup{bot}}\|_\F^2  + \|\vec{N}\|_\F^2\right).
\end{align*}
Our choices of $k$, $\sright$, and $\sleft$ ensure $\frac{k}{\sright-k-1} \leq 1$ and $\frac{\sright}{\sleft-\sright-1} \leq 1$. Hence,
\begin{equation}\label{eq:Bbound}
    \EE[B^2] \leq \frac{8k}{\sright-k-1}\|\vec{M}\|_\F^2 + \frac{16\sright}{\sleft-\sright-1}\|\vec{\Sigma}_{\textup{bot}}\|_\F^2  +  \frac{4\sright}{\sleft-\sright-1}\|\vec{N}\|_\F^2.
\end{equation}
Combining \eqref{eq:Abound} and \eqref{eq:Bbound} and noting $\|\vec{\Sigma}_{\textup{bot}}\|_\F^2 = \|\vec{B} -\llbracket\vec{B}\rrbracket_k\|_\F^2$ yields
\begin{equation*}
    2\EE[A^2 + B^2] \leq \frac{18k}{\sright-k-1}\|\vec{M}\|_\F^2 + \frac{32\sright}{\sleft-\sright-1}\|\vec{\Sigma}_{\textup{bot}}\|_\F^2  +  \frac{8\sright}{\sleft-\sright-1}\|\vec{N}\|_\F^2 :=E_2,
\end{equation*}
as required. 
\end{proof}

\subsection{Analysis of Algorithm \ref*{alg:main}}
\label{sec:main_analysis}

We are now prepared to analyze \cref{alg:main}.

\begin{definition}
    For each $\ell=1,2,\ldots, L$, let $\mathcal{F}_{\ell}$ be the sigma algebra representing the information known to the algorithm at the start of the $\ell$-th step.
\end{definition}

We begin by proving the key approximation guarantee for an off-diagonal low-rank block at level $\ell$.
This shows that, even in the presence of noise, we can obtain a near-optimal low-rank approximation to each of the off-diagonal blocks (and the on-diagonal blocks at the final level $\ell=L$).
This can be viewed as a version of \cref{thm:GN_perturb_exp} adapted to the errors appearing within the peeling algorithm.
\begin{lemma}\label{thm:level_l_gurantee}
Fix a rank parameter $k$ and let $\beta\in(0,1)$. 
Suppose $\tleft\geq 1$ and $\sright$, $\tleft$, and $\sright$ are such that
\begin{align*}
\frac{k}{\sright-k-1}\leq \frac{\beta}{30}
,\qquad
\frac{k}{\sright-k-1}\cdot \frac{1}{\tright} \leq \frac{\beta^2}{30^2}
,\qquad 
\frac{\sright}{\sleft-\sright-1} \leq \frac{\beta^2}{30^2}.
\end{align*}
Then for each $\ell=1,2,\ldots, L$ the off-diagonal low-rank blocks obtained by \cref{alg:main} are nearly optimal in the sense that 
    \begin{align*}
        \EE\Bigl[
        \|\vec{A}^{(\ell)}_{\jp,j} - \vec{Q}^{(\ell)}_j \llbracket \vec{X}^{(\ell)}_j \rrbracket_k \|_\F^2 \Big| \mathcal{F}_{\ell}\Bigr]
        &\leq \| \vec{A}^{(\ell)}_{\jp,j} - \llbracket \vec{A}^{(\ell)}_{\jp,j} \rrbracket_k \|_\F^2 
        + \beta E^{(\ell)}_j .
    \end{align*}
    where
    \[
    E^{(\ell)}_j := \max\Bigg\{ 
    \| \vec{A}^{(\ell)}_{\jp,j} - \llbracket \vec{A}^{(\ell)}_{\jp,j} \rrbracket_k \|_\F^2 
    ,
    \sum_{\substack{i=1\\i\neq j,\jp}}^{2^\ell} \frac{1}{2}\|\vec{A}^{(\ell)}_{i,j}\|_\F^2
    ,
    \sum_{\substack{i=1\\i\neq j,\jp}}^{2^\ell}  \frac{1}{2}\|\vec{A}^{(\ell)}_{\jp,i}\|_\F^2
    \Bigg\}.
    \]
    Moreover, at the final level $\ell=L$, the on-diagonal blocks obtained by \cref{alg:main} are nearly exact in  that 
    \[
    \EE\Bigl[
    \|\widehat{\vec{A}}^{(L)}_{j,j} - \widehat{\vec{X}}_j  \|_\F^2 \Big| \mathcal{F}_{L+1}\Bigr]
    \leq \beta  \sum_{\substack{i=1\\i\neq j}}^{2^L} \|\widehat{\vec{A}}^{(L)}_{i,j}\|_\F^2 .
    \]

\end{lemma}

\begin{proof}
We begin with the first result. Let $d = 2^{\ell}$ and let $\vec{\Omega}^{\pm}$ and $\vec{\Psi}^{\pm}$ be the sketches used at level $\ell$ of \Cref{alg:main}. Fix $j$ and recall that $\rho$ and $\sigma$ are the unique indices so that $\xi_{j,\rho}^{\pm} = 1$ and $\zeta_{j\pm1,\sigma}^{\mp} = 1$. Define 
\begin{align*}
    \vec{M} &=
\begin{bmatrix}
    \xi^\pm_{1,\rho}\vec{A}^{(\ell)}_{\jp,1} & \cdots & \xi^\pm_{j-1,\rho}\vec{A}^{(\ell)}_{\jp,j-1} & \xi^\pm_{j+1,\rho}\vec{A}^{(\ell)}_{\jp,j+1} & \cdots & \xi^\pm_{d,\rho}\vec{A}^{(\ell)}_{\jp,d},
\end{bmatrix}
\\
\vec{N} &=
\begin{bmatrix}
    \zeta^\mp_{1,\sigma} \vec{A}^{(\ell)}_{1,j} & \cdots & \zeta^\mp_{\jp-1,\sigma} \vec{A}^{(\ell)}_{\jp-1,j} & \zeta^\mp_{\jp+1,\sigma} \vec{A}^{(\ell)}_{\jp+1,j} & \cdots & \zeta^\mp_{d,\sigma} \vec{A}^{(\ell)}_{d,j},
\end{bmatrix}
\end{align*}
and the random Gaussian matrices
\begin{align*}
    \widetilde{\vec{\Omega}} &= 
    \begin{bmatrix} (\vec{\Omega}_1)^\T & \cdots & (\vec{\Omega}_{j-1})^\T & (\vec{\Omega}_{j+1})^\T & \cdots & (\vec{\Omega}_d)^\T \end{bmatrix}^\T
    \\
    \widetilde{\vec{\Psi}} &= 
    \begin{bmatrix} (\vec{\Psi}_1)^\T & \cdots & (\vec{\Psi}_{j\pm1-1})^\T & (\vec{\Psi}_{j\pm1+1})^\T & \cdots & (\vec{\Psi}_d)^\T \end{bmatrix}^\T.
\end{align*}
Then we can write 
\[
\vec{E}_j^{(\ell)} = \vec{M}\widetilde{\vec{\Omega}}
,\qquad
\vec{F}_j^{(\ell)} = \widetilde{\vec{\Psi}}^\T \vec{N},
\]
where $\vec{E}_j^{(\ell)}$ and $\vec{F}_j^{(\ell)}$ are are the additive perturbation terms as defined in \cref{eqn:Ej_def,eqn:Fj_def}. 
Furthermore, note that $\widetilde{\vec{\Omega}}$ and $\widetilde{\vec{\Psi}}$ are independent of $\vec{\Omega}_j$ and $\vec{\Psi}_{j\pm1}$. 
Recall that $\vec{Q}_{j}^{(\ell)}$ is an orthonormal basis for $\mathrm{range}(\vec{A}_{j\pm1,j} \vec{\Omega}_j + \vec{E}_j^{(\ell)})$ and $\vec{X}_{j}^{(\ell)} = (\vec{\Psi}_{j\pm 1}^\T \vec{Q}_{j}^{(\ell)})^{\dagger}(\vec{\Psi}_j^\T\vec{A}_{j\pm 1,j} + \widetilde{\vec{\Psi}}^\T \vec{N})$. 

The specified conditions on $k$, $\sright$, and $\sleft$ ensure that $\sleft \geq 2\sright+1$ and $\sright\geq 2k+1$.
Therefore, \Cref{thm:GN_perturb_exp} yields
\begin{equation}
    \label{eqn:bound_E1E2}
    \EE\Bigl[
    \|\vec{A}^{(\ell)}_{\jp,j} - \vec{Q}^{(\ell)}_j \llbracket \vec{X}^{(\ell)}_j \rrbracket_k \|_\F^2 \Big| \mathcal{F}_{\ell},\bm{\xi},\bm{\zeta} \Bigr]
    \leq E_1 + E_2 + 2\sqrt{E_1E_2},
\end{equation}
where
\begin{align*}
    E_1 &:= \left(1+\frac{k}{\sright-k-1} \right) \| \vec{A}^{(\ell)}_{\jp,j} - \llbracket \vec{A}^{(\ell)}_{\jp,j} \rrbracket_k \|_\F^2
    ,
    \\
    E_2 &:= \frac{18k}{\sright-k-1} \|\vec{M}\|_\F^2 + \frac{8\sright}{\sleft-\sright-1} \|\vec{N}\|_{\F}^2 + \frac{32\sright}{\sleft-\sright-1} \| \vec{A}^{(\ell)}_{\jp,j} - \llbracket \vec{A}^{(\ell)}_{\jp,j} \rrbracket_k \|_\F^2
    .
\end{align*}
By our choice of $\bm{\xi}^\pm$, the $\xi^\pm_{i,\rho}$ are independent for different $i$ and $\EE[\xi^\pm_{i,\rho}]$ is zero or $1/\tright$ depending on the parity of $\rho$. Thus, $\EE[\xi^\pm_{i,\rho}]\leq 1/\tright$. 
Furthermore, note that $\EE[\xi^{\pm}_{j\pm1,\rho}] = 0$ by construction.
Thus, 
\begin{align*}
    \EE\Bigl[\|\vec{M}\|_\F^2\Big] 
    &= \sum_{i\neq j} \EE[\xi^\pm_{i,\rho}] \|\vec{A}^{(\ell)}_{\jp,i}\|_\F^2 
    \leq  \sum_{\substack{i=1\\i\neq j,\jp}}^{d}  \|\vec{A}^{(\ell)}_{\jp,i}\|_\F^2.
\intertext{By a similar argument,}
    \EE\Bigl[\|\vec{N}\|_\F^2\Big] 
    &= \sum_{i\neq \jp} \EE[\zeta^{\mp}_{i,\sigma}] \|\vec{A}^{(\ell)}_{i,j}\|_\F^2 
    \leq \frac{1}{\tleft} \sum_{\substack{i=1\\i\neq j,\jp}}^{d} \|\vec{A}^{(\ell)}_{i,j}\|_\F^2 .
\end{align*}
Hence, by the law of total expectation, the Cauchy-Schwarz inequality, and \eqref{eqn:bound_E1E2} we have
\begin{equation}
    \label{eqn:bound_E1E2v2}
    \EE\Bigl[
    \|\vec{A}^{(\ell)}_{\jp,j} - \vec{Q}^{(\ell)}_j \llbracket \vec{X}^{(\ell)}_j \rrbracket_k \|_\F^2 \Big| \mathcal{F}_{\ell} \Bigr]
    \leq E_1 + E_2' + 2\sqrt{E_1E_2'},
\end{equation}
where
\begin{align*}
    E_2' &:= \frac{18k}{\sright-k-1}\cdot \frac{1}{\tright} \sum_{\substack{i=1\\i\neq j,\jp}}^{d} \|\vec{A}^{(\ell)}_{\jp,i}\|_\F^2 + \frac{8\sright}{\sleft-\sright-1}\cdot \frac{1}{\tleft}\sum_{\substack{i=1\\i\neq j,\jp}}^{d} \|\vec{A}^{(\ell)}_{i,j}\|_\F^2 + \frac{32\sright}{\sleft-\sright-1} \| \vec{A}^{(\ell)}_{\jp,j} - \llbracket \vec{A}^{(\ell)}_{\jp,j} \rrbracket_k \|_\F^2
    .
\end{align*}
Using that
\begin{equation*}
    \max\Biggl\{\| \vec{A}^{(\ell)}_{\jp,j} - \llbracket \vec{A}^{(\ell)}_{\jp,j} \rrbracket_k \|_\F^2, \frac{1}{2}\sum_{\substack{i=1\\i\neq j,\jp}}^{d}  \|\vec{A}^{(\ell)}_{\jp,i}\|_\F^2, \frac{1}{2} \sum_{\substack{i=1\\i\neq j,\jp}}^{d} \|\vec{A}^{(\ell)}_{i,j}\|_\F^2\Bigg\} \leq E_j^{(\ell)}
\end{equation*}
and our assumptions on $\tleft$, $\tright$, $\sleft$, and $\sright$ allows us to get bounds 
\begin{align*}
    \frac{k}{\sright-k-1} \| \vec{A}^{(\ell)}_{\jp,j} - \llbracket \vec{A}^{(\ell)}_{\jp,j} \rrbracket_k \|_\F^2
    &\leq \frac{\beta}{30} E_j^{(\ell)}
    &
    \frac{18k}{\sright-k-1}\cdot \frac{1}{\tright} \sum_{\substack{i=1\\i\neq j,\jp}}^{d} \|\vec{A}^{(\ell)}_{\jp,i}\|_\F^2 
    &\leq 36 \frac{\beta^2}{30^2}E_j^{(\ell)}
    \\
    \frac{8\sright}{\sleft-\sright-1}\cdot \frac{1}{\tleft}\sum_{\substack{i=1\\i\neq j,\jp}}^{d} \|\vec{A}^{(\ell)}_{i,j}\|_\F^2 
    &\leq 16 \frac{\beta^2}{30^2}E_j^{(\ell)}
    &
    \frac{32\sright}{\sleft-\sright-1} \| \vec{A}^{(\ell)}_{\jp,j} - \llbracket \vec{A}^{(\ell)}_{\jp,j} \rrbracket_k \|_\F^2
    &\leq 32 \frac{\beta^2}{30^2}E_j^{(\ell)}.
\end{align*}
Therefore, since $E_1\leq 2 E_j^{(\ell)}$ and $1/30 + (36+16+32)/30^2 + 2\sqrt{2(36+16+32)/30^2}\approx 0.991 \leq 1$, we get a bound
\begin{equation}\label{eq:finalineq}
    E_1 + E_2' + 2\sqrt{E_1E_2'} \leq \| \vec{A}^{(\ell)}_{\jp,j} - \llbracket \vec{A}^{(\ell)}_{\jp,j} \rrbracket_k \|_\F^2 + \beta E_j^{(\ell)}.
\end{equation}
Combining \eqref{eq:finalineq} with \eqref{eqn:bound_E1E2v2} yields the first inequality.

The proof of the second result is similar to the analysis of the Generalized Nystr\"om method \cref{thm:GN_perturb_exp}.
First define
\begin{equation*}
    \widehat{\vec{N}} = \begin{bmatrix} \widehat{\zeta}_{1,\rho} \widehat{\vec{A}}_{1,j}^{(L)} & \cdots & \widehat{\zeta}_{j-1,\rho} \widehat{\vec{A}}_{j-1,j}^{(L)} & \widehat{\zeta}_{j+1,\rho} \widehat{\vec{A}}_{j+1,j}^{(L)} & \cdots & \widehat{\zeta}_{2^L,\rho} \widehat{\vec{A}}_{2^L,j}^{(L)}\end{bmatrix}.
\end{equation*}
and the random Gaussian matrix
\begin{equation*}
    \widetilde{\vec{\Psi}} = \begin{bmatrix} (\widehat{\vec{\Psi}}_1)^\T & \cdots & (\widehat{\vec{\Psi}}_{j-1})^\T & (\widehat{\vec{\Psi}}_{j+1})^\T & \cdots & (\widehat{\vec{\Psi}}_{2^L})^\T \end{bmatrix}^\T.
\end{equation*}
Similar to above we have 
\[
\EE\bigl[\|\widehat{\vec{N}}\|_\F^2\bigr] 
\leq \frac{1}{\tleft} \sum\limits_{\substack{i=1\\i\neq j} }^{2^L}\|\widehat{\vec{A}}_{i,j}^{(L)}\|_\F^2.
\]
Since $\sleft \geq \sright>2k+1$ and $k \geq n/2^L$ we can apply \cref{thm:gaussian_expectations} to obtain
\begin{align*}
    \EE\Bigl[
    \|\widehat{\vec{A}}^{(L)}_{j,j} - \widehat{\vec{H}}_j \|_\F^2 \Big| \mathcal{F}_{L+1}\Bigr] 
    &= \EE\Bigl[\|(\widehat{\vec{\Psi}}_j^\T)^{\dagger}\widetilde{\vec{\Psi}}^\T \widehat{\vec{N}}\|_{\F}^2 \Big| \mathcal{F}_{L+1} \Big]\\
    &= \frac{n/2^L}{\sleft - n/2^L-1} \mathbb{E}\bigl[\|\widehat{\vec{N}}\|_{\F}^2\bigr]\\
    & \leq  \frac{k}{\sright - k-1} \frac{1}{\tleft}\sum_{\substack{i=1\\i\neq j}}^{2^L} \|\widehat{\vec{A}}^{(L)}_{i,j}\|_\F^2\\
    &\leq \frac{\beta}{30}\sum_{\substack{i=1\\i\neq j}}^{2^L} \|\widehat{\vec{A}}^{(L)}_{i,j}\|_\F^2 
    \leq\beta  \sum_{\substack{i=1\\i\neq j}}^{2^L} \|\widehat{\vec{A}}^{(L)}_{i,j}\|_\F^2,
\end{align*}
as required.
\end{proof}

\pagebreak[3]

Finally, we prove our main theorem. 
\begin{theorem}\label{thm:main_full}
Fix a rank parameter $k$ and let $\beta\in(0,1)$. Suppose $\tleft\geq 1$ and $\sright$, $\tleft$, and $\sright$ are such that
\begin{align*}
\frac{k}{\sright-k-1}\leq \frac{\beta}{30}
,\qquad
\frac{k}{\sright-k-1}\cdot \frac{1}{\tright} \leq \frac{\beta^2}{30^2}
,\qquad 
\frac{\sright}{\sleft-\sright-1} \leq \frac{\beta^2}{30^2}.
\end{align*}
Let $L=\lceil \log_2(n/k) \rceil$.
Then, using $2L\sright\tright$ products with $\vec{A}$ and $(2L+1)\sleft\tleft$ products with $\vec{A}^\T$, \cref{alg:main} outputs a $\HODLR(k)$ matrix $\widetilde{\vec{A}}$ such that
    \[
    \mathbb{E}\Bigl[\| \vec{A} - \widetilde{\vec{A}} \|_\F^2\Bigr] \leq (1+\beta)^{L+1} \cdot \min_{\vec{H}\in\HODLR(k)} \|\vec{A} - \vec{H} \|_\F^2.
    \]
\end{theorem}

\Cref{thm:main} is an immediate consequence of \cref{thm:main_full}.
\begin{nproof}[Proof of \cref{thm:main}]
We can instantiate  \cref{thm:main_full} with
\begin{equation}\label{eqn:main_thm_params}
\sright = O(k/\beta)
,\quad
\tright = O(1/\beta)
,\quad
\sleft = O(\sright/\beta^2) = O(k/\beta^3)
,\quad
\tleft = 1.
\end{equation}
Write $\ERR^2:=\| \vec{A} - \widetilde{\vec{A}}\|_\F^2$ and $\OPT^2 := \min_{\vec{H}\in\HODLR(k)} \|\vec{A} - \vec{H} \|_\F^2$.
Rearanging \cref{thm:main} we see $\EE[\ERR^2 - \OPT^2] \leq ((1 + \beta)^{L+1}-1) \OPT^2$.
Then, since $\ERR^2 - \OPT^2 \geq 0$, we can apply Markov's inequality and obtain a bound
\[
\PP\Bigl[ \ERR^2 - \OPT^2 > 100 \big((1 + \beta)^{L+1}-1\big) \OPT^2\Big] \leq \frac{1}{100}.
\]
Note that $1+100((1+\beta)^{L+1}-1)\leq (1+100\beta)^{L+1}$ for all $\beta,L > 0$. 
Therefore, $\ERR^2 \leq (1+100\beta)^{L+1} \OPT^2$, except with probability at most $1/100$.
Taking a square root and adjusting $\beta$ by a constant factor gives the result.
\end{nproof}

The parameters in \cref{eqn:main_thm_params} are chosen to minimize the beta dependencies in the resulting bound.
Another interesting parameter setting is 
\begin{equation}\label{eqn:GNpeel_params}
    \sright = O(k/\beta^2)
,\quad
\tright = 1
,\quad
\sleft = O(k/\beta^4)
,\quad
\tleft = 1,
\end{equation}
which corresponds to an implementation of the standard peeling algorithm with the Generalized Nystr\"om Method (i.e. does not use randomly perforated sketching) attaining the same $(1+\beta)^{L+1}$ approximation factor.
In the experiments in \cref{sec:examples}, we find that it does not seem necessary to set $\sleft = O(k/\beta^4)$, and that $\sleft = O(k/\beta^2)$ (with $\tleft=1$) seems sufficient.

\begin{nproof}[Proof of \cref{thm:main_full}]
First, note that at each of the $L$ levels we use $2\sright\tright$ matrix-vector products with $\vec{A}$ and $2\sleft\tleft$ matrix-vector products with $\vec{A}^\T$.
When recovering the diagonal we use an additional $\sleft\tleft$ products with $\vec{A}^\T$.
This results in the specified cost.

We now analyze the error.
For each $\ell = 1,2, \ldots, L$ partition $\vec{A}^{(\ell)}$ as in \cref{eqn:Aell_decomp} and $\vec{H}^{(\ell)}$ as in \cref{eqn:Hell_decomp}.
Define, 
\[
\OPT(\ell) := \bigg( \sum_{j=1}^{2^\ell} \| \vec{A}^{(\ell)}_{\jp,j} - \llbracket \vec{A}^{(\ell)}_{\jp,j} \rrbracket_k \|_\F^2 \bigg)^{1/2}.
\]
By the definition of HODLR matrices the on-diagonal blocks at level $L$ are of size at most $k$.
Therefore, since the different levels are disjoint,
\begin{equation}\label{eqn:opt_expansion}
\min_{\vec{H}\in\HODLR(k)} \| \vec{A} - \vec{H} \|_\F^2
= \OPT(1)^2 + \cdots + \OPT(L)^2.
\end{equation}
Next, define the error incurred in the off-diagonal low-rank blocks at level $\ell$,
\[
\ERR(\ell) := \bigg( \sum_{j=1}^{2^\ell} \| \vec{A}^{(\ell)}_{\jp,j} - \vec{H}^{(\ell)}_j \|_\F^2 \bigg)^{1/2}.
\]
Define also the error of the on-diagonal blocks at the final level,
\[
\widehat{\ERR} := \bigg( \sum_{j=1}^{2^L} \| \vec{A}^{(L)}_{j,j} - \widehat{\vec{H}}_j \|_\F^2 \bigg)^{1/2}.
\]
Since the approximations at each level are disjoint (see \cref{eqn:final_approx}) we have
\begin{equation}\label{eqn:err_expansion}    
    \|\vec{A} - \widetilde{\vec{A}}\|_\F^2
    = \ERR(1)^2 + \cdots + \ERR(L)^2 + \widehat{\ERR}^2.
\end{equation}
We claim that it suffices to prove that for each $\ell=1,2,\ldots, L$,
\begin{equation}\label{eqn:level_l_goal}
\EE\bigl[ \ERR(\ell)^2 \big]
\leq \OPT(\ell)^2 +  \beta \cdot(1+\beta)^{\ell-1} \cdot \bigl( \OPT(1)^2 + \cdots +\OPT(\ell)^2 \bigr).
\end{equation}
and that at the final level 
\begin{equation}\label{eqn:level_q_diag_goal}
\EE\bigl[ \widehat{\ERR}^2 \big]
\leq  \beta \cdot(1+\beta)^{L} \cdot \bigl( \OPT(1)^2 + \cdots + \OPT(L)^2 \bigr).
\end{equation}
Indeed, since $\ERR(\ell')^2\geq 0$, together \cref{eqn:level_l_goal,eqn:level_q_diag_goal} imply that 
\begin{align*}
\EE\bigl[ \ERR(1)^2 + \cdots + \ERR(L)^2 + \widehat{\ERR}^2 \big]
&\leq \Big(1+\beta\cdot\big(1 + (1+\beta) + \cdots + (1+\beta)^{L}\big) \Big)\cdot \Bigl( \OPT(1)^2 + \cdots + \OPT(L)^2 \Bigr)
\\&= ( 1+\beta )^{L+1}\cdot \Bigl( \OPT(1)^2 + \cdots + \OPT(L)^2 \Bigr).
\end{align*}
In light of \cref{eqn:err_expansion,eqn:opt_expansion}, this is the desired result.

We begin by proving \cref{eqn:level_l_goal} by induction. First note that \cref{eqn:level_l_goal} at level $\ell = 1$ we have incurrred no error from previous levels and \cref{thm:level_l_gurantee} gives $\EE[\ERR(1)] \leq (1+\beta) \OPT(1)$. 
Now suppose the analog of \cref{eqn:level_l_goal} holds for each level up to $\ell-1$.
Then, since $\ERR(\ell')^2\geq 0$,
\begin{align}
\EE\bigl[\ERR(1)^2 + \cdots + \ERR(\ell-1)^2  \bigr]
&\leq \big( 1 + \beta\cdot(1 + (1+\beta) + \cdots + (1+\beta)^{\ell-2}) \big) \cdot \big(\OPT(1)^2 + \cdots + \OPT(\ell-1)^2\big)
\nonumber\\&=(1+\beta)^{\ell-1} \cdot \big(\OPT(1)^2 + \cdots + \OPT(\ell-1)^2\big).
\label{eqn:err1_to_errlm1}
\end{align}
Applying \cref{thm:level_l_gurantee} for each $j=1,2, \ldots 2^\ell$ and summing over $j$ we find 
\begin{equation}
\label{eqn:errl}
    \EE\bigl[\ERR(\ell)^2 \big]
    = \EE\bigl[\EE\bigl[\ERR(\ell)^2\big| \mathcal{F}_{\ell} \big]\big]
    \leq \EE\Biggl[\OPT(\ell)^2 + \beta \sum_{j=1}^{2^\ell} E^{(\ell)}_j \Bigg].
\end{equation}
Observe that the sum over all blocks of $\vec{A}^{(\ell)}$ except the on-diagonal blocks and the off-diagonal low-rank blocks at level $\ell$ is expressed as 
\[
    \sum_{j=1}^{2^\ell}\sum_{\substack{i=1\\i\neq j,\jp}}^{2^\ell} \|\vec{A}^{(\ell)}_{i,j}\|_\F^2
    =
    \sum_{j=1}^{2^\ell}\sum_{\substack{i=1\\i\neq j,\jp}}^{2^\ell} \|\vec{A}^{(\ell)}_{\jp,i}\|_\F^2.
\]
Therefore,
\begin{align}
    \sum_{j=1}^{2^\ell} E^{(\ell)}_j
    &\leq \sum_{j=1}^{2^\ell}\Bigg(
    \| \vec{A}^{(\ell)}_{\jp,j} - \llbracket \vec{A}^{(\ell)}_{\jp,j} \rrbracket_k \|_\F^2 
    +
    \sum_{\substack{i=1\\i\neq j,\jp}}^{2^\ell} \frac{1}{2}\|\vec{A}^{(\ell)}_{i,j}\|_\F^2
    +
    \sum_{\substack{i=1\\i\neq j,\jp}}^{2^\ell} \frac{1}{2}\|\vec{A}^{(\ell)}_{\jp,i}\|_\F^2
    \Bigg)
    \nonumber\\&= \OPT(\ell)^2 + \sum_{j=1}^{2^\ell}
    \sum_{\substack{i=1\\i\neq j,\jp}}^{2^\ell} \|\vec{A}^{(\ell)}_{i,j}\|_\F^2
    \nonumber\\&= \OPT(\ell)^2 + \big(\ERR(1)^2 + \cdots + \ERR(\ell-1)^2\big).
    \label{eqn:Elj_sum}
\end{align}
In the final equality above we have used that the $(i,j)$ blocks for $i\neq j,\jp$ contain exactly the errors made at previous levels. 
Hence, taking the expectation of \cref{eqn:Elj_sum}, using \cref{eqn:err1_to_errlm1}, and plugging into \cref{eqn:errl} gives
\begin{align*}
    \EE\bigl[\ERR(\ell)^2 \big]
    &\leq \OPT(\ell)^2 + \beta \big( \OPT(\ell)^2 + (1+\beta)^{\ell-1} \big(\OPT(1)^2 + \cdots + \OPT(\ell-1)^2\big) \big)
    \\&\leq \OPT(\ell)^2 + \beta(1+\beta)^{\ell-1} \big(\OPT(1)^2 + \cdots + \OPT(\ell)^2\big),
\end{align*}
which gives \cref{eqn:level_l_goal}.

The proof of \cref{eqn:level_q_diag_goal} follows analogously. By \cref{thm:level_l_gurantee} we have
\begin{equation*}
    \EE\bigl[\widehat{\ERR}^2] \leq \beta \sum\limits_{j=1}^{2^L}\sum\limits_{\substack{i=1\\i \neq j}}^{2^L} \EE\bigl[\|\widehat{\vec{A}}_{i,j}^{(L)}\|_{\F}^2\bigr] = \beta \sum\limits_{\ell = 1}^{L} \EE\big[\ERR(\ell)^2] \leq \beta(1+\beta)^L(\OPT(1)^2+\cdots+\OPT(L)^2),
\end{equation*}
where we used that $\sum_{j=1}^{2^L}\sum_{i=1,i \neq j}^{2^L} \|\widehat{\vec{A}}_{i,j}^{(L)}\|_{\F}^2$ is precisely the errors made at previous levels and \eqref{eqn:err1_to_errlm1}. This gives the desired inequality. 
\end{nproof}

\section{Lower Bounds}
\label{sec:lower_bounds}
In this section, we prove that our main result on HOLDR matrix approximation (\Cref{thm:main}) is close to optimal in terms of the number of matrix-vector products needed to solve \Cref{prob:approx}. 

\subsection{Lower bound for exact recovery}
We begin with a simple lower bound in the setting where $\vec{A}$ is exactly HODLR, in which case solving \Cref{prob:approx} requires \emph{exactly} recovering $\vec{A}$. 
For this setting, we can appeal to prior work by Halikias and Townsend \cite{HalikiasTownsend:2023} on the query complexity of recovering matrices from \emph{linearly parameterized matrix families}. 
In particular, any set $\vec{B}_1, \ldots, \vec{B}_m$ of \emph{linearly independent} base matrices induces a linearly parameterized family $\mathcal{L}$ consisting of linear combinations of the base matrices; i.e. 
\begin{align*}
\mathcal{L} = \Bigg\{ \sum_{i=1}^m \theta_i \vec{B}_i: \bm{\theta} = [\theta_1, \ldots, \theta_m] \in \R^{m} \Bigg\}
\end{align*}
Halikias and Townsend observe (see \cite{HalikiasTownsend:2023}, Lemma 2.2) that recovering a matrix $\vec{A}\in \mathcal{L}$ requires at least $\lceil m/n\rceil$ matvec queries with $\vec{A}$.
In particular, any matrix-vector product with $\vec{A}$ or $\vec{A}^\T$ yields $n$ linear equations in entries of $\bm{\theta}$.
Since recovering $\vec{A}$ amounts to determining the corresponding $\bm{\theta}$, at least $m$ such equations are needed to uniquely determine $\vec{A}$. 
We leverage this observation to prove the following:
\begin{theorem}\label{thm:hodlr_lower_bound}
    Any algorithm that can recover any $n\times n$ matrix $\vec{A} \in\HODLR(k)$ from adaptive matvec queries must use at least $\lceil k \log_2(n/k)\rceil$ queries.
    Therefore, solving \Cref{prob:approx} requires $\Omega(k \log(n/k))$ matvec queries for any finite approximation factor $\Gamma$.
\end{theorem}
\begin{proof}
The proof follows from observing that, while $\HODLR(k)$ is not itself a linearly parameterized family, it contains a  linearly parameterized family $\mathcal{L}$. In particular, consider the subset of $\HODLR(k)$ matrices where every rank-$k$ block in the matrix is zero everywhere except in its first $k$ columns. The first $k$ columns in each block can be chosen to have entries with any value such that the block is rank-$k$, and the matrix  is therefore $\HODLR(k)$.  As such, $\mathcal{L}$ is equal to the set of matrices that can be written as a linear combination of $\vec{B}_1, \ldots, \vec{B}_m$, where each $\vec{B}_i$ is a matrix that is zero everywhere, but has a single $1$ in one of the first $k$ columns in one of the rank-$k$ blocks of the HODLR structure. Recall that for an $n\times n$ matrix to be $\HODLR(k)$, it must be that $n = n_{\text{base}}\cdot 2^p$ for integers $n_{\text{base}} \in (k/2,k]$ and $p = \log_2(n/n_{\text{base}})$.\footnote{If $n\leq k/2$ then the bound is vacuously true.}
It can be checked that the total number of base matrices, $m$, equals:
\begin{align*}
m = 2n_{\text{base}}n + nk(p-1) > nkp \geq nk\log_2(n/k)
\end{align*}
So, by Lemma 2.2 in \cite{HalikiasTownsend:2023}, we require $\lceil nk \log_2(n/k)\rceil/ n\rceil  = \lceil k \log_2(n/k)\rceil$ matvec queries to exactly recover a given $\vec{A}\in \mathcal{L}$. Since $\mathcal{L} \subset \HODLR(k)$, the theorem follows. We note that to solve \Cref{prob:approx} for finite $\Gamma$, we must exactly recover all $\vec{A} \in \HODLR(k)$.
\end{proof}

\subsection{Lower bound for approximation}
Next, we prove a lower bound in the setting where $\vec{A}$ is not exactly HODLR, and we seek to solve \Cref{prob:approx} with error $\Gamma = (1+\varepsilon)$ for some $\varepsilon \in (0,1)$. A natural approach for such a lower bound might be to leverage lower bounds for near-optimal rank-$k$ approximation, since HODLR approximation is a strictly harder problem. However, as discussed in \Cref{sec:contributions}, the best known lower bound for Frobenius-norm error rank-$k$ approximation in the matrix-vector query model is just $O(k + 1/\varepsilon^{1/3})$ \cite{BakshiNarayanan:2023}. Such a lower bound does not show that the multiplicative relationship between $k$ and $\poly(1/\varepsilon)$ in our upper bound, \Cref{thm:main}, is necessary. 

We obtain a stronger lower bound by instead proving a reduction from the problem of \emph{fixed-pattern sparse matrix approximation}, which was recently studied in \cite{AmselChenDumanKelesHalikiasMuscoMusco:2024}. That work proves tight lower bounds in the matrix-vector query model for approximating matrices by a wide variety of matrix classes with fixed sparsity patterns, including diagonal matrices, banded matrices, and more. 
Below, we state a special case of the lower bound from \cite{AmselChenDumanKelesHalikiasMuscoMusco:2024} for approximation by \emph{block diagonal matrices}. This lower bound will be used  to obtain our lower bound for HODLR approximation. For integers $n$ and $b$, such that $b$ divides $n$, we let $\mathcal{B}(n,b)$ denote the set of $n\times n$ block-diagonal matrices with blocks of size $b\times b$. We have the following:

\begin{lemma}[{Corollary of Thm. 2 from \cite{AmselChenDumanKelesHalikiasMuscoMusco:2024}\protect\footnote{The result in \cite{AmselChenDumanKelesHalikiasMuscoMusco:2024} is stated for a particular choice of $n_0 = \Theta(b/\epsilon)$, but it can be seen to hold for all $n > n_0$ as well by simply padding the matrices in the hard input distribution with zeros to enlarge their size.}}]
\label{lem:block_diag}
There are absolute constants $c,C>0$ such that the following holds:
For any $\varepsilon > 0$ and any positive integers $b,n$ such that $b$ divides $n$ and $n \geq cb/\varepsilon$, there is a distribution over $n\times n$ matrices $\vec{A}$ such that any algorithm that accesses $\vec{A}$ with $\leq C b/\varepsilon$ adaptive matvec queries and returns an approximation $\widetilde{\vec{B}}\in \mathcal{B}(n,b)$ must have $\|\vec{A} - \widetilde{\vec{B}}\|_\F \geq (1+\varepsilon)\min_{\vec{B}\in \mathcal{B}(n,b)}\|{\vec{B}} - \vec{A}\|_\F$ with probabilty $\geq 24/25$.
The probability is taken over the randomness of $\vec{A}$, as well as possible randomness in the algorithm.
\end{lemma}
\Cref{lem:block_diag} establishes that, even to succeed with small positive probability, any algorithm for computing a near-optimal block-diagonal approximation  to an arbitrary matrix $\vec{A}$ requires $\Omega(b/\varepsilon)$ matvec queries to $\vec{A}$.
Intuitively this result is useful because block diagonal approximation is an \emph{easier} problem than HODLR approximation. 
In particular, it is not hard to verify that $\mathcal{B}(n,2k)\subset \HODLR(k)$; matrices in $\mathcal{B}(n,2k)$ are zero except on the diagonal and off-diagonal low-rank blocks of the final level of $\HODLR(k)$ matrices. Formally, we prove that, if we had an algorithm for finding a near-optimal HODLR approximation to a given matrix, the result could be post-processed via projection onto $\mathcal{B}(n,2k)$ to obtain a near-optimal block-diagonal approximation. If we found the near-optimal HODLR approximation with $o(k/\epsilon)$ matvec queries, we would violate \Cref{lem:block_diag}. This approach results in the following lower bound:

\begin{theorem}
\label{thm:lower_bound_inexact}
There are absolute constants $c,C>0$ such that the following holds:
For any $k$, $\varepsilon > 0$, and $n \geq ck/\varepsilon$,\footnote{Recall that $\HODLR(k)$ matrices are only defined for dimensions $n$ of the form $n = n_{\text{base}}2^p$ for $n_{\text{base}} \leq k$ and $p\geq 0$. Formally, the theorem holds for any such $n$ that is $\geq ck/\varepsilon$.} there is a distribution over $n\times n$ matrices $\vec{A}$ such that any algorithm which accesses $\vec{A}$ with $\leq C k/\varepsilon$ adaptive matvec queries fails to solve \cref{prob:approx} with probability $\geq 24/25$.
\end{theorem}
\begin{proof}
First note that we may assume $\epsilon \leq 1$, as the result already holds by \Cref{thm:hodlr_lower_bound} for larger values of $\epsilon$.
Let $n = n_{\text{base}}2^p$ for $n_{\text{base}} \in \lfloor k/2 + 1\rfloor, \ldots, k$ and $p\geq 0$.
Let $b = 2n_\text{base}$ and observe that $b \geq k$, and for $c \geq 2$, $b \leq n$.
Let $\vec{S} \in \{0,1\}^{n\times n}$ be an indicator matrix for a block-diagonal sparsity pattern with $n/b$ blocks of size $b$. 
I.e., $\vec{S}$ is zero everywhere except that, for each $t\in 0, \ldots, n/b - 1$, $\vec{S}_{i,j} = 1$ for $i,j \in \{bt + 1, \ldots, bt+b\}$. Observe that for any block diagonal matrix $\vec{B}\in \mathcal{B}(n,b)$, $\vec{B}\circ \vec{S} = \vec{B}$, where ``\:$\circ$\:'' denotes the entrywise product. Let $\bar{\vec{S}}$ denote the compliment of $\vec{S}$: for all $i,j$, $\bar{\vec{S}}_{i,j} = 1-\vec{S}_{i,j}$. Observe that for
any matrix $\vec{H}$, we can write:
\begin{align}
\label{eq:bd_decomp}
    \|\vec A - \vec  H\|_\F^2 = \|\vec A \circ \vec S- \vec H \circ \vec S\|_\F^2 +  \|\vec A \circ \Bar{\vec S}- \vec H \circ \Bar{\vec S}\|_\F^2.
\end{align}
Define
\[
\vec B^* = \vec{A}\circ \vec{S}\argmin_{\vec B\in \mathcal{B}(n,b)} \|\vec A- \vec B\|_\F^2 
,\qquad 
\vec H^*=\argmin_{\vec H\in \HODLR(k)} \|\vec A- \vec H\|_\F^2.\]
Since $\mathcal{B}(n,b)\subset \HODLR(k)$, we have that $\|\vec{A} - \vec{H}^*\|_\F \leq \|\vec{A} - \vec{B}^*\|_\F$.

Our proof approach is to show that, if we can solve \Cref{prob:approx} with error $\Gamma = (1+\varepsilon)$ using $m$ matrix-vector products, i.e. if we can find some $\widetilde{\vec{H}}\in \HODLR(k)$ that satisfies
\begin{align}
\label{eq:condition_bd}
\|\vec{A} - \widetilde{\vec{H}}\|_\F \leq (1+\varepsilon)\|\vec{A} - \vec{H}^*\|_\F,
\end{align}
then $\widetilde{\vec{H}}\circ \vec{S}$ is a near-optimal block diagonal approximation to $\vec{A}$. 
Concretely, we will show that 
\begin{align}
\label{eq:to_prove_bd}
\|\vec{A} - \widetilde{\vec{H}}\circ \vec{S}\|_\F \leq (1+\varepsilon)\|\vec{A} - \vec{B}^*\|_\F.
\end{align}
Accordingly, we reach a contradiction to \Cref{lem:block_diag} unless $m = \Omega(k/\varepsilon) = \Omega(b/\varepsilon)$ matrix-vector products were used to compute $\widetilde{\vec{H}}$.

To see why \eqref{eq:condition_bd} implies \eqref{eq:to_prove_bd}, the main observation is that:
\begin{align}
\label{eq:main_observe_bd}
\|\vec A \circ \Bar{\vec S}- \widetilde{\vec{H}} \circ \Bar{\vec S}\|_\F^2 \geq \|\vec A- \vec H^*\|_\F^2.
\end{align}
\eqref{eq:main_observe_bd} follows from the fact that the entries in the bottom level of a $\HODLR(k)$ matrix (those corresponding to the ones in $\vec{S}$) can be set arbitrarily without violating the $\HODLR$ property. Accordingly, by adjusting those entries to match $\vec{A}$ exactly, we can obtain a HODLR approximation with error equal to $\|\vec A \circ \Bar{\vec S}- \widetilde{\vec{H}} \circ \Bar{\vec S}\|_\F^2$. The optimal approximation $\vec{H}^*$ cannot have larger error.

Then, using \eqref{eq:condition_bd}, \eqref{eq:bd_decomp}, and \eqref{eq:main_observe_bd}, we have:
    \begin{align*}
        (1+\varepsilon)^2 \|\vec A- \vec H^*\|_\F^2 \geq \|\vec A- \widetilde{\vec{H}} \|_\F^2 &= \|\vec A \circ \vec S- \widetilde{\vec{H}} \circ \vec S\|_\F^2 +  \|\vec A \circ \Bar{\vec S}- \widetilde{\vec{H}} \circ \Bar{\vec S}\|_\F^2 \\
        &\geq \|\vec A \circ \vec S- \widetilde{\vec{H}} \circ \vec S\|_\F^2 + \|\vec A- \vec H^*\|_\F^2.
    \end{align*}
    We conclude that
    \begin{align*}
    \|\vec A \circ \vec S- \widetilde{\vec{H}} \circ \vec S\|_\F^2 \leq (2\varepsilon + \varepsilon^2)\|\vec A- \vec H^*\|_\F^2 \leq (2\varepsilon + \varepsilon^2)\|\vec A- \vec B^*\|_\F^2.
    \end{align*}
    Moreover, $\|\vec A - \widetilde{\vec{H}} \circ \vec S\|_\F^2 = \|\vec A \circ \vec S- \widetilde{\vec{H}} \circ \vec S\|_\F^2 + \|\vec{A} - \vec{A}\circ\vec{S}\|_\F^2$, and recall that $\vec{A}\circ\vec{S} = \vec{B}^*$, So:
    \begin{align*}
    \|\vec A - \widetilde{\vec{H}} \circ \vec S\|_\F^2 \leq (2\varepsilon + \varepsilon^2)\|\vec A- \vec B^*\|_\F^2 + \|\vec A- \vec B^*\|_\F^2 = (1+\varepsilon)^2\|\vec A- \vec B^*\|_\F^2.
    \end{align*}
    Taking a square root on both sides proves \eqref{eq:to_prove_bd}, and as explained above, \Cref{thm:lower_bound_inexact} follows.
\end{proof}
We conclude the section by noting that \cref{thm:lowerbd} (stated in \Cref{sec:intro}) follows from combining \Cref{thm:hodlr_lower_bound} and \Cref{thm:lower_bound_inexact}. An interesting question for future work is to prove a lower bound that multiplicatively combines $k$, $\log(n/k)$, and $1/\varepsilon$; e.g., to prove that $m = \Omega(k\log(n/k)/\varepsilon)$ matrix-vector product queries are necessary to solve \Cref{prob:approx}. This is currently beyond the reach of our current approach and that of \cite{AmselChenDumanKelesHalikiasMuscoMusco:2024}. In particular, the lower bound instance in \cite{AmselChenDumanKelesHalikiasMuscoMusco:2024} is based on a random Wishart matrix, which has found applications in a number of prior results on lower bounds for adaptive matrix-vector query algorithms \cite{BravermanHazanSimchowitzWoodworth:2020}. It can be checked that a constant factor near-optimal HODLR approximation for such a matrix can be found by simply returning a near-optimal block diagonal approximation for block size $O(k)$. Since that can be done with $O(k)$ matrix-vector products using the algorithm from \cite{AmselChenDumanKelesHalikiasMuscoMusco:2024}, we cannot hope to use the Wishart instance to prove a lower bound that combines $k$ and $\log(n/k)$.

\section{Numerical experiments and examples}
\label{sec:examples}

In this section we provide several examples which provide insight into the behavior of peeling-based algorithms.
The code used to generate our figures is available at \url{https://github.com/tchen-research/HODLR_approx}.

In our experiments we consider two parameter choices for the Generalized Nystr\"om Method based method \cref{alg:main} and two parameter choices for a Randomized SVD based method \cref{alg:RSVD_peel} described in \cref{sec:RSVD_peel}.
The parameter values are summarized in \cref{tab:alg_params}.
The aim of our numerical experiments is to provide some basic insight into how the various parameters of peeling algorithms impact their performance. 
There are many reasonable parameter combinations, and a comprehensive understanding of the best choice of parameters is far beyond the scope of this paper, but would be an interesting topic for future work.

\begin{table}[ht]
    \centering
    \begin{tabular}{cccccccc}
         \toprule
         name & style & algorithm & $\sright$ & $\tright$ & $\sleft$  & $\tleft$ & \# matvecs \\
         \midrule
         GN1&\includegraphics[scale=.7]{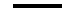}&\cref{alg:main} & $k/\beta$ & 1  & $k/\beta^2$& 1 & $O(k/\beta^2)$\\
         GN2&\includegraphics[scale=.7]{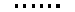}&\cref{alg:main} & $k/\beta$ & $1/\beta$ & $k/\beta^2$ & 1 & $O(k/\beta^2)$\\
         RSVD1&\includegraphics[scale=.7]{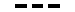}&\cref{alg:RSVD_peel} & $k/\beta$ & 1 & $\sim$ & 1 & $O(k/\beta)$ \\
         RSVD2&\includegraphics[scale=.7]{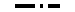}&\cref{alg:RSVD_peel} & $k/\beta$ & $1/\beta$ & $\sim$ & $1/\beta$ & $O(k/\beta^2)$ \\
         \bottomrule
    \end{tabular}
    \caption{Parameter choices used in our numerical experiments.}
    \label{tab:alg_params}
\end{table}
The RSVD1 and GN1 methods do not use random perforated sketches and can therefore be viewed as standard implementations of the peeling algorithm (e.g. as described in \cite{LinLuYing:2011,Martinsson:2016}).
The RSVD1 method uses the same sketching dimensions as required to obtain a $(1+O(\beta))$-optimal rank-$k$ approximation to a matrix \cite{HalkoMartinssonTropp:2011}.
Similarly, the GN1 method uses the sketching dimensions required for the Generalized Nystr\"om Method (without truncation to rank-$k$) to produce a low-rank approximation with error less than $(1+O(\beta))$ times the optimal rank-$k$ approximation to a matrix \cite{TroppYurtzeverUdellCevher:2017}.\footnote{As we discuss in \cref{sec:perturb_GNYS}, it is believed that this also produces a $(1+O(\beta))$-optimal approximation with truncation.}
The choice of paramaters for GN1 is also equivalent (after rescaling $\beta$) to the parameters \cref{eqn:GNpeel_params} needed for \cref{thm:main_full} to guarantee convergence for Generalized Nystr\"om peeling method without perforation.

Both RSVD2 and GN2 use the randomly perforated sketches described in \cref{sec:techniques}.
In particular, the GN2 method adds perforation to the sketch used to obtain the approximate range.
Compared to GN1, this does not increase the asymptotic cost but improves the bound for $\Gamma$ which can be obtained from \cref{thm:main_full}.
Interestingly, in our numerical experiments, the accuracy of GN1 and GN2 is very similar.
The RSVD2 method adds perforation to both sketches.
This increases the asymptotic cost over RSVD1. 
However, as illustrated in \cref{sec:hard_instance}, regardless of $\beta$, RSVD1 cannot solve \cref{prob:approx} for arbitrary $\Gamma>1$. 
On the other hand, RSVD2 can.

For an approximation $\widetilde{\vec{A}}$ to $\vec{A}$, we will consider the relative and absolute errors defined by:
\[
    \textup{relative error:}~ \frac{\|\vec{A} - \widetilde{\vec{A}}\|_\F - \|\vec{A} - \vec{A}^\star\|_\F^2}{\|\vec{A} - \vec{A}^\star\|_\F^2}
    ,\qquad
    \textup{absolute error:}~ \|\vec{A} - \widetilde{\vec{A}}\|_\F,
\]
where $\vec{A}^\star := \min_{\vec{H}\in\HODLR(k)}\|\vec{A} - \vec{H}\|_\F^2$ is the best possible $\HODLR(k)$ approximation to $\vec{A}$.
When $\widetilde{\vec{A}}$ is $\HODLR(k)$, the relative error corresponds to the smallest value of $\varepsilon$ for which \cref{prob:approx} is solved with $\Gamma = (1+\varepsilon)$.

\subsection{Poisson's equation}

In this example, we take $\vec{A}$ as the discretized solution operator to a  differential equation. 
In particular, we consider the 2-dimensional Poisson's equation 
\[
\frac{\partial^2}{\partial x^2} u(x,y) + \frac{\partial^2}{\partial y^2} u(x,y) = f(x,y)
\]
on a periodic domain $(x,y)\in[0,1]^2$ (i.e., with boundary conditions $u(x,0) = u(x,1)$ and $u(0,y) = u(1,y)$.
We discretize the problem on a uniform $t\times t$ grid $x_i = i/t$, $y_j = j/t$ for $i,j=0,1,\ldots, t-1$ (so that $n=t^2$).
The matrix $\vec{A}$ is defined as the $n\times n$ linear map taking forcing data $\vec{f} = \{ f(x_i,y_j)\}_{i,j=1}^{t}\in\mathbb{R}^{n}$ to the solution approximate $\vec{u} = \{ u_{i,j} \}_{i,j=1}^{t}$.
Matrix-vector products with $\vec{A}$ (and $\vec{A}^\T$) can be efficiently computed using a FFT-based 2D Poisson solver.
Specifically, given $\vec{f} \in \mathbb{R}^n = \mathbb{R}^{t\times t}$, we define $\vec{A}$ by
\begin{equation}
\label{eqn:poisson_def}
    \vec{A}\vec{f} 
= 
\operatorname{IDFT}_2(\vec{D} \operatorname{DFT}_2(\vec{f})),
\end{equation}
where $\vec{D}:\mathbb{R}^{t\times t} \to \mathbb{R}^{t\times t}$ is the diagonal map with diagonal entries $\vec{D}_{i,j} = -1/(\kappa_i^2 + \kappa_j^2)$, where the harmonics $\kappa_i$ are defined by $\kappa_i = i$ if $i\leq t/2-1$ and $\kappa_i = m-t$ if $i>t/2-1$.
Here $\operatorname{DFT}_2$ and $\operatorname{IDFT}_2$ are respectively the 2-dimensional Discrete Fourier Transform and Inverse Discrete Fourier Transform.
In our experiment we set $t=32$ so that $n = t^2 = 1024$.

\begin{figure}[tb]
    \centering
    \includegraphics[scale=.7,trim={0 0 0 0cm},clip]{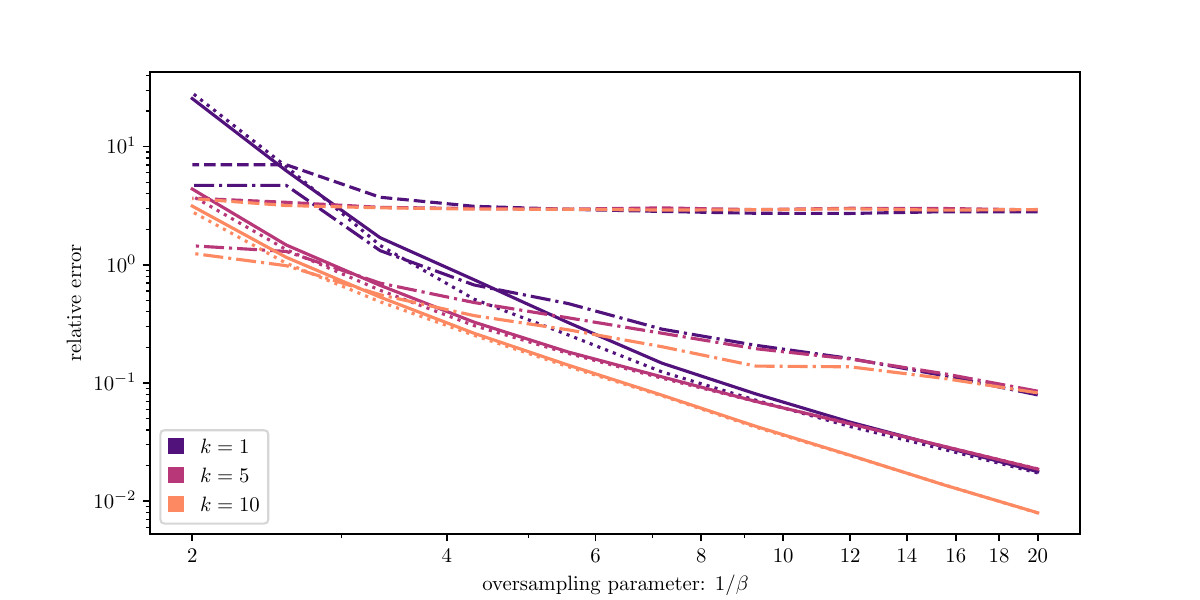}
    \caption{Relative error of peeling algorithms from \cref{tab:alg_params} on discrete inverse Poisson matrix described in \cref{eqn:poisson_def} as a function of the the oversampling parameter $\beta$ for several choices of the rank-parameter $k$.
    \emph{Legend}:
    GN1 (\includegraphics[scale=.7]{imgs/legend/solid.pdf}),
    GN2 (\includegraphics[scale=.7]{imgs/legend/dot.pdf}),
    RSVD1 (\includegraphics[scale=.7]{imgs/legend/dash.pdf}),
    RSVD2 (\includegraphics[scale=.7]{imgs/legend/dashdot.pdf}).
    \emph{Takeaway}: Both Generalized Nystr\"om Method-based variants seem to perform similarly, and produce an increasingly accurate output as $1/\beta$ increases. 
    On the other hand, the standard randomized SVD-based variant, RSVD1, stagnates.
    The RSVD2 variant, which uses randomly perforated sketches, converges. 
    }
    \label{fig:poisson_oversample}
\end{figure}

\begin{figure}[htb]
    \centering
    \includegraphics[scale=.7,trim={0 0 0 0cm},clip]{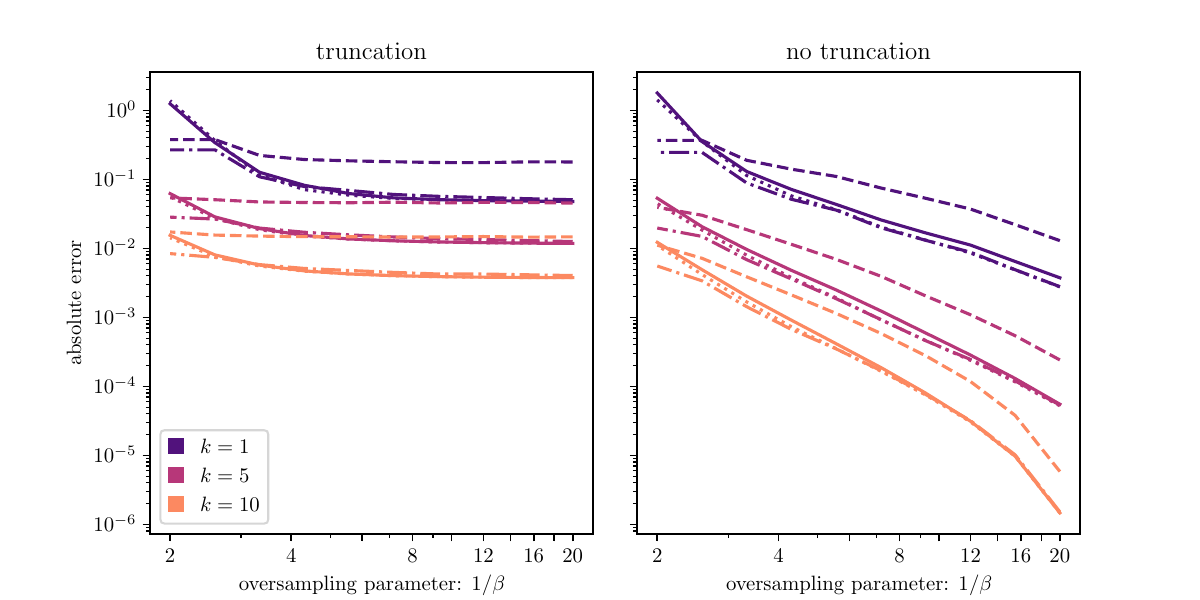}
    \caption{Absolute error of peeling algorithms from \cref{tab:alg_params} on discrete inverse Poisson matrix described in \cref{eqn:poisson_def} as a function of the oversampling parameter $\beta$ with and without truncation to rank-$k$ for several choices of the rank-parameter $k$.
    \emph{Legend}:
    GN1 (\includegraphics[scale=.7]{imgs/legend/solid.pdf}),
    GN2 (\includegraphics[scale=.7]{imgs/legend/dot.pdf}),
    RSVD1 (\includegraphics[scale=.7]{imgs/legend/dash.pdf}),
    RSVD2 (\includegraphics[scale=.7]{imgs/legend/dashdot.pdf}).
    \emph{Takeaway}:
    Without truncation, all of the algorithms can perform significantly better. 
    However, the resulting approximations are not $\HODLR(k)$, and are therefore more expensive to store and work with.
    }
    \label{fig:poisson_oversample_nt}
\end{figure}

In \cref{fig:poisson_oversample} we show the relative error of the algorithms from \cref{tab:alg_params} for various ranks $k$ as a function of the oversampling parameters $\beta$ (averaged over 20 trials). 
As expected, as $\beta \to 0$, the relative errors of GN1, GN2, and RSVD1 all converge to zero. 
On the other hand, RSVD1 stagnates. 
In the left panel of \cref{fig:poisson_oversample_nt} we show the absolute error. 
Here we see that while the RSVD1 algorithm has a much higher relative error than the other algorithms, the absolute error is not significantly larger. 
Finally, in the right panel of \cref{fig:poisson_oversample_nt} we show the absolute error of the algorithms without truncation to rank-$k$. 
This results in a much better approximations for the same number of matrix-vector products, but the resulting approximations have HODLR rank larger than $k$.
The best tradeoff between HODLR rank and matvecs depends on the problem at hand and the computing environment.

\begin{figure}[b]
    \centering
    \includegraphics[scale=.7]{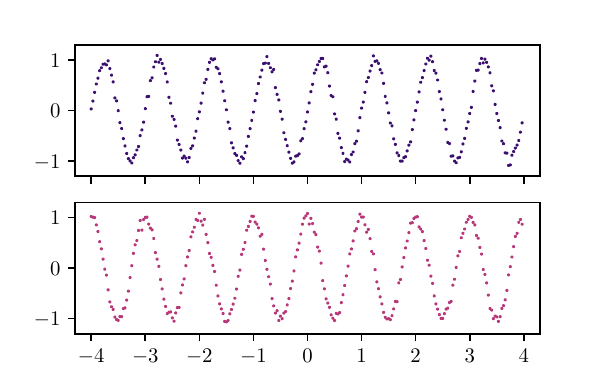}
    \hspace{1cm}
    \includegraphics[scale=.7]{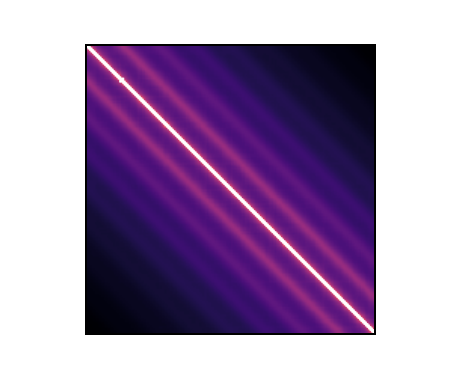}
    \caption{\emph{Top/Bottom Left}: $(x_i,y_i)$ and $(x_i,z_i)$, ordered by $x$-value.
    \emph{Right}: Log-magnitude of entries of $\vec{A}$ defined in \cref{eqn:kernel_def}.
    In both plots we subsample the data to $256$ points for visual clarity.}
    \label{fig:distances_setup}
\end{figure}

\subsection{Kernel Matrix}

In this example we consider a kernel matrix $\vec{A}$ defined by
\begin{equation}
    \label{eqn:kernel_def}
[\vec{A}]_{i,j} 
= \begin{cases}
1/\|\vec{x}_i - \vec{x}_j \|_2 & i\neq j \\
0 & i=j
\end{cases},
\end{equation}
where $\{\vec{x}_i\}_{i=1}^n$ are a collection of points.
Matrix-vector products with matrices such as $\vec{A}$ can be efficiently performed using algorithms such as the Fast Multipole Method \cite{GreengardRokhlin:1987}. 

In our numerical experiments, we set $n=16384$ and sample points $\vec{x}_i = (x_i,y_i,z_i)$ by taking $x_i$ as uniformly spaced points in $[-4,4]$, $y_i = \sin(2\pi x_i) + 0.05 \xi_i$, and $z_i = \cos(2\pi x_i) + 0.05 \zeta_i$, where $\xi_i$ and $\zeta_i$ are independent standard normal random variables.
Products with $\vec{A}$ are performed using the Flatiron Institute's FMM3D code \cite{fmm3d}.
The left panel of \cref{fig:distances_setup} shows a sample of the points we use, and the right panel shows the magnitude of the entries of the kernel matrix \cref{eqn:kernel_def} induced by these points.
In \cref{fig:distances_oversample} we show the absolute error of the GN1 and RSVD1 methods as a function of the target rank, for several values of $\beta$.

\subsection{Hard instances for RSVD-based peeling}
\label{sec:hard_instance}

In this section we provide two examples which illustrate that a RSVD-based implementation of the peeling algorithm, which truncates the low-rank approximations to rank-$k$ at every level, cannot solve \cref{prob:approx} for certain values of $\Gamma$.

We emphasize that the purpose of this section is simply to illustrate a potential failure mode which must be addressed by an algorithm provably solving \cref{prob:approx}.
In particular, these hard instances are not hard instances for more practical RSVD-based variants which do not truncate.

\subsubsection{An illustrative example}
\label{ex:RSVD_hard}

Our first hard instance illustrates that all of the error from a one level can propagate to the next level.
This example provides clear intuition for why the Randomized SVD-based peeling algorithm (with sketches of size $\sright$ and truncation to rank-$k$ at every level) cannot solve \cref{prob:approx} for arbitrary $\Gamma>1$.
We emphasize that the hard instance depends on the truncation rank-$k$ used by the algorithm. 
If we allow the algorithm to use an adaptively chosen rank or do not use truncation (both of which are often done in practice \cite{LinLuYing:2011,LevittMartinsson:2024a}) then it would no longer be a hard instance.

Define, for some $\eta \gg 1$,
\[
\vec{X} = \begin{bmatrix}
    \vec{I}_k & \vec{0} \\
    \vec{0} & \vec{0}
\end{bmatrix}
,\qquad
\vec{Y} = \eta \begin{bmatrix}
    \vec{0} & \vec{0} \\
    \vec{0} & \vec{I}_k
\end{bmatrix}.
\]
Construct the matrix
\[
\vec{A} 
= 
\begin{bNiceArray}{ccIcc}[margin,custom-line = {letter=I,tikz = dotted}]
    \vec{0} & \vec{X} & \vec{Y} & \vec{X}  \\
    \vec{X} & \vec{0} & \vec{X} & \vec{0}   \\ \Hline[tikz = dotted]
    \vec{Y}&\vec{X}&\vec{0} & \vec{X}  \\
    \vec{X}&\vec{0}&\vec{X} & \vec{0}  \\
\end{bNiceArray}.
\]
For notational convenience, we write equality for the limits as $\eta\to\infty$.
In particular, when $\eta\to\infty$, Randomized SVD (with any $\sright\geq k$) will recover optimal rank-$k$ approximations to the bottom left and top right blocks.

We then remove the low-rank components we found at the first level to obtain 
\[
\begin{bNiceArray}{ccIcc}[margin,custom-line = {letter=I,tikz = dotted}]
    \vec{0} & \vec{X} & \vec{Y} & \vec{X}  \\
    \vec{X} & \vec{0} & \vec{X} & \vec{0}   \\ \Hline[tikz = dotted]
    \vec{Y}&\vec{X}&\vec{0} & \vec{X}  \\
    \vec{X}&\vec{0}&\vec{X} & \vec{0}  \\
\end{bNiceArray}
-
\begin{bNiceArray}{ccIcc}[margin,custom-line = {letter=I,tikz = dotted}]
    \vec{0} & \vec{0} & \vec{Y} & \vec{0}  \\
    \vec{0} & \vec{0} & \vec{0} & \vec{0}   \\ \Hline[tikz = dotted]
    \vec{Y}&\vec{0}&\vec{0} & \vec{0}  \\
    \vec{0}&\vec{0}&\vec{0} & \vec{0}  \\
\end{bNiceArray}
=
\begin{bNiceArray}{ccIcc}[margin,custom-line = {letter=I,tikz = dotted}]
    \vec{0} & \vec{X} &\vec{0} & \vec{X} \\
    \vec{X} & \vec{0} &\vec{X} & \vec{0} \\ \Hline[tikz = dotted]
    \vec{0}&\vec{X}&\vec{0} & \vec{X}  \\
    \vec{X}&\vec{0}&\vec{X} & \vec{0}  \\
\end{bNiceArray}.
\]
Note, however, that since the off-diagonal low-rank blocks of $\vec{A}$ are not rank-$k$, we  fail to zero out the off-diagonal blocks at the first level.

\begin{figure}
    \centering
    \includegraphics[scale=.7,trim={0 0 0 0cm},clip]{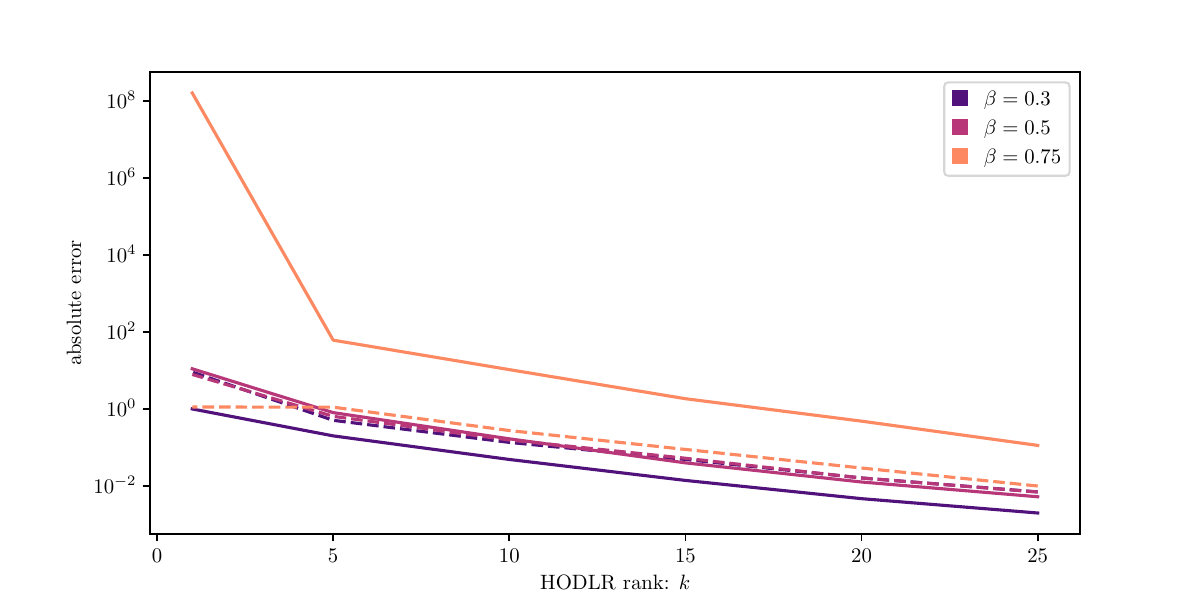}
    \caption{
    Absolute error of peeling algorithms from \cref{tab:alg_params} on the kernel matrix described in \cref{eqn:kernel_def} as a function of the rank-parameter $k$ for several values of $\beta$.
    \emph{Legend}:
    GN1 (\includegraphics[scale=.7]{imgs/legend/solid.pdf}),
    RSVD1 (\includegraphics[scale=.7]{imgs/legend/dash.pdf}).
    \emph{Takeaway}: While RSVD1 may not produce near-optimal low-rank approximations, it can sometimes produce good approximations.
    In addition, when the sketch size used for the regression problem is too large, the Generalized Nystr\"om Method does not produce highly accurate low-rank approximations.
    }
    \label{fig:distances_oversample}
\end{figure}

\begin{figure}[tb]
    \centering
    \includegraphics[scale=.7,trim={0 0 0 0cm},clip]{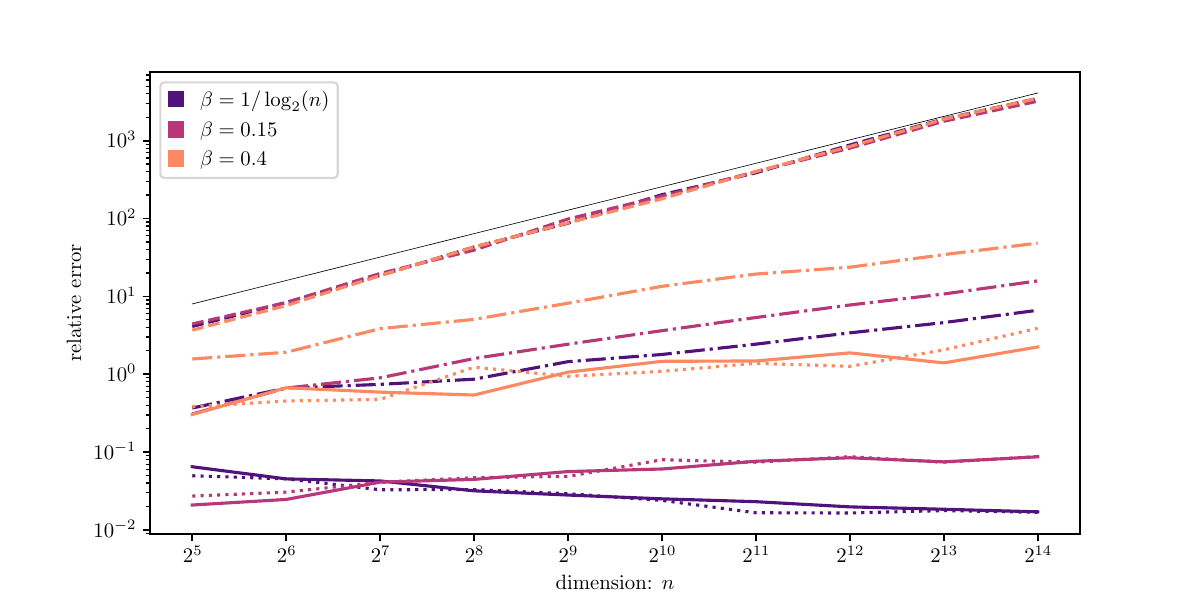}
    \caption{
    Relative error of peeling algorithms on hard instance described in \cref{eqn:exp_hard} for several values of $\beta$.
    \emph{Legend}:
    GN1 (\includegraphics[scale=.7]{imgs/legend/solid.pdf}),
    GN2 (\includegraphics[scale=.7]{imgs/legend/dot.pdf}),
    RSVD1 (\includegraphics[scale=.7]{imgs/legend/dash.pdf}),
    RSVD2 (\includegraphics[scale=.7]{imgs/legend/dashdot.pdf}).
    \emph{Takeaway}: The RSVD1 variant has error growing as $n$ (thin solid reference line); i.e. exponential growth at every level.
    When the sketch size increases sufficiently quickly with $n$, the Generalized Nystr\"om Method-based variants produce an error bounded independent of $n$. 
    }
    \label{fig:exponential_growth}
    \vspace{-1em}
\end{figure}

At the next level, the randomized SVD (with any $\sright\geq k)$ will exactly recover an orthonormal basis $\vec{Q}$ containing the range of $\vec{X}$. 
In particular, we perform a product
\[
\begin{bNiceArray}{ccIcc}[margin,custom-line = {letter=I,tikz = dotted}]    \vec{0} & \vec{X} &\vec{0} & \vec{X} \\
    \vec{X} & \vec{0} &\vec{X} & \vec{0} \\ \Hline[tikz = dotted]
    \vec{0}&\vec{X}&\vec{0} & \vec{X}  \\
    \vec{X}&\vec{0}&\vec{X} & \vec{0}  \\
\end{bNiceArray}
\begin{bNiceArray}{c}[margin,custom-line = {letter=I,tikz = dotted}] 
\vec{\Omega}_1 \\ \vec{0} \\\Hline[tikz = dotted] \vec{\Omega}_3 \\ \vec{0}
\end{bNiceArray}
=
\begin{bNiceArray}{c}[margin,custom-line = {letter=I,tikz = dotted}] 
\vec{0} \\ \vec{X}(\vec{\Omega}_1+\vec{\Omega}_3) \\ \Hline[tikz = dotted] \vec{0} \\ \vec{X}(\vec{\Omega}_1+\vec{\Omega}_3)
\end{bNiceArray}.
\]
We then compute\footnote{Here we assume $\operatorname{orth}(\cdot)$ is a deterministic function.} $\vec{Q} = \operatorname{orth}(\vec{X}(\vec{\Omega}_1+\vec{\Omega}_3))$ and 
\[
\begin{bNiceArray}{ccIcc}[custom-line = {letter=I,tikz = dotted}]    \vec{0} & \vec{Q}^\T & \vec{0} & \vec{Q}^\T 
\end{bNiceArray}
\begin{bNiceArray}{ccIcc}[margin,custom-line = {letter=I,tikz = dotted}]
    \vec{0} & \vec{X} &\vec{0} & \vec{X} \\
    \vec{X} & \vec{0} &\vec{X} & \vec{0} \\ \Hline[tikz = dotted]
    \vec{0}&\vec{X}&\vec{0} & \vec{X}  \\
    \vec{X}&\vec{0}&\vec{X} & \vec{0}  \\
\end{bNiceArray}
= 
\begin{bNiceArray}{ccIcc}[custom-line = {letter=I,tikz = dotted}]
    2\vec{Q}^\T\vec{X} & \vec{0} & 2\vec{Q}^\T\vec{X} & \vec{0}
\end{bNiceArray},
\]
and analogously for the super-diagonal off-diagonal blocks.
Since $\vec{Q}\vec{Q}^\T\vec{X} = \vec{X}$, our rank-$k$ approximation to each of the off-diagonal blocks at the second level is $2\vec{X}$.
In particular, we see that all of the error from the first level propagates to the second level.
Even
assuming we exactly recover the diagonals (if we do not, this can only increase the error), the final approximation is thus
\[
\widetilde{\vec{A}}
= 
\begin{bNiceArray}{ccIcc}[margin,custom-line = {letter=I,tikz = dotted}]
    \vec{0} & 2\vec{X} & \vec{Y} & \vec{0}  \\
    2\vec{X} & \vec{0} & \vec{0} & \vec{0}   \\ \Hline[tikz = dotted]
    \vec{Y}&\vec{0}&\vec{0} & 2\vec{X}  \\
    \vec{0}&\vec{0}&2\vec{X} & \vec{0}  \\
\end{bNiceArray}.
\]
As long as $\eta > 1$, the best HODLR approximation to $\vec{A}$ is 
\[
\vec{A}^\star
= 
\begin{bNiceArray}{ccIcc}[margin,custom-line = {letter=I,tikz = dotted}]
    \vec{0} & \vec{X} & \vec{Y} & \vec{0}  \\
    \vec{X} & \vec{0} & \vec{0} & \vec{0}   \\ \Hline[tikz = dotted]
    \vec{Y}&\vec{0}&\vec{0} & \vec{X}  \\
    \vec{0}&\vec{0}&\vec{X} & \vec{0}  \\
\end{bNiceArray}.
\]
Therefore we find
$\|\vec{A} - \vec{A}^\star \|_\F^2
= 4 \|\vec{X}\|_\F^2$ while $\|\vec{A} - \widetilde{\vec{A}} \|_\F^2 = 8 \|\vec{X}\|_\F^2$.

\subsubsection{Exponential growth}

The example in \cref{ex:RSVD_hard} shows that there are problems for which a Randomized SVD-based peeling algorithm cannot solve \cref{prob:approx} for small (constant) $\Gamma
$.
We now exhibit a problem instance for a Randomized SVD-based peeling algorithm which suggests that such algorithms cannot even guarantee better than an $O(n)$-factor approximation.
This is exponentially large in the number of levels, and to the best of our knowledge, is the first instance demonstrating an exponential instability in the algorithm.

This is the first problem instance we are aware of for which a variant of the peeling algorithm actually incurs a significant propagation of error from level-to-level.
It remains an open question whether other variants of the peeling algorithm (e.g., using Randomized SVD without truncation) can fail for the approximation problem.

For any integer $L>0$ let $n = 2^L$ and define the $n\times n$ matrix $\vec{A}_n$ by
\begin{equation}
\label{eqn:exp_hard}
    [\vec{A}_n]_{i,j} = 
\begin{cases}
    1 & j=0,~ i~\text{odd} \\
    \eta & j=1,~ i = 2^1,2^2, \ldots, 2^L \\
    0 & \text{otherwise}
\end{cases}.
\end{equation}
For each of the nonzero off-diagonal blocks, note that the second column is orthogonal to the first column.
Hence, the optimal rank-$1$ HODLR approximation to $\vec{A}_n$ is just $\vec{A}_n$ with the first column set to zero.
This means the error of the optimal approximation is from the $n/2-1$ ones in the first column excluding the $(1,1)$ entry on the diagonal.

We set $\eta = 10^8$ and run an implementation of the peeling algorithm using the Randomized SVD with truncation to rank-$1$ at each level to obtain a HODLR rank-1 matrix $\widehat{\vec{A}}_n$.
This is repeated for increasing values of $n$.
In \cref{fig:exponential_growth} we plot (as a function of $n$) the quantity $\|\vec{A}_n - \widehat{\vec{A}}_n\|_\F / \sqrt{n/2-1}-1$, which is the smallest value of $\Gamma$ for which the output $\widehat{\vec{A}}_n$ solves \cref{prob:approx}  (averaged over 20 trials).
This experiment suggests that $\|\vec{A}_n - \widehat{\vec{A}}_n\|_\F / \sqrt{n/2-1}-1= \Theta(n)$ as $n\to\infty$. 
Since $n = 2^L$, this is exponentially bad (with a constant base $2$) in the number of levels $L$.

\section{Outlook}

We have presented an algorithm provably solving the HODLR approximation problem in the matrix-vector query model.
As far as we can tell, this is the first result on the approximation problem in the matrix-vector query model for any hierarchical matrix family.
While we focus on HODLR approximation for clarity of exposition, we expect the ideas used in our analysis can be extended to algorithms for other hierarchical families. 
The most natural would extension would be to the coloring-based variant of the peeling algorithm for $\mathcal{H}$-matrices introduced in \cite{LevittMartinsson:2024a}. 
The $\mathcal{H}$ format is a generalization of HODLR which allows for a more general tree structure and recursive blocks off of the diagonal, and is significantly more efficient than HODLR for multi-dimensional problems.

Our work raises a number of interesting questions on approximation algorithms for hierarchical matrices:
\begin{itemize}
    \item  
    For any constant $c>0$, we show that \cref{prob:approx} can be solved to accuracy $\Gamma = n^c$ with $O(k \log(n/k))$ matrix-vector queries. 
    Up to constants, this is the best possible query complexity; as we prove in \cref{thm:lowerbd}, exactly recovering a HODLR matrix requires $O(k\log(n/k))$ queries. 
    It remains open whether there exist matvec query algorithms which solve \cref{prob:approx} to higher accuracy (e.g. $\Gamma = \log(n)$) with $O(k \log (n/k))$ queries.

    \item Hierarchical Semi-Separable (HSS) matrices are an important subfamily of HODLR matrices.
    The low-rank factors of HSS matrices at different levels are related in such a way that they can be stored and manipulated more efficiently than general HODLR matrices. 
    In particular, there are a number of algorithms for recovering an exactly HSS matrix that require $O(k)$ matrix-vector products \cite{LevittMartinsson:2024,HalikiasTownsend:2023}.
    It is therefore natural to ask whether there exists an HSS \emph{approximation} algorithm producing $(1+\varepsilon)$-optimal HSS approximation using $O(\operatorname{poly}(k,1/\varepsilon))$ matvecs.
    A major difficulty is that, in contrast with HODLR matrices, we are unaware of a simple characterization of the best HSS approximation to a given matrix.
       
    \item Operator learning aims to learn representations of operators mapping functions to functions \cite{LuJinPangZhangKarniadakis:2021, KovachkiLanthalerStuart:2024, BoulleTownsend:2023, LiKovachkiLiuBhattacharyaStuartAnandkumar:2021, GinSheaBruntonKutz:2021}.
    A number of recent works study the problem of approximating certain classes of infinite dimensional linear operators \cite{BoulleTownsend:2022,BoulleHalikiasTownsend:2023} by hierarchically structured operators. 
    Understanding how our analysis and theoretical techniques extend to the infinite dimensional setting may yield stronger theoretical guarantees for some problems in operator learning.  

    \item Our structural perturbation bound \cref{thm:RSVD_perturb} for low-rank approximation suggests fundamental differences between the behaviors of the Randomized SVD and Generalized Nystr\"om Method  in the presence of noisy matrix-vector products.
    The bound also highlights limitations in our current understanding of the impact of truncation on the Generalized Nystr\"om Method.
    The best known bounds for the method with truncation are worse than without truncation, but we are unaware of any convincing evidence that such bounds are sharp.
    Further theoretical and numerical studies of these methods would be of interest.

\end{itemize}

\section*{Acknowledgements}

Diana Halikias was supported by the Office of Naval Research (ONR) under grant N00014-23-1-2729. 
Cameron Musco was partially supported by NSF Grants 2046235 and 2427363.
Tyler Chen and Christopher Musco were partially supported by NSF Grants 2045590 and 2427363.
David Persson was supported by the SNSF research project \textit{Fast algorithms from low-rank updates}, grant number: 200020\_178806.
\\

\noindent
We also thank Noah Amsel for useful discussions in the early stages of this work.

\appendix
\crefalias{section}{appendix}

\section{A randomized-SVD peeling algorithm}
\label{sec:RSVD_peel}

As discussed in \cref{sec:hard_instance}, a simple RSVD-based peeling algorithm with truncation cannot perform well, as there is no way of controlling errors made in the projection step.
However, by using a similar randomized perforation technique as described in \cref{sec:sketching_dist}, one can implement an RSVD-based algorithm which can solve \cref{prob:approx} for arbitrary $\Gamma > 1$.
We now define the additional notation needed.
The full algorithm is described in \cref{alg:RSVD_peel}.

At level $\ell$, the RSVD-based peeling algorithm will obtain left subspaces $\vec{Q}^{(\ell)}_j$ as in \cref{sec:left_subspace}.
However, rather than solving a regression problem like \cref{alg:main}, the algorithm will attempt to directly compute $(\vec{Q}^{(\ell)}_j)^\T\vec{A}^{(\ell)}_{\jp,j}$.
 To control the error, we first sample $\bm{\zeta}^+,\bm{\zeta}^-\sim\operatorname{PerfCountSketch}(d,\tright)$ and define sketching matrices
\[
\vec{\Psi}^+ = \bm{\xi}^+ \bullet  \vec{Q}^{(\ell)},
\qquad
\vec{\Psi}^- = \bm{\xi}^- \bullet \vec{Q}^{(\ell)},
\qquad
\vec{Q}^{(\ell)} 
:= \begin{bmatrix}
    (\vec{Q}^{(\ell)}_{1})^\T &
    \cdots &
    (\vec{Q}^{(\ell)}_{d})^\T
\end{bmatrix}^\T
\]
By setting $\tright$ large, the expected squared error for each block can be driven arbitrarily small.

As with \cref{alg:main}, at the final level, we do not need to sketch $\vec{A}$.
However, to limit the error when trying to perform the projections (onto the identity), we sample $\widehat{\bm{\xi}}\sim\operatorname{CountSketch}(d,\tright)$ and use the sketching matrix 
\[
\widehat{\vec{\Psi}} = \widehat{\bm{\xi}} \bullet (\vec{1}\otimes\vec{I}).
\]
Here $\vec{1}$ is the all-ones vector and ``\:$\otimes$\:'' is denotes the Kronecker product so that $\vec{1}\otimes\vec{I}$ is the stacked identity matrix.

\begin{algorithm}[ht]
\caption{Randomized SVD Peeling algorithm for HODLR approximation}\label{alg:RSVD_peel}
\fontsize{11}{15}\selectfont
\begin{algorithmic}[1]
\Procedure{RandomizedSVDPeeling}{$\vec{A},k,\sleft,\tleft,\sright$}
\State Set $L = \lceil \log_2(n/k) \rceil$
\Comment{Final level blocks of size at most $k$}
\For{$\ell=1,2,\ldots, L$}
\State Allocate and partition $\vec{H}^{(\ell)}$ as in \cref{eqn:Hell_decomp} \Comment{blocks of size $n/2^\ell\times n/2^\ell$}
\State Sample $\bm{\xi}^+,\bm{\xi}^-,\vec{\Omega}^+,\vec{\Omega}^- \sim \operatorname{RandPerfGaussian}(n,2^\ell,\sright,\tright)$ \Comment{as in \cref{def:RandPerfGaussian}}
\State Compute $\vec{A}^{(\ell)}\vec{\Omega}^\pm$
\Comment{$2\sright\sright$ matvecs with $\vec{A}$}

\For{$j=1,2,\ldots, 2^\ell$}
\State $\vec{Q}^{(\ell)}_j = \operatorname{orth}(\vec{Y}^{(\ell)}_j)$  \Comment{$\vec{Y}^{(\ell)}_j$ is $(\jp,\rho)$ block of $\vec{A}^{(\ell)}\vec{\Omega}^\pm$}

\EndFor
\State Sample $\bm{\zeta}^+,\bm{\zeta}^-\sim\operatorname{PerfCountSketch}(d,\tleft)$ \Comment{as in \cref{def:perforatedCS}}
\State Set $
\vec{\Psi}^\pm = \bm{\xi}^\pm \bullet  \vec{Q}^{(\ell)}$
\State Compute $(\vec{\Psi}^\pm)^\T\vec{A}^{(\ell)}$
\Comment{$2\sleft\tleft$ matvecs with $\vec{A}^\T$}
\For{$j=1,2,\ldots, 2^\ell$}
\State Extract $\vec{X}^{(\ell)}_j$  \Comment{$\vec{X}^{(\ell)}_j$ is $(j,\sigma)$ block of $(\vec{\Psi}^\pm)^\T \vec{A}$}
\State $\vec{H}^{(\ell)}_j = \vec{Q}^{(\ell)}_j\llbracket\vec{X}^{(\ell)}_j\rrbracket_k$  \Comment{$\vec{H}^{(\ell)}_j$ is $(\jp,j)$-th block of $\vec{H}^{(\ell)}$}
\EndFor
\EndFor
\State Allocate and partition $\widehat{\vec{H}}$ as in \cref{eqn:Hhat_def} \Comment{blocks of size $n/2^L\times n/2^L$}
\State Sample $\widehat{\bm{\zeta}}\sim\operatorname{CountSketch}(2^\ell,\tleft)$\Comment{as in \cref{def:CS}}
\State Set $\widehat{\vec{\Psi}} = \widehat{\bm{\zeta}} \bullet ([1,\ldots, 1]^\T\otimes \vec{I})$
\State $\widehat{\vec{\Psi}}^\T \widehat{\vec{A}}^{(L)} = \vec{A}^\T \widehat{\vec{\Psi}} - (\vec{H}^{(L)} + \cdots + \vec{H}^{(1)})\vec{\Psi}$ \Comment{$\sleft\tleft$ matvecs with $\vec{A}^\T$}
\For{$j=1,2,\ldots, 2^L$} 
\State Extract $\widehat{\vec{X}}_j$  \Comment{$\widehat{\vec{X}}_j$ is $(j,\sigma)$ block of $\vec{\Psi}^\T \widehat{\vec{A}}^{(L)}$}
\State $\widehat{\vec{H}}_j = \widehat{\vec{X}}_j$  \Comment{$\widehat{\vec{H}}_j$ is $(j,j)$-th block of $\widehat{\vec{H}}$}
\EndFor
\State \Return $\widetilde{\vec{A}} = \vec{H}^{(1)} + \cdots + \vec{H}^{(L)} + \widehat{\vec{H}}$
\EndProcedure
\end{algorithmic}
\end{algorithm}

\begin{theorem}\label{thm:main_full_RSVD}
Fix a rank parameter $k$ and let $\beta\in(0,1)$. Suppose $\tleft\geq 1$ and $\sright$, $\tleft$, and $\sright$ are such that
\begin{align*}
\frac{k}{\sright-k-1}\leq \frac{\beta}{10}
,\qquad
\frac{k}{\sright-k-1}\cdot \frac{1}{\tright} \leq \frac{\beta^2}{10^2}
,\qquad 
\frac{1}{\tleft} \leq \frac{\beta^2}{10^2}.
\end{align*}
Let $L=\lceil \log_2(n/k) \rceil$.
Then, using $2L\sright\tright$ products with $\vec{A}$ and $(2L+1)\sleft\tleft$ products with $\vec{A}^\T$, \cref{alg:RSVD_peel} outputs a $\HODLR(k)$ matrix $\widetilde{\vec{A}}$ such that
    \[
    \mathbb{E}\Bigl[\| \vec{A} - \widetilde{\vec{A}} \|_\F^2\Bigr] \leq (1+\beta)^{L+1} \cdot \min_{\vec{H}\in\HODLR(k)} \|\vec{A} - \vec{H} \|_\F^2.
    \]
\end{theorem}

\begin{proof}
The proof is analagous to \cref{thm:main_full}.

First, note that in \cref{thm:GN_perturb_exp} we change the definition of $\vec{X}$ to
\[
\vec{X} := \vec{Q}^\T\vec{B} + \vec{E},
\]
where $\vec{E}\in\mathbb{R}^{\sleft\times m_2}$.
The resulting bound now has
\[
E_2 := \frac{2k}{\sright-k-1} \| \vec{M} \|_\F^2 + 2 \|\vec{E}\|_\F^2.
\]

Next, in \cref{thm:level_l_gurantee} we now plug in $\vec{E} := (\vec{Q}^{(\ell)})^\T \vec{N}$ and note that $\|\vec{E}\|_\F^2 \leq \|\vec{N}\|_\F^2$.
The bound \cref{eqn:bound_E1E2} therefore becomes 
\begin{equation}
    \label{eqn:bound_E1E2_RSVD}
    \EE\Bigl[
    \|\vec{A}^{(\ell)}_{\jp,j} - \vec{Q}^{(\ell)}_j \llbracket \vec{X}^{(\ell)}_j \rrbracket_k \|_\F^2 \Big| \mathcal{F}_{\ell},\bm{\xi},\bm{\zeta} \Bigr]
    \leq E_1 + E_2 + 2\sqrt{E_1E_2},
\end{equation}
where $E_1$ is as before and
\begin{align*}
    E_2 &:= \frac{2k}{\sright-k-1} \|\vec{M}\|_\F^2 + 2 \|\vec{N}\|_{\F}^2    .
\end{align*}
Therefore, we get
\begin{align*}
    E_2' &:= \frac{2k}{\sright-k-1}\cdot \frac{1}{\tright} \sum_{\substack{i=1\\i\neq j,\jp}}^{d} \|\vec{A}^{(\ell)}_{\jp,i}\|_\F^2 +  \frac{2}{\tleft}\sum_{\substack{i=1\\i\neq j,\jp}}^{d} \|\vec{A}^{(\ell)}_{i,j}\|_\F^2 .
\end{align*}
and our assumptions on $\tleft$, $\tright$, and $\sleft$ allow us to bound
\begin{align*}
    \frac{k}{\sright-k-1} \| \vec{A}^{(\ell)}_{\jp,j} - \llbracket \vec{A}^{(\ell)}_{\jp,j} \rrbracket_k \|_\F^2
    &\leq \frac{\beta}{10} E_j^{(\ell)},
\end{align*}
\begin{align*}
    \frac{2k}{\sright-k-1}\cdot \frac{1}{\tright} \sum_{\substack{i=1\\i\neq j,\jp}}^{d} \|\vec{A}^{(\ell)}_{\jp,i}\|_\F^2 
    &\leq 4 \frac{\beta^2}{10^2}E_j^{(\ell)},
    &    \frac{2}{\tleft}\sum_{\substack{i=1\\i\neq j,\jp}}^{d} \|\vec{A}^{(\ell)}_{i,j}\|_\F^2 
    &\leq 4 \frac{\beta^2}{10^2}E_j^{(\ell)}.
\end{align*}
Finally, since $1/10 + (4+4)/10^2 + 2\sqrt{2(4+4)/10^2}\approx 0.980 \leq 1$ we get the desired bound.
\end{proof}

\printbibliography[]

\end{document}